\documentclass[11pt,a4paper]{amsart}

\usepackage[utf8]{inputenc}
\usepackage[english]{babel}
\usepackage{graphicx}	
\RequirePackage[OT1]{fontenc}
\RequirePackage{amsthm,amsmath}
\RequirePackage[colorlinks,citecolor=blue,urlcolor=blue]{hyperref}
\RequirePackage{cleveref}
\RequirePackage[numbers]{natbib}
\usepackage{fullpage}

\usepackage{tikz}
\usetikzlibrary{graphs, arrows.meta, positioning}

\newtheorem{theorem}{Theorem}
\newtheorem{lemma}[theorem]{Lemma}
\newtheorem{proposition}[theorem]{Proposition}

\newtheorem{remark}[theorem]{Remark}

\newcommand{\R}{\mathbb{R}}
\newcommand{\C}{\mathbb{C}}
\newcommand{\E}{\mathbb{E}}
\renewcommand{\P}{\mathbb{P}}

\renewcommand{\O}{\mathcal{O}}
\renewcommand{\Re}{\mathrm{Re} \,}
\renewcommand{\Im}{\mathrm{Im} \,}
\renewcommand{\epsilon}{\varepsilon}

\def\:{\mathop{:}}
\newcommand{\gff}{\mathrm{GFF}}
\newcommand{\sg}{\mathrm{sG}}

\newcommand{\gffcf}[1]{\left\langle#1\right\rangle_{\gff}}

\newcommand{\sgcf}[1]{\left\langle#1\right\rangle_{\sg(\beta,\mu,\psi)}}

\numberwithin{equation}{section}
\numberwithin{theorem}{section}
\numberwithin{figure}{section}
	
\title{Operator product expansions of derivative fields in the sine-Gordon model}

\author{Alex Karrila \and Tuomas Virtanen \and Christian Webb}

\begin{document}

\begin{abstract}
In this article, we initiate the study of operator product expansions (OPEs) for the sine-Gordon model. For simplicity, we focus on the model below the first threshold of collapse ($\beta<4\pi$) and on the singular terms in OPEs of derivative-type fields $\partial \varphi$ and $\bar\partial\varphi$. We prove that compared to corresponding free field OPEs, the sine-Gordon OPEs develop logarithmic singularities and generate Wick ordered exponentials. Our approach for proving the OPEs relies heavily on Onsager-type inequalities and associated moment bounds for GFF correlation functions involving Wick ordered exponentials of the free field. 
\end{abstract}

\maketitle

{
  \hypersetup{linkcolor=black}
  \tableofcontents
}

\section{Introduction and main results}\label{sec:intro}

In this introductory section, we first review the notion of operator product expansions, give an informal definition of the sine-Gordon model (with a precise definition in Section \ref{sec:sG}), and then state Theorem \ref{th:main}, which is our main result about operator product expansions for derivative fields in the sine-Gordon model.

\subsection{The operator product expansion} In the theoretical physics literature, the \emph{operator product expansion} (OPE) is a profound tool for studying quantum field theories (QFTs). The OPE was originally introduced by Wilson \cite{Wilson}. An OPE describes what happens to two local quantum fields evaluated at points that are allowed to merge. It is typically postulated in theoretical physics that for any QFT that has a meaningful continuum definition, for $x$ in some neighborhood of $y$, the product of any two local fields $A(x)$ and $B(y)$ has an expansion of the form 
\begin{equation}\label{eq:ope}
A(x)B(y)\sim \sum_{i}c_i(x-y)C_i(y), 
\end{equation}
where $C_i$ are some further local fields, $c_i(x-y)$ are $\C$-valued functions that may diverge at zero, and $\sim$ indicates that this may be an asymptotic expansion or that the difference of the two sides is regular at $x=y$. The existence of the OPE \eqref{eq:ope} can be viewed either as a useful tool for defining possibly more complicated fields $C_i$ from the possibly simpler fields $A,B$, or then as a fundamental axiom of QFT needed to define interacting QFTs. For an overview of OPEs from the point of view of physics, see e.g. \cite{HW,Wilson}.

OPEs have been particularly important in studying \emph{conformal field theories} (CFTs), where OPEs are a fundamental building block of the \emph{conformal bootstrap} approach to CFT  -- see e.g. \cite[Section 6.6.5]{DFMS} and \cite{PRV}. Mathematically, this approach is closely related to that of \emph{vertex operator algebras}, which are an algebraic approach to defining $1+1$ dimensional CFTs \cite{Kac}.

In the past few decades, a major theme in modern probability theory has been understanding conjectures of physicists about describing scaling limits of critical models of statistical mechanics in terms of CFT. Part of this question is to understand concepts such as OPEs and the conformal bootstrap approach in probabilistic models defined for example as scaling limits of lattice models. While a general, mathematically rigorous solution to such a question seems out of reach currently, there has been progress in some particular models. For example, for the \emph{Liouville model}, the leading order terms for the OPE were studied in \cite{BW}, and the conformal bootstrap approach was justified in \cite{GKRV}. For the scaling limit of the \emph{critical Ising model}, the leading order terms of the OPE were studied in \cite[Section 6]{CHI}. The OPE for the scaling limit of the critical Ising model as well as that for the (compactified) massless scalar field played also an important role for proving \emph{bosonization identities} relating Ising correlation functions to free field correlation functions in \cite{BIVW}. Recently, a lattice-level version of the idea of generating more complex fields from simpler ones via OPEs (in a boundary CFT) was mathematically formalized in \cite{KLPR}.

\subsection{The sine-Gordon model}

The (Euclidean, massless) \emph{sine-Gordon model} is a Euclidean quantum field theory, that physicists might define through a path integral with a measure 
\begin{equation}\label{eq:sGformal}
\frac{1}{Z_{\beta,\mu}}e^{-\frac{1}{2}\int |\nabla \varphi|^2+\mu \int \:\cos(\sqrt{\beta}\varphi)\:}\mathcal D\varphi,
\end{equation}
where $Z_{\beta,\mu}$ is a (possibly infinite) normalization constant known as the partition function, $\mu\in \R$ and $\beta>0$, the integrals are over $\R^2$,  $\:\cdot\:$ denotes \emph{Wick ordering} (we return to a precise definition of this in Section \ref{sec:wick}), and $\mathcal D\varphi$ denotes the non-existent Lebesgue measure on the space of functions from $\R^2$ to $\R$. Before turning to a mathematically precise defintion of this model (which we present in detail in Section \ref{sec:sG} for $\beta\in(0,4\pi)$), we mention some of the reasons that the sine-Gordon model is an important model in the study of two-dimensional QFTs. 
\begin{itemize}
\item On the level of the partition function and suitable correlation functions, the sine-Gordon model is equivalent to a model of classical statistical mechanics describing (in the grand canonical formalism) a gas of charged particles interacting through the two-dimensional Coulomb potential. See e.g. \cite{Frohlich} for a mathematically rigorous discussion of this.
\item The sine-Gordon model is a prime example of a \emph{near critical model}. It is not a CFT, but a rather mild perturbation of one (namely the free scalar field). In fact, physicists predict that despite losing the structure associated with conformal invariance, the model retains some type of integrability, and certain quantities can be computed exactly. See e.g. \cite{LZ,Mussardo}.
\item The sine-Gordon model is a prime example of an interacting quantum field theory exhibiting a duality known as \emph{bosonization}. The dual theory is a fermionic theory known as the \emph{massive Thirring model}. This was first predicted in the physics literature by Coleman \cite{Coleman}. For mathematical results of this flavor, see \cite{BaWe,BFM}.
\item The sine-Gordon model is conjectured to be the scaling limit of a wide class of models of statistical mechanics, such as the \emph{six-vertex model}. See e.g. \cite{Lashkevich} for conjectures from the physics literature and \cite{Mason} for mathematical work relating the scaling limit of the \emph{dimer model} and the sine-Gordon model.
\item The sine-Gordon model is an interesting model from the point of view of \emph{renormalization}. One expects that the sine-Gordon model is non-trivial only in the regime $\beta\in(0,8\pi)$ and that there exists an increasing sequence of thresholds $\beta=\beta_n$ with $\lim_{n\to\infty}\beta_n=8\pi$ such that more and more complicated renormalization (counter terms) is needed to define the model for $\beta\in[\beta_n,\beta_{n+1})$. We will focus on the case $\beta<\beta_1=4\pi$, where Wick ordering is the only renormalization that is needed. See for example \cite{DH,Frohlich,NRS} for further details on this. 
\end{itemize}

As is typical in a mathematical approach to Euclidean QFT (see e.g. \cite{GJ}), the way of making sense of the path integral \eqref{eq:sGformal} physicists propose for the sine-Gordon model is to interpret $e^{-\frac{1}{2}\int |\nabla \varphi|^2}\mathcal D\varphi= \P(d\varphi)$ as a centered Gaussian measure with covariance $-\Delta^{-1}$, and try to construct the interacting measure $\frac{1}{Z_{\beta,\mu}}e^{\mu \int \:\cos(\sqrt{\beta}\varphi)\:}\P(d\varphi)$ through regularization and renormalization. The usual difficulty here is that the Gaussian measure is ill defined since the Laplacian is not invertible in the full plane. There are various ways to try to circumvent this issue, but our choice will be to focus on the \emph{finite volume} sine-Gordon model. That is, we fix some bounded open set $\Omega$ (with a nice boundary), a function $\psi\in C_c^\infty(\Omega)$, and try to view the finite volume sine-Gordon model as the $\epsilon\to 0$ limit of the measures
\begin{equation}\label{eq:sGreg}
\frac{1}{Z_{\beta,\mu,\psi}(\epsilon)}e^{\mu \int_{\R^2}d^2x\, \psi(x)\epsilon^{-\frac{\beta}{4\pi}}\cos(\sqrt{\beta}\varphi_\epsilon(x))}\P(d\varphi_\epsilon),
\end{equation}
where $\varphi_\epsilon$ is a regularization of the \emph{Gaussian free field} (GFF) on $\Omega$ with zero Dirichlet boundary conditions on $\partial \Omega$, and $Z_{\beta,\mu,\psi}(\epsilon)$ is a normalization constant. We will return to this regularization in detail in Section \ref{sec:wick}, but for now we mention that it is such that almost surely $\varphi_\epsilon$ is a continuous function on $\Omega$, so for each $\epsilon>0$, the measure above is well defined.\footnote{Strictly speaking, we find it convenient to define the measure \eqref{eq:sGreg} in a slightly different way, where $\mu$ is replaced by a slightly different constant. This has the benefit of yielding an object that does not depend on the choice of our regularization. See \eqref{eq:wickdef} and Section \ref{sec:sG}.}

The main fields whose correlation functions are of interest (from various points of view: the Coulomb gas picture, the model being a perturbation of a CFT, bosonization) are so-called \emph{charge or vertex fields} $\:e^{i\alpha \varphi(x)}\:$ for $\alpha\in \R$, and \emph{current or derivative\footnote{Here $\partial$ and $\bar\partial$ are the holomorphic and antiholomorphic derivatives: for $x=(x_1,x_2)\in \R^2$, $\partial f(x)=\frac{1}{2}(\frac{\partial}{\partial x_1}-i\frac{\partial}{\partial x_2})f(x)$ and $\bar\partial f(x)=\frac{1}{2}(\frac{\partial}{\partial x_1}+i\frac{\partial}{\partial x_2})f(x)$.} fields} $\partial \varphi(x)$ and $\bar\partial\varphi(x)$ (the charge/current terminology comes from the Coulomb gas picture).  Thus the main objects of interest in the model are the following \emph{correlation functions}
\begin{align}\label{eq:sgcf}
&\sgcf{\prod_{j=1}^p \partial \varphi(u_j)\prod_{k=1}^q \bar\partial \varphi(v_k) \prod_{l=1}^r \:e^{i\alpha_l\varphi(x_l)}\:}\\
&\qquad =\lim_{\epsilon\to 0}\frac{1}{Z_{\beta,\mu,\psi}(\epsilon)}\E\bigg[\prod_{j=1}^p \partial \varphi_\epsilon(u_j)\prod_{k=1}^q \bar\partial \varphi_\epsilon(v_k) \prod_{l=1}^r \epsilon^{-\frac{\alpha_l^2}{4\pi}}e^{i\alpha_l\varphi_\epsilon(x_l)}\notag  e^{\mu \int_{\R^2}d^2x\, \psi(x)\epsilon^{-\frac{\beta}{4\pi}}\cos(\sqrt{\beta}\varphi_\epsilon(x))}\bigg].\notag 
\end{align}
Above, $\E[\cdot]$ denotes expectation with respect to the law of the regularization of the GFF, $\varphi_\epsilon$.

Ideally, the way one would study the OPEs \eqref{eq:ope} in the sine-Gordon model would be to construct these correlation functions, prove that they are continuous functions when none of the points involved coincide, and then one would study asymptotics of these correlation functions as two (or more) points merge. As we discuss shortly, we will take a slightly softer approach that avoids some technical difficulties, but still gives precise meaning to the OPE.

Ultimately, one would like to take the \emph{infinite volume limit}, $\Omega\to \R^2$, $\psi\to 1$, but this is a rather difficult question. Moreover, as one expects that the OPEs, which are the main topic of this article, are local notions which should not depend heavily on whether we are in finite or infinite volume, we choose to focus on the finite volume model.

As suggested above, there are various mathematically rigorous ways of studying the sine-Gordon model. The mathematical study of the model was initiated in \cite{Frohlich}. For other approaches, see for example \cite{Barashkov,BaWe,DH,NRS}.

\subsection{The main result} Given the fundamental importance of the concept of the OPE as well as the sine-Gordon model, it is natural to ask what is known about OPEs in the sine-Gordon model. OPEs feature in various way in studies of the sine-Gordon model in the physics literature, but as an example, let us mention the work \cite{Kehrein}, where the OPE for the vertex fields of the sine-Gordon model plays a role in diagonalizing the Hamiltonian operator of the sine-Gordon model (in the Minkowski space/Hilbert space representation of the model). Mathematically there is not much that is known, and the purpose of this article is to initiate the study of OPEs for the sine-Gordon model. 

The first question in the study of the OPEs for the sine-Gordon model is how to interpret \eqref{eq:ope} for such a probabilistic model. The typical way this is done in the physics and mathematics literature (see \cite{BIVW,CHI,KM}) is to view \eqref{eq:ope} as holding within an arbitrary correlation function. That is, for $\O$ in a suitable class \emph{spectator observables}, we should have  
\begin{equation}\label{eq:opecorr}
\sgcf{A(x)B(y)\O}\sim \sum_{i}c_i(x-y)\sgcf{C_i(y)\O},
\end{equation}
where $A,B$ are either derivative fields $\partial \varphi, \bar\partial \varphi$, or vertex fields $\:e^{i\alpha\varphi}\:$, and $\sim$ indicates that the difference between the left and right hand sides has a limit as $x\to y$ (or possibly even being an asymptotic expansion). The most natural choice for the class of spectator observables, and the one typically studied in the physics literature, is to take $\O$ to be a product of fields evaluated at distinct points, e.g. of the form 
\begin{equation*}
\prod_{j=1}^p \partial \varphi(u_j)\prod_{k=1}^q \bar\partial \varphi(v_k) \prod_{l=1}^r \:e^{i\alpha_l\varphi(x_l)}\:.
\end{equation*}
The interpretation is that if \eqref{eq:opecorr} holds for any such $\O$, then \eqref{eq:ope} holds.

As we are taking the first step into rigorous analysis of OPEs for the sine-Gordon model, we choose to simplify the question of analyzing of the OPEs in various ways.
\begin{itemize}
\item First of all, we consider the sine-Gordon model only for $\beta<4\pi$. In this regime, the sine-Gordon measure is absolutely continuous with respect to the measure of the GFF, which makes analysis significantly simpler.
\item The second simplification we make is that we focus only on the case where $A$ and $B$ are derivative fields. It seems that in our approach to the OPEs, the vertex fields are more challenging to control. To be more precise, certain high dimensional integral operators we encounter in Section \ref{sec:wick} are more singular for the vertex fields and they are harder to control accurately.
\item  Our third simplification is that instead considering spectator observables that are products of fields evaluated at distinct points, we take $\O$ to be in a class of random variables that generate the $\sigma$-algebra of the underlying probability space. More precisely, we take $\O$ to be an \emph{exponential functional in the field}, namely 
\begin{equation*}
\O=e^{i\varphi(f)}
\end{equation*}
for some $f\in C_c^\infty(\Omega)$. At least formally, one could generate all possible correlation functions \eqref{eq:sgcf} from these by taking suitable limits. For example, if one would take $f(x)=\sum_{j=1}^n \alpha_j \delta_\epsilon(x-x_j)$, where $\delta_\epsilon$ is a suitable approximation of the Dirac delta-function of width $\epsilon$, then multiplying $\O$ by a suitable $\epsilon$–dependent quantity and taking $\epsilon\to 0$ would yield vertex fields in the $\epsilon\to 0$ limit. Derivative fields could in principle be generated by similar limiting arguments and series expanding the exponential for suitable choices of $f$. Controlling such limits in the OPE does not seem particularly simple though.

From the point of probability theory, this is a natural interpreation for our interpretation of the OPE since if for two integrable random variables $X$ and $Y$ we have $\E(X\O)=\E(Y\O)$ for all $\O$ in some class of random variables that generates the full $\sigma$-algebra of the probability space, we have $X=Y$ almost surely.

 On a technical level, this simplification bypasses the need to control the correlation functions \eqref{eq:sgcf} in general, and again this simplifies the analysis of certain singular integrals we encounter. 
\item Finally, we consider only the singular part of the OPE \eqref{eq:opecorr}, namely focusing on terms in \eqref{eq:opecorr} where $c_i$ does not have a limit as $x\to y$, instead of trying to prove an all order expansion. 
\end{itemize}
We expect that with significantly more involved analysis of the integrals appearing in Section \ref{sec:wick}, most of these simplifications can be dispensed with, at least for $\beta<4\pi$ -- the extension to $\beta\geq 4\pi$ requiring a slightly different approach. Fully understanding OPEs for the sine-Gordon model (involving also vertex fields, understanding the expansions to all order and all values of $\beta$) would be a very interesting question that we leave for further study.

Our main result is as follows.

\begin{theorem}\label{th:main} 
 For $\Omega\subset \R^2$ bounded, connected, and open with smooth boundary, $\beta\in (0,4\pi)$, $\mu\in \R$, $\psi\in C_c^\infty(\Omega)$,  and $\O$ an exponential functional in the field $\varphi$, namely $\O=e^{i\varphi(f)}$ with $f\in C_c^\infty(\Omega)$, the following claims hold.
\begin{enumerate}
\item For $x,y\in \Omega$ distinct, the correlation functions 
\begin{equation*}
\sgcf{\partial \varphi(x)\partial \varphi(y)\O}, \qquad \sgcf{\bar\partial \varphi(x)\bar\partial \varphi(y)\O}, \qquad \sgcf{\partial \varphi(x)\bar\partial \varphi(y)\O}
\end{equation*}
exist and are continuous functions on the set $\{(x,y)\in \Omega^{2}: x\neq y\}$. Also, the correlation functions
\begin{equation*}
\sgcf{\:e^{\pm i\sqrt{\beta}\varphi(x)}\:\O}
\end{equation*}
exist and are continuous in $x\in \Omega$.
\item The following limits exist (and are finite) for $y\in \Omega$
\begin{align}\label{eq:OPE1}
\lim_{x\to y}\left(\sgcf{\partial\varphi(x)\partial\varphi(y)\O}+\frac{1}{4\pi}\frac{1}{(x-y)^2}\sgcf{\O}-\mu \frac{\beta}{16\pi}\frac{(\bar x-\bar y)}{x-y}\psi(y)\sgcf{\:\cos (\sqrt{\beta}\varphi(y))\:\O}\right),
\end{align}
\begin{align}\label{eq:OPE2}
\lim_{x\to y}\left(\sgcf{\bar\partial\varphi(x)\bar\partial\varphi(y)\O}+\frac{1}{4\pi}\frac{1}{(\bar x-\bar y)^2}\sgcf{\O}-\mu \frac{\beta}{16\pi}\frac{(x-y)}{\bar x-\bar y}\psi(y)\sgcf{\:\cos (\sqrt{\beta}\varphi(y))\:\O}\right),
\end{align}
\begin{align}\label{eq:OPE3}
\lim_{x\to y}\left(\sgcf{\partial\varphi(x)\bar \partial\varphi(y)\O}+\mu \frac{\beta}{8\pi}\log |x-y|^{-1}\psi(y)\sgcf{\:\cos (\sqrt{\beta}\varphi(y))\:\O}\right),
\end{align}
where we have written 
\begin{align*}
\sgcf{\:\cos (\sqrt{\beta}\varphi(y))\:\O}=\frac{1}{2}\left(\sgcf{\:e^{i\sqrt{\beta}\varphi(y)}\:\O}+\sgcf{\:e^{-i\sqrt{\beta}\varphi(y)}\:\O}\right).
\end{align*}
\end{enumerate}
\end{theorem}
\begin{remark}\label{rem:main}
We comment here on a few issues related to Theorem \ref{th:main}. 
\begin{enumerate}
\item This result can be seen as a way of defining a renormalized pointwise product of the derivatives. Note for the GFF, one would only need to Wick order the product, but for the sine-Gordon model, this renormalization is more complicated.
\item Closely related correlation functions were recently constructed in \cite{FC}, where the authors prove that after subtracting a suitable $\epsilon$-dependent deterministic constant (which is divergent in the $\epsilon\to 0$ limit), correlation functions involving $\partial_\mu\varphi_\epsilon(x)\partial_\nu\varphi_\epsilon(x)$ (where $\mu,\nu$ index coordinates of the spatial variable $x$) converge in the $\epsilon\to 0$ limit. The connection to the OPE is not discussed though.
\item It is instructive to compare these OPEs to the corresponding GFF ones. We discuss these in Section \ref{sec:GFFope}, but the upshot is that the singular quantities in the OPEs \eqref{eq:OPE1}--\eqref{eq:OPE3} are new compared to the GFF OPEs. Moreover, the GFF OPEs never generate exponentials from the derivatives. This illustrates how OPEs in near critical models/interacting models can differ from OPEs in CFTs. 
\item Comparing with \eqref{eq:ope}, one should interpret for example \eqref{eq:OPE1} as saying that 
\begin{equation*}
\partial \varphi(x)\partial\varphi(y)=-\frac{1}{4\pi}\frac{1}{(x-y)^2}+\mu \frac{\beta}{16\pi}\frac{\bar x-\bar y}{x-y} \psi(y) \:\cos (\sqrt{\beta}\varphi(y))\:+\, \mathrm{regular}.
\end{equation*}
Perhaps a more thorough justification of this would allow replacing the spectator observables $\O$ by all possible fields obtained from OPEs of the gradient and vertex fields, and we expect that the result is true in this generality, but it would require a full understanding of all order OPEs for all possible fields. 
\item In \cite{BH}, a variant of the sine-Gordon model is studied through a coupling between the GFF and the sine-Gordon field. The authors prove that for $\beta<6\pi$, one can construct a probability space, where the two fields differ by a Hölder-continuous function (continuously differentiable for $\beta<4\pi$). It would be interesting to understand the OPEs discussed here from this coupling perspective as well.
\item Using the bosonization results of \cite{BW}, one should be able to study gradient OPEs also for the infinite volume limit of the sine-Gordon model at $\beta=4\pi$. 
\item It would be interesting to study similar questions for a closely related model known as the sinh-Gordon model, whose definition is similar, but $\:\cos(\sqrt{\beta}\varphi)\:$ is replaced by $\:\cosh(\gamma \varphi)\:$. This model has recently been studied e.g. in \cite{BOW,GGV,HZ}.
\end{enumerate}
\end{remark}

The proof of Theorem \ref{th:main} is split into two parts: in Section \ref{sec:sG}, we construct the relevant correlation functions, while in Section \ref{sec:sGope}, we prove the OPE. Our proof is based on proving that the correlation functions are analytic in $\mu$ (using so-called Onsager-inequalities and improvements of associated moment bounds studied in \cite{GP,JSW}), with Taylor coefficients given by suitable GFF correlation functions, which we analyze in Section \ref{sec:wick}. We review some basic facts about the GFF as well as an analogue of Theorem \ref{th:main} for the GFF in Section \ref{sec:GFFope}.

\medskip

{\bf Acknowledgments: } A.K. and C.W. thank the Academy of Finland for support through the grants 339515 and 348452, respectively. T.V. thanks the Doctoral Network in Information Technologies and Mathematics at Åbo Akademi University for the financial support.

\section{A review of the GFF and an analogue of Theorem \ref{th:main} for the GFF}\label{sec:GFFope} 
In this section, we review basic properties of the \emph{Gaussian free field} (GFF), with the aim of discussing the analogue of Theorem \ref{th:main} in the setting of the GFF.   Much of what we discuss here is well known: see e.g. \cite{Sheffield} for the construction of the GFF as well as \cite[Section 2]{BIVW}, and \cite[Lecture 3]{KM} for OPEs for the GFF (with spectator fields that are products of fields at distinct points).

We begin with a brief review of what the GFF is.

\subsection{The Gaussian free field} The two-dimensional Gaussian free field, or the Euclidean massless free scalar field, is a Gaussian process $\varphi$ on $\R^2$ whose covariance is $(-\Delta)^{-1}$. This of course has some issues. First of all, $\Delta$ is not invertible on $\R^2$. To remedy this issue, we will work on a bounded open set $\Omega\subset \R^2$‚ and assume that it has a smooth boundary. In this case, $-\Delta_\Omega$, the Laplacian with zero Dirichlet boundary conditions, is invertible, and the inverse is an integral operator with  kernel $(-\Delta_\Omega)^{-1}(x,y)=G_\Omega(x,y)$, the \emph{Green's function of the domain} (with zero boundary conditions). For example, if $\Omega=B(0,1)$, then one readily checks that
\begin{equation*}
G_\Omega(x,y)=\frac{1}{2\pi}\log |x-y|^{-1}-\frac{1}{2\pi}\log |1-x\bar y|^{-1}.
\end{equation*} 
In general, 
\begin{equation}\label{eq:gomega}
G_\Omega(x,y)=\frac{1}{2\pi}\log |x-y|^{-1}+g_\Omega(x,y),
\end{equation}
where $x\mapsto g_\Omega(x,y)$ is harmonic in $\Omega$, $g_\Omega(x,y)=-\frac{1}{2\pi}\log |x-y|^{-1}$ for $x\in \partial \Omega$, and $g_\Omega(x,y)=g_\Omega(y,x)$. 

The second issue we encounter is that $G_\Omega(x,x)=\infty$, so $\E(\varphi(x)^2)=\infty$, which means that $\varphi$ cannot be a random function with a Gaussian distribution (what is a Gaussian random variable with infinite variance?). This can be remedied by viewing $\varphi$ as a \emph{random generalized function} -- e.g. as a random element of $\mathcal D'(\Omega)$ or a suitable Sobolev space with a negative regularity index. The upshot is that there exists a probability space on which we have a random generalized function $\varphi$ such that for any collection $f_1,...,f_n\in C_c^\infty(\Omega)$, the random variables $\varphi(f_1),...,\varphi(f_n)$ are jointly Gaussian, centered, and have covariance
\begin{equation}\label{eq:gffcov}
\gffcf{\varphi(f_i)\varphi(f_j)}:=\E[\varphi(f_i)\varphi(f_j)]=\int_{\Omega\times \Omega}d^2x\, d^2y\, f_i(x)f_j(y)G_\Omega(x,y).
\end{equation}
In what follows, we will often use $\gffcf{\cdot}$ instead of $\E$ for expectation with respect to the probability measure on this probability space. When relevant, we write $\P$ for the probability measure on this probability space. For a proof of existence of the GFF, see e.g. \cite{Sheffield}. The reader can also give a quick proof of existence on a suitable space of generalized functions using Minlos' theorem.

Note that since $\varphi$ is a random generalized function, also its derivatives are well defined random generalized functions: for $f\in C_c^\infty(\Omega)$, we have for example $\partial \varphi(f):=-\varphi(\partial f)$.

Since we are dealing with jointly Gaussian random variables, we can compute joint moments of random variables of the form $\varphi(f)$ (or $\partial \varphi(f)$ and $\bar\partial \varphi(f)$) by Wick's theorem (also known as Isserlis' theorem): for $f_1,...,f_p,g_1,...,g_q,h_1,...,h_r\in C_c^\infty(\Omega)$,
\begin{align}\label{eq:gffcorr}
&\gffcf{\prod_{j=1}^p \partial \varphi(f_j)\prod_{k=1}^q \bar\partial\varphi(g_k)\prod_{l=1}^r \varphi(h_l)}\\
&\quad=\sum_{\Pi \in P_{p+q+r}}\int_{\Omega^{p+q+r}}d^{2(p+q+r)}x \prod_{j=1}^p(-\partial f(x_j))\prod_{k=p+1}^{p+q} (-\bar \partial g(x_k))\prod_{l=p+q+1}^{p+q+r} h(x_l)\prod_{\{a,b\}\in \Pi}G_\Omega(x_a,x_b),\notag
\end{align} 
where $P_{p+q+r}$ denotes the set of pairings of $\{1,...,p+q+r\}$, that is sets of the form $\{\{a_i,b_i\}\}_{i=1}^n$, where $a_i\neq b_i$ for all $i$, $\{a_i,b_i\}\cap \{a_{i'},b_{i'}\}=\emptyset$ for $i\neq i'$, and $\bigcup_{i=1}^n \{a_i,b_i\}=\{1,...,p+q+r\}$. Note that $P_{p+q+r}$ is empty (and the corresponding moment is zero) if $p+q+r$ is odd, and if it is even, then we must have $n=\frac{p+q+r}{2}$.

More relevantly for our analogue of Theorem \ref{th:main} for the GFF, we can compute mixed moments of random variables of the form $\varphi(f)$ and $e^{i\varphi(f)}$. This can be done with a standard ``complete the square argument'', but for the convenience of the reader, we offer a proof of this in Appendix \ref{app:gauss}. Using Lemma \ref{le:girsa} (applied to the centered jointly Gaussian random variables $\varphi(g),\varphi(h),\varphi(f)$, the polynomial $P(x)=x_1x_2$, and $z=i$)  one finds that for $f,g,h\in C_c^\infty(\Omega)$
\begin{align}\label{eq:csq}
&\gffcf{\varphi(g)\varphi(h)e^{i\varphi(f)}}\\
&\quad =\gffcf{e^{i\varphi(f)}}\gffcf{(\varphi(g)+i\gffcf{\varphi(g)\varphi(f)})(\varphi(h)+i\gffcf{\varphi(h)\varphi(f)})}\notag \\
&\quad =\gffcf{e^{i\varphi(f)}} \bigg(\int_{\Omega\times \Omega} d^2 x d^2 y g(x)h(y)G_\Omega(x,y)\notag \\
&\qquad -\int_{\Omega\times \Omega} d^2 x_1 d^2 y_1 f(x_1)g(y_1)G_\Omega(x_1,y_1)\int_{\Omega\times \Omega} d^2 x_2 d^2 y_2 f(x_2)h(y_2)G_\Omega(x_2,y_2)\bigg).\notag
\end{align}

\subsection{An analogue of Theorem \ref{th:main} for the GFF}

In this section, we show the analogue of Theorem \ref{th:main} for the GFF. The main purpose of this is to emphasize that for the sine-Gordon model, the OPE is genuinely different from the GFF one. This also serves the purpose of illustrating the general philosophy of how we define the relevant correlation functions also for the sine-Gordon model later on (though the task is much more involved there). We do however point out that there are other approaches to OPEs for the GFF -- see \cite{BIVW,KM}.

To formulate the result, we first need to make sense of correlation functions of the form 
\begin{equation*}
\gffcf{\partial \varphi(x)\partial \varphi(y)\O}, \qquad \gffcf{\bar\partial \varphi(x)\bar\partial \varphi(y)\O}, \qquad \text{and} \qquad \gffcf{\bar\partial \varphi(x)\partial \varphi(y)\O},
\end{equation*}
with $\O$ an exponential functional in the field: $\O=e^{i\varphi(f)}$ with $f\in C_c^\infty(\Omega)$. Note that $\partial\varphi(x)$ and $\bar\partial \varphi(y)$ are no longer honest random variables -- $\partial \varphi$ is a random generalized function, so it does not have a pointwise meaning. We define these correlation functions through the following simple lemma. 
\begin{lemma}\label{le:gffpw}
For a given exponential functional in the field $\O=e^{i\varphi(f)}$ with $f\in C_c^\infty(\Omega)$, there exists a unique smooth function on $\{(x,y)\in\Omega\times \Omega: x\neq y\}$, which we write as $(x,y)\mapsto \gffcf{\partial\varphi(x)\partial \varphi(y)\O}$, such that for any $g,h\in C_c^\infty(\Omega)$ with disjoint supports 
\begin{align*}
\gffcf{\partial \varphi(g)\partial \varphi(h)\O}=\int_{\Omega\times \Omega}d^2 x d^2 yg(x)h(y)\gffcf{\partial\varphi(x)\partial \varphi(y)\O}.
\end{align*}
Similar claims are true for $\gffcf{\bar\partial \varphi(g)\bar\partial \varphi(h)\O}$ and $\gffcf{\bar\partial \varphi(g)\partial \varphi(h)\O}$.
\end{lemma}
Note that the assumption about $g$ and $h$ having disjoint supports is important here -- apart from the $\bar\partial \varphi(x)\partial \varphi(y)$-case, these functions will not be continuous on the diagonal $x=y$. 
\begin{proof}
We begin by noting that from \eqref{eq:csq} (replacing $g$ by $-\partial g$ and $h$ by $-\partial h$ and integrating by parts), we have at least formally 
\begin{align*}
\gffcf{\partial \varphi(g)\partial \varphi(h)\O}&=\gffcf{\O}\int_{\Omega\times \Omega}d^2 xd^2 y g(x)h(y)\bigg(\partial_x\partial_y G_\Omega(x,y)\\
&\qquad -\int_\Omega d^2 u f(u)\partial_x G_\Omega (u,x)\int_{\Omega}d^2 v f(v)\partial_y G_\Omega(v,y)\bigg).\notag
\end{align*}
This means that our candidate for the correlation function is 
\begin{align*}
\gffcf{\partial\varphi(x)\partial \varphi(y)\O}=\gffcf{\O}\bigg(\partial_x\partial_y G_\Omega(x,y)-\int_\Omega d^2 u f(u)\partial_x G_\Omega (u,x)\int_{\Omega}d^2 v f(v)\partial_y G_\Omega(v,y)\bigg).
\end{align*}
All we need to do to prove the claim is to justify the integration by parts and show that this is a smooth function on the domain in question. Integration by parts for $\partial_x\partial_y G_\Omega(x,y)$ is justified by the disjoint support assumption on the test functions. Also, this is a  smooth function on $\{(x,y)\in \Omega^2: x\neq y\}$.

For the second term, it is sufficient for us to argue that for $f\in C_c^\infty(\Omega)$, the function 
\begin{equation*}
x\mapsto \int_\Omega d^2 u f(u)\partial_x G_\Omega(u,x)
\end{equation*}
is smooth on $\Omega$. This follows readily from \eqref{eq:gomega}. In fact, since 
\begin{equation*}
\int_\Omega d^2 u f(u)\nabla_x \log |u-x|^{-1}=\int_\Omega d^2u  \nabla f(u) \log |x-u|^{-1},
\end{equation*}
we see by iterating this identity that 
\begin{align*}
x\mapsto \int_\Omega d^2 u f(u)\partial_x G_\Omega(u,x)
\end{align*}
is a smooth function on $\Omega$.

The claims for the correlation functions with one or both of the holomorphic derivatives $\partial$ replaced by antiholomorphic ones $\bar\partial$ are proven the same way.
\end{proof}
With this fact in hand, we turn to the analogue of Theorem \ref{th:main} for the GFF.
\begin{lemma}\label{le:gffope}
For $\O=e^{i\varphi(f)}$ with $f\in C_c^\infty(\Omega)$, the following limits exist
\begin{align}\label{eq:gffOPE1}
\lim_{x\to y}\left(\gffcf{\partial\varphi(x)\partial\varphi(y)\O}+\frac{1}{4\pi}\frac{1}{(x-y)^2}\gffcf{\O}\right)
\end{align}
\begin{align}\label{eq:gffOPE2}
\lim_{x\to y}\left(\gffcf{\bar\partial\varphi(x)\bar\partial\varphi(y)\O}+\frac{1}{4\pi}\frac{1}{(\bar x-\bar y)^2}\gffcf{\O}\right) 
\end{align}
\begin{align}\label{eq:gffOPE3}
\lim_{x\to y}\gffcf{\partial\varphi(x)\bar \partial\varphi(y)\O}.
\end{align}
\end{lemma}
Again, we emphasize here the difference with Theorem \ref{th:main} -- neither singular terms like $\frac{\bar x-\bar y}{x-y}$, $\log |x-y|^{-1}$, nor objects like $\:\cos(\sqrt{\beta}\varphi)\:$ are present for the GFF. 
\begin{proof}
We saw in the proof of Lemma \ref{le:gffpw} that 
\begin{equation*}
\gffcf{\partial\varphi(x)\partial \varphi(y)\O}=\gffcf{\O}\bigg(\partial_x\partial_y G_\Omega(x,y)-\int_\Omega d^2 u f(u)\partial_x G_\Omega (u,x)\int_{\Omega}d^2 v f(v)\partial_y G_\Omega(v,y)\bigg)
\end{equation*}
and that the terms involving $f$ are smooth functions on $\Omega$‚ so they have no singularities. Our task is thus to show that the limit
\begin{equation*}
\lim_{x\to y}\left(\partial_x\partial_yG_\Omega(x,y)+\frac{1}{4\pi}\frac{1}{(x-y)^2}\right)
\end{equation*}
exists. This follows from \eqref{eq:gomega} since 
\begin{equation*}
\partial_x\partial_y \frac{1}{2\pi}\log |x-y|^{-1}=-\frac{1}{4\pi}\frac{1}{(x-y)^2}.
\end{equation*}
The other claims are similar, though for the cross term one notes that 
\begin{equation*}
\partial_x\bar\partial_y \frac{1}{2\pi}\log |x-y|^{-1}=0
\end{equation*}
for $x\neq y$.
\end{proof}

\section{Wick ordered exponentials and analysis of their correlation functions}\label{sec:wick}

In this section, we make sense of \emph{Wick ordered exponentials}, namely objects of the form $\:e^{i\alpha \varphi}\:$ which are needed for our construction of the sine-Gordon model (we define $\:\cos(\sqrt{\beta}\varphi)\:$ as the real part of $\:e^{i\sqrt{\beta}\varphi}\:$). We also analyze correlation functions involving such objects that we need in analysis of the sine-Gordon correlation functions. Our construction of the Wick ordered exponentials $\:e^{i\alpha \varphi}\:$ follows \cite{JSW}, where similar objects were studied, and the analysis of their correlation functions is based on \emph{Onsager inequalities} and associated moment bounds studied in \cite{GP,JSW}, though we require more detailed analysis for our purposes.

\subsection{Wick ordered exponentials}

Non-linear operations, such as exponentiation, are not well defined for generalized functions, so we need to construct the Wick ordered exponentials $\:e^{i\alpha\varphi}\:$ through \emph{regularization} and \emph{renormalization}. 

For this purpose, let $\rho\in C_c^\infty(\R^2)$ be radially symmetric, non-negative, and $\int_{\R^2}d^2x\, \rho(x)=1$.  To simplify notation, let us also assume that $\mathrm{supp}(\rho)\subset B(0,1)$. Let us also define $\rho_\epsilon(x)=\epsilon^{-2}\rho(\frac{x}{\epsilon})$. It follows from \cite[Proposition 2.3]{JSW} that if we extend $\varphi$ as zero outiside of $\Omega$, then $\varphi\in \mathcal S'(\R^2)$ (the space of tempered distributions on $\R^2$). This implies that the convolution\footnote{Recall that for $\varphi\in \mathcal S'(\R^d)$ and $\rho\in \mathcal S(\R^d)$, the convolution $\rho\star\varphi\in \mathcal S'(\R^d)$ is defined by the condition $(\rho\star\varphi)(f)=\varphi(\widetilde \rho\star f)$, where $\widetilde \rho(x)=\rho(-x)$. } $\varphi_\epsilon:=\rho_\epsilon\star\varphi$ is well defined for each $\epsilon>0$, and by \cite[Theorem 2.3.20]{Grafakos}, it is in fact a $C^\infty$ function. Note that all of these (Gaussian) random functions are constructed on a common probability space with $\varphi$. Moreover, it follows for example from \cite[Chapter IX.1]{RS} that  for each $f\in C_c^\infty(\Omega)$ and $\varphi\in \mathcal S'(\R^2)$, 
\begin{equation*}
\lim_{\epsilon\to 0}\varphi_\epsilon(f):=\lim_{\epsilon\to 0}\int_\Omega d^2 x\varphi_\epsilon(x)f(x)=\varphi(f).
\end{equation*}
Thus we can view $\varphi_\epsilon$ as an approximation or regularization of $\varphi$.

We record now some standard estimates for the covariance of $\varphi_\epsilon$ that we will make use of shortly when discussing Wick ordered exponentials. 
\begin{lemma}\label{le:varphiecov}
We have the following estimates for the covariance of the process $\varphi_\epsilon$.
\begin{enumerate}
\item For $K\subset \Omega$ compact, we have for $x,y\in K$ distinct and $0< \epsilon,\delta<\min(\mathrm{dist}(K,\partial\Omega),\frac{|x-y|}{3})$ 
\begin{equation}\label{eq:gffcov1}
\gffcf{\varphi_\epsilon(x)\varphi_\delta(y)}=G_\Omega(x,y).
\end{equation}
\item For $K\subset \Omega$ compact, $x\in K$, and $0<\epsilon<\mathrm{dist}(K,\partial \Omega)$, we have
\begin{equation}\label{eq:gffcov2}
\gffcf{\varphi_\epsilon(x)^2}=\frac{1}{2\pi}\log \epsilon^{-1}+\int_{\R^2\times \R^2}d^2z d^2 w\rho(z)\rho(w)\frac{1}{2\pi}\log |z-w|^{-1} +g_\Omega(x,x),
\end{equation}
where $g_\Omega$ is as in \eqref{eq:gomega}.
\item For $K\subset \Omega$ compact, there exists a constant $c=c_{K,\rho}>0$ such that for  $x,y\in K$ and $\frac{|x-y|}{3}\leq \max(\epsilon,\delta)<\mathrm{dist}(K,\partial \Omega)$, we have
\begin{align}\label{eq:gffcov3}
\left|\gffcf{\varphi_\epsilon(x)\varphi_\delta(y)}- \frac{1}{2\pi}\log \min(\epsilon^{-1},\delta^{-1})\right|\leq c.
\end{align} 
\end{enumerate}
\end{lemma}
The results here are well known, but as we are not aware of a reference with exactly these statements proven, and the proof is very simple, we present a proof here for the convenience of the reader. 
\begin{proof}
\begin{enumerate}
\item We start by noting that e.g. by \cite[the proof of Theorem 2.3.20]{Grafakos}, for $x,y\in \Omega$
\begin{equation*}
(\varphi_\epsilon(x),\varphi_\delta(y))=(\varphi(\tau^x(\rho_\epsilon)),\varphi(\tau^y (\rho_\delta))),
\end{equation*}
where $(\tau^x(\rho_\epsilon))(z)=\rho_\epsilon(z-x)$. Thus by \eqref{eq:gffcov}, we have for $x,y\in \Omega$
\begin{align*}
\gffcf{\varphi_\epsilon(x)\varphi_\delta(y)}&=\int_{\Omega\times \Omega}d^2 z\, d^2 w\, \rho_\epsilon(z-x)\rho_\delta(w-y)G_\Omega(z,w).
\end{align*}
Now let $K\subset \Omega$ be compact as in the claim and let $x,y\in K$. Since $\mathrm{supp}(\rho)\subset B(0,1)$, we see that $\rho_\epsilon(z-x)=0$ if $z\notin \Omega$ and $\epsilon<\mathrm{dist}(K,\partial \Omega)$. Thus for $x,y\in K$ and $\epsilon,\delta<\mathrm{dist}(K,\partial \Omega)$, we can write 
\begin{align}\label{eq:phiecov1}
\gffcf{\varphi_\epsilon(x)\varphi_\delta(y)}&=\int_{\R^2\times \R^2}d^2 z\, d^2 w\, \rho_\epsilon(z-x)\rho_\delta(w-y)G_\Omega(z,w)\\
&=\int_{\R^2\times \R^2} d^2 a\, d^2 b\, \rho(a)\rho(b)\left(\frac{1}{2\pi}\log |x+\epsilon a-y-\delta b|^{-1}+g_\Omega(x+\epsilon a,y+\delta b)\right)\notag \\
&=\int_{B(0,1)\times B(0,1)} d^2 a\, d^2 b\, \rho(a)\rho(b)\left(\frac{1}{2\pi}\log |x+\epsilon a-y-\delta b|^{-1}+g_\Omega(x+\epsilon a,y+\delta b)\right),\notag 
\end{align}
where in the last step we recalled that $\mathrm{supp}(\rho)\subset B(0,1)$. If we now consider $\epsilon,\delta<\min(\mathrm{dist}(K,\partial\Omega),\frac{|x-y|}{3})$ with $x,y$ distinct, we see that the function 
\begin{equation*}
a\mapsto \frac{1}{2\pi}\log |x+\epsilon a-y-\delta b|^{-1}+g_\Omega(x+\epsilon a,y+\delta b)
\end{equation*}
is harmonic in $B(0,1)$ (and similarly for $b$). Thus by rotational invariance of $\rho$ and the mean value property of harmonic functions, we see that for $\epsilon,\delta<\min(\mathrm{dist}(K,\partial\Omega),\frac{|x-y|}{3})$ with $x,y$ distinct,
\begin{equation*}
\gffcf{\varphi_\epsilon(x)\varphi_\delta(y)}=G_\Omega(x,y),
\end{equation*}
which was the first claim.
\item For the variance estimate, we use \eqref{eq:phiecov1} to write for $\epsilon<\mathrm{dist}(K,\partial\Omega)$
\begin{align*}
\gffcf{\varphi_\epsilon(x)^2}=\int_{B(0,1)\times B(0,1)}d^2a\, d^2 b\, \rho(a)\rho(b)\left(\frac{1}{2\pi}\log \epsilon^{-1}+\frac{1}{2\pi}\log |a-b|^{-1}+g_\Omega(x+\epsilon a,x+\epsilon b)\right).
\end{align*}
Again by the mean value property and rotational invariance of $\rho$, the last integral can be evaluated, and we find (using the fact that $\int \rho=1$) that for $\epsilon<\mathrm{dist}(K,\partial\Omega)$,
\begin{align*}
\gffcf{\varphi_\epsilon(x)^2}=\frac{1}{2\pi}\log \epsilon^{-1}+\int_{B(0,1)\times B(0,1)}d^2 a\, d^2 b\, \rho(a)\rho(b)\frac{1}{2\pi}\log |a-b|^{-1}+g_\Omega(x,x),
\end{align*} 
which was the claim. 
\item For notational simplicity, let us assume that $\delta\geq \epsilon$. From \eqref{eq:phiecov1}, we see (once again with the mean value property) that for $x,y\in K\subset \Omega$ compact and $\max(\epsilon,\delta)<\mathrm{dist}(K,\partial \Omega)$, 
\begin{align*}
\gffcf{\varphi_\epsilon(x)\varphi_\delta(y)}=\frac{1}{2\pi}\int_{\Omega\times \Omega}d^2 z\, d^2 w\, \rho_\epsilon(z-x)\rho_\delta(w-y)\log |z-w|^{-1}+g_\Omega(x,y).
\end{align*}
We then write $w=y+\delta b$, and find that 
\begin{align*}
\gffcf{\varphi_\epsilon(x)\varphi_\delta(y)}&=\frac{1}{2\pi}\int_{\Omega\times B(0,1)}d^2 z\, d^2 b\, \rho_\epsilon(z-x)\rho(b)\log |z-y-\delta b|^{-1}+g_\Omega(x,y)\\
&=\frac{1}{2\pi}\log \delta^{-1}+\frac{1}{2\pi}\int_{\Omega\times B(0,1)}d^2 z\, d^2b\, \rho_\epsilon(z-x)\rho(b)\log \left|\frac{z-y}{\delta}-b\right|^{-1}+g_\Omega(x,y). 
\end{align*}
With the change of variables $z=y+\delta a$, we write this as 
\begin{align*}
\gffcf{\varphi_\epsilon(x)\varphi_\delta(y)}&=\frac{1}{2\pi}\int_{\Omega\times B(0,1)}d^2 z\, d^2 b\, \rho_\epsilon(z-x)\rho(b)\log |z-y-\delta b|^{-1}+g_\Omega(x,y)\\
&=\frac{1}{2\pi}\log \delta^{-1}+\frac{1}{2\pi}\int_{\R^2\times B(0,1)}d^2 a\, d^2b\, \rho_{\epsilon/\delta}\left(a-\frac{x-y}{\delta}\right)\rho(b)\log \left|a-b\right|^{-1}\\
&\qquad +g_\Omega(x,y). 
\end{align*}
Note that under our assumptions, $\frac{|x-y|}{\delta}\leq 3$ and $\frac{\epsilon}{\delta}\leq 1$, so only $|a|\leq 4$ contributes to the integral. Thus 
\begin{align*}
&\left|\frac{1}{2\pi}\int_{\Omega\times B(0,1)}d^2 a\, d^2b\, \rho_{\epsilon/\delta}\left(a-\frac{x-y}{\delta}\right)\rho(b)\log \left|a-b\right|^{-1}\right|\\
&\leq \frac{1}{2\pi}\int_{B(0,4)}d^2a\,  \rho_{\epsilon/\delta}\left(a-\frac{x-y}{\delta}\right)\sup_{|u|\leq 4}\int_{B(0,1)}d^2 b \rho(b)|\log |u-b||\\
&=\frac{1}{2\pi}\sup_{|u|\leq 4}\int_{B(0,1)}d^2 b \rho(b)|\log |u-b||\\
&<\infty.
\end{align*}
This yields the claim.
\end{enumerate}
\end{proof}

The purpose of constructing $\varphi_\epsilon$ is to construct non-linear objects from it via renormalization. For Wick orderered exponentials, the approach is as follows: for $\alpha\in \R$ and $x\in \Omega$, let 
\begin{equation}\label{eq:wickdef}
\:e^{i\alpha\varphi_\epsilon(x)}\:=e^{\frac{\alpha^2}{2}c_\rho}\epsilon^{-\frac{\alpha^2}{4\pi}}e^{i\alpha\varphi_\epsilon(x)},
\end{equation}
where 
\begin{equation*}
c_\rho=\frac{1}{2\pi}\int_{B(0,1)\times B(0,1)}d^2a\, d^2 b\, \rho(a)\rho(b)\log |a-b|^{-1}.
\end{equation*}
Let us explain the reason for choosing this particular $\epsilon$-dependent quantity as the prefactor. It follows from \eqref{eq:gffcov2}  that for a fixed $x\in \Omega$, $\gffcf{\varphi_\epsilon(x)^2}=\frac{1}{2\pi}\log \epsilon^{-1}+c_\rho+g_\Omega(x,x)+o(1)$ as $\epsilon\to 0$. Since for a centered Gaussian random variable $V$ we have $\E(e^{iV})=e^{-\frac{1}{2}\E(V^2)}$, this means that $\gffcf{\:e^{i\alpha\varphi_\epsilon(x)}\:}=(1+o(1))e^{-\frac{\alpha^2}{2}g_\Omega(x,x)}$ as $\epsilon\to 0$. Thus $\:e^{i\alpha\varphi_\epsilon(x)}\:$ has at least a chance of converging as $\epsilon\to 0$. We also include the $c_\rho$-term in the definition of the Wick-ordered exponential to get correlation functions that are independent of $\rho$.

We mention here that one could normalize the Wick ordered exponentials in slightly different ways, e.g. by looking at $e^{\frac{\alpha^2}{2}\gffcf{\varphi_\epsilon(x)^2}}e^{i\alpha \varphi_\epsilon(x)}$. This is natural since this quantity has expectation one. This was the choice made in \cite{JSW} (where generalizations of these objects are called \emph{imaginary multiplicative chaos distributions}). We choose to go with the definition \eqref{eq:wickdef} since in \cite{BIVW} it was noticed that it is a convenient definition for studying OPEs (though we do not study OPEs of the Wick ordered exponentials in this article).

It turns out that convergence of $\:e^{i\alpha\varphi_\epsilon}\:$ is a slightly delicate question: it does not happen pointwise, but in the sense of distributions, and also it does not happen for all values of $\alpha$. The following lemma, which is an important building block in our construction of the sine-Gordon model, summarizes the relevant facts. 
\begin{lemma}\label{le:wickconv}
For $f\in C_c^\infty(\Omega)$ and $\alpha\in(-\sqrt{4\pi},\sqrt{4\pi})$, the random variables
\begin{equation*}
\:e^{i\alpha\varphi_\epsilon}\:(f)=\int_{\Omega}d^2 x\, \:e^{i\alpha\varphi_\epsilon(x)}\: f(x)
\end{equation*}
converge in probability as $\epsilon\to 0$. We write $\:e^{i\alpha\varphi}\:(f)$ for the limiting random variable.
\end{lemma}
A stronger version of this result was proven in \cite[Proposition 3.1]{JSW}, where it was proven that for\footnote{The normalization is slightly different in \cite{JSW}: the covariance of the field is assumed to have a singularity of the form $\log |x-y|^{-1}$ instead of $\frac{1}{2\pi}\log|x-y|^{-1}$, so the assumption on $\alpha$ in \cite[Proposition 3.1]{JSW} (which is denoted by $\gamma$ there) looks different from ours, but is indeed equivalent if one simply rescales the field $\varphi_\epsilon$ by $\sqrt{2\pi}$.} $\alpha\in (-\sqrt{4\pi},\sqrt{4\pi})$, the functions $x\mapsto e^{\frac{\alpha^2}{2}\gffcf{\varphi_\epsilon(x)^2}}e^{i\alpha\varphi_\epsilon(x)}$ converge in probability in a suitable Sobolev space of generalized functions. Lemma \ref{le:wickconv} would follow quite readily from this result combined with \eqref{eq:gffcov2}, but given that this is an important result in our construction of the sine-Gordon model, we give a proof for the sake of completeness. 
\begin{proof}
The idea of the proof is to show that for $f\in C_c^\infty(\Omega)$ (for notational simplicity, let us take $f$ to be real valued) and $\alpha\in (-\sqrt{4\pi},\sqrt{4\pi})$
\begin{equation}\label{eq:epsdel}
\lim_{\epsilon,\delta\to 0}\gffcf{|\:e^{i\alpha \varphi_\epsilon}\:(f)-\:e^{i\alpha \varphi_\delta}\:(f)|^2}=0.
\end{equation}
Convergence in $L^2(\P)$ then follows from completeness of $L^2(\P)$, and the claim follows from the fact that convergence in $L^2$ implies convergence in probability.

Let us thus try to show \eqref{eq:epsdel}. Since we are dealing with Gaussian processes, we have 
\begin{align*}
&\gffcf{|\:e^{i\alpha \varphi_\epsilon}\:(f)-\:e^{i\alpha \varphi_\delta}\:(f)|^2}\\
&=e^{\alpha^2 c_\rho}\int_{\Omega\times \Omega}d^2x\, d^2 y\, f(x)f(y) \bigg[\epsilon^{-\frac{\alpha^2}{2\pi}}e^{-\frac{\alpha^2}{2}\gffcf{(\varphi_\epsilon(x)-\varphi_\epsilon(y))^2}}-2\epsilon^{-\frac{\alpha^2}{4\pi}}\delta^{-\frac{\alpha^2}{4\pi}}e^{-\frac{\alpha^2}{2}\gffcf{(\varphi_\epsilon(x)-\varphi_\delta(y))^2}}\\
&\qquad \qquad \qquad \qquad +\delta^{-\frac{\alpha^2}{2\pi}}e^{-\frac{\alpha^2}{2}\gffcf{(\varphi_\delta(x)-\varphi_\delta(y))^2}}\bigg].
\end{align*}
Note that since $f\in C_c^\infty(\Omega)$, we can apply Lemma \ref{le:varphiecov} with $K=\mathrm{supp}(f)$. First from \eqref{eq:gffcov2}, we find for $0<\epsilon,\delta<\mathrm{dist}(\mathrm{supp}(f),\partial\Omega)$,
\begin{align*}
\epsilon^{-\frac{\alpha^2}{4\pi}}\delta^{-\frac{\alpha^2}{4\pi}}e^{-\frac{\alpha^2}{2}\gffcf{(\varphi_\epsilon(x)-\varphi_\delta(y))^2}}&=\epsilon^{-\frac{\alpha^2}{4\pi}} e^{-\frac{\alpha^2}{2}\gffcf{\varphi_\epsilon(x)^2}}\delta^{-\frac{\alpha^2}{4\pi}} e^{-\frac{\alpha^2}{2}\gffcf{\varphi_\delta(y)^2}} e^{\alpha^2\gffcf{\varphi_\epsilon(x)\varphi_\delta(y)}}\\
&=e^{-\alpha^2 c_\rho-\frac{\alpha^2}{2}(g_\Omega(x,x)+g_\Omega(y,y))} e^{\alpha^2\gffcf{\varphi_\epsilon(x)\varphi_\delta(y)}}.
\end{align*}
We can of course use this estimate in the case $\epsilon=\delta$ as well, so we find for $\epsilon,\delta<\mathrm{dist}(\mathrm{supp}(f),\partial \Omega)$ 
\begin{align}\label{eq:imc2mom}
&\gffcf{|\:e^{i\alpha\varphi_\epsilon}\:(f)-\:e^{i\alpha\varphi_\delta}\:(f)|^2}\\
&\quad =\int_{\Omega\times \Omega}d^2 x\, d^2y\, f(x)f(y)e^{-\frac{\alpha^2}{2}(g_\Omega(x,x)+g_\Omega(y,y))}\notag \\
&\qquad \qquad \times \left[e^{\alpha^2 \gffcf{\varphi_\epsilon(x)\varphi_\epsilon(y)}}-2e^{\alpha^2 \gffcf{\varphi_\epsilon(x)\varphi_\delta(y)}}+e^{\alpha^2 \gffcf{\varphi_\delta(x)\varphi_\delta(y)}}\right].\notag
\end{align}
Let us study the $\epsilon,\delta\to 0$ limits of the three integrals appearing here. First of all, we have for $\epsilon<\mathrm{dist}(\mathrm{supp}(f),\partial \Omega)$ (by \eqref{eq:gffcov1} and \eqref{eq:gffcov3})  
\begin{align*}
&\int_{\Omega\times \Omega}d^2x\, d^2y\, f(x)f(y)e^{-\frac{\alpha^2}{2}(g_\Omega(x,x)+g_\Omega(y,y))}e^{\alpha^2\gffcf{\varphi_\epsilon(x)\varphi_\epsilon(y)}}\\
&=\int_{|x-y|\geq 3\epsilon}d^2x\, d^2y\, f(x)f(y)e^{-\frac{\alpha^2}{2}(g_\Omega(x,x)+g_\Omega(y,y))}e^{\alpha^2\gffcf{\varphi_\epsilon(x)\varphi_\epsilon(y)}}\\
&\quad +\int_{|x-y|< 3\epsilon}d^2x\, d^2y\, f(x)f(y)e^{-\frac{\alpha^2}{2}(g_\Omega(x,x)+g_\Omega(y,y))}e^{\alpha^2\gffcf{\varphi_\epsilon(x)\varphi_\epsilon(y)}}\\
&=\int_{|x-y|\geq 3\epsilon}d^2x\, d^2y\, f(x)f(y)e^{-\frac{\alpha^2}{2}(g_\Omega(x,x)+g_\Omega(y,y))}e^{\alpha^2G_\Omega(x,y)}+O(\epsilon^{-\frac{\alpha^2}{2\pi}}\epsilon^2).
\end{align*}  
Since $\alpha^2< 4\pi$, we see from \eqref{eq:gomega} (say by dominated convergence) that this last integral has a limit as $\epsilon\to 0$ and the error term tends to zero as $\epsilon\to 0$. Thus 
\begin{align*}
&\lim_{\epsilon\to 0}\int_{\Omega\times \Omega}d^2x\, d^2y\, f(x)f(y)e^{-\frac{\alpha^2}{2}(g_\Omega(x,x)+g_\Omega(y,y))}e^{\alpha^2\gffcf{\varphi_\epsilon(x)\varphi_\epsilon(y)}}\\
&\quad =\int_{\Omega\times \Omega}d^2x\, d^2y\, f(x)f(y)e^{-\frac{\alpha^2}{2}(g_\Omega(x,x)+g_\Omega(y,y))}e^{\alpha^2G_\Omega(x,y)}.
\end{align*}
Of course we also have the same result with $\epsilon$ replaced by $\delta$, so it remains to estimate the cross term in \eqref{eq:imc2mom}. For this, let us for notational simplicity assume that $\delta\geq \epsilon$. The argument is essentially identical and we find that as $\delta\to 0$
\begin{align*}
&\int_{\Omega\times \Omega}d^2x\, d^2y\, f(x)f(y)e^{-\frac{\alpha^2}{2}(g_\Omega(x,x)+g_\Omega(y,y))}e^{\alpha^2\gffcf{\varphi_\epsilon(x)\varphi_\delta(y)}}\\
&=\int_{|x-y|\geq 3\delta}d^2x\, d^2y\, f(x)f(y)e^{-\frac{\alpha^2}{2}(g_\Omega(x,x)+g_\Omega(y,y))}e^{\alpha^2\gffcf{\varphi_\epsilon(x)\varphi_\delta(y)}}\\
&\quad +\int_{|x-y|< 3\delta}d^2x\, d^2y\, f(x)f(y)e^{-\frac{\alpha^2}{2}(g_\Omega(x,x)+g_\Omega(y,y))}e^{\alpha^2\gffcf{\varphi_\epsilon(x)\varphi_\delta(y)}}\\
&=\int_{|x-y|\geq 3\delta}d^2x\, d^2y\, f(x)f(y)e^{-\frac{\alpha^2}{2}(g_\Omega(x,x)+g_\Omega(y,y))}e^{\alpha^2G_\Omega(x,y)}+O(\delta^{-\frac{\alpha^2}{2\pi}}\delta^2)\\
&=\int_{\Omega\times \Omega}d^2x\, d^2y\, f(x)f(y)e^{-\frac{\alpha^2}{2}(g_\Omega(x,x)+g_\Omega(y,y))}e^{\alpha^2G_\Omega(x,y)}+o(1).
\end{align*}  
Plugging these estimates into \eqref{eq:imc2mom}, we conclude that \eqref{eq:epsdel} holds, and as argued at the beginning of the proof, this implies the claim.
\end{proof}

\subsection{Onsager inequalities} For the sine-Gordon model we will need to make sense of quantities like $\gffcf{\partial \varphi(x)\partial\varphi(y)\O e^{\mu \Re(\:e^{i\sqrt{\beta}\varphi}\:(\psi))}}$. The way we will do this is by proving that they are analytic in $\mu$ with Taylor coefficients given by  $\gffcf{\partial \varphi(x)\partial\varphi(y)\O[\Re(\:e^{i\sqrt{\beta}\varphi}\:(\psi))]^k}$. We will also need to make sense of the normalization constant (or partition function) of the model, which is $\gffcf{e^{\mu \Re(\:e^{i\sqrt{\beta}\varphi}\:(\psi))}}$. This also turns out to be analytic in $\mu$ with Taylor coefficients $\gffcf{[\Re(\:e^{i\sqrt{\beta}\varphi}\:(\psi))]^k}$. A formal Gaussian computation suggests that one should have for example
\begin{align*}
&\gffcf{[\Re(\:e^{i\sqrt{\beta}\varphi}\:(\psi))]^k}\\
&\quad =2^{-k}\sum_{\sigma_1,...,\sigma_k\in \{-1,1\}}\int_{\Omega^k} d^{2k}x \prod_{j=1}^k\left( \psi(x_j) e^{-\frac{\beta}{2}g_\Omega(x_j,x_j)} \right)e^{-\beta \sum_{1\leq l<m\leq k}\sigma_l\sigma_m G_\Omega(x_l,x_m)}.
\end{align*}

To estimate such integrals and justify this approach to the sine-Gordon model, we will make use of so-called \emph{Onsager} or \emph{electrostatic inequalities} which bound the sum appearing in the exponential above. A similar approach was taken already in \cite{GP,JSW}. We offer a proof here since it is simple enough and we give a slightly sharper inequality in the regularized case. While we do not make use of it, we generalize the results of \cite{GP,JSW} slightly and prove the inequality in the case that $\sigma_j\sqrt{\beta}$ is replaced by an arbitrary $\alpha_j\in (-\sqrt{4\pi},\sqrt{4\pi})$ since the proof is identical. The following lemma describes these inequalities.
\begin{lemma}\label{le:ons}
Let $K\subset \Omega$ be compact, $\alpha_1,...,\alpha_n\in (-\sqrt{4\pi},\sqrt{4\pi})$, and $x_1,...,x_n\in K$. 
\begin{enumerate}
\item There exists a constant $c=c_{K,\rho}$ such that for $0<\epsilon<\mathrm{dist}(K,\partial \Omega)$
\begin{equation*}
-\sum_{1\leq j<k\leq n}\alpha_j\alpha_k\gffcf{\varphi_\epsilon(x_j)\varphi_\epsilon(x_k)}\leq \frac{1}{4\pi}\sum_{j=1}^n \alpha_j^2 \log \frac{1}{\min_{k\neq j}|x_k-x_j|}+cn.
\end{equation*} 
\item There exists a constant $c=c_{K}$ such that 
\begin{equation*}
-\sum_{1\leq j<k\leq n}\alpha_j\alpha_kG_\Omega(x_j,x_k)\leq \frac{1}{4\pi}\sum_{j=1}^n \alpha_j^2\log \frac{1}{\min_{k\neq j}|x_k-x_j|}+cn.
\end{equation*}
\end{enumerate}
\end{lemma}
We mention here that the restriction $\alpha_1,...,\alpha_n\in (-\sqrt{4\pi},\sqrt{4\pi})$ is not essential -- without this restriction, the constants $c$ would simply need to depend on $\max_i |\alpha_i|$. However, for the corresponding exponential to be integrable, this restriction is important.  

We also mention that the main importance here is that in item (1), the constant $c$ does not depend on $n$ or $\epsilon$. This will imply that all moments of $\:e^{i\alpha\varphi_\epsilon}\:(f)$ are bounded in $\epsilon$ and grow at a controllable rate. Item (2) also implies that the moments of $\:e^{i\alpha\varphi}\:(f)$ grow at a controllable rate.  This is ultimately what justifies analyticity of the sine-Gordon correlation function as a function of $\mu$ in our approach.

Finally, before going into the proof, we mention that as in \cite[Proposition 3.9]{JSW}, one could at least in item (2) drop the assumption that $x_1,...,x_n\in K$ (and simply assume that $x_1,...,x_n\in \Omega$). As we do not need this case and the proof is cosmetically more involved, we do not state the inequality in this setting.
\begin{proof}
\begin{enumerate}
\item Let us introduce 
\begin{equation*}
r_j=\min\left(\frac{1}{4}\min_{k\neq j}|x_j-x_k|, \mathrm{dist}(K,\partial \Omega)\right)
\end{equation*}
and define the random variables 
\begin{equation*}
Z_j=\varphi_{\max(\epsilon,r_j)}(x_j).
\end{equation*}
These are of course centered jointly Gaussian random variables, and we have 
\begin{equation*}
\E(Z_jZ_k)=\gffcf{\varphi_{\max(\epsilon,r_j)}(x_j)\varphi_{\max(\epsilon,r_k)}(x_k)}.
\end{equation*}
We note that for $j\neq k$, we always have $|x_j-x_k|\geq 3r_j$ and $|x_j-x_k|\geq 3r_k$. If we also have $|x_j-x_k|\geq 3\epsilon$ and since we are assuming that $\epsilon<\mathrm{dist}(K,\partial \Omega)$, then by \eqref{eq:gffcov1}
\begin{equation*}
\E(Z_jZ_k)=\gffcf{\varphi_{\max(\epsilon,r_j)}(x_j)\varphi_{\max(\epsilon,r_k)}(x_k)}=G_\Omega(x_j,x_k)=\gffcf{\varphi_\epsilon(x_j)\varphi_\epsilon(x_k)}.
\end{equation*}
If on the other hand we have $|x_j-x_k|<3\epsilon$, then also $r_j,r_k<\epsilon$, and we also have 
\begin{align*}
\E(Z_jZ_k)=\gffcf{\varphi_{\max(\epsilon,r_j)}(x_j)\varphi_{\max(\epsilon,r_k)}(x_k)}=\gffcf{\varphi_\epsilon(x_j)\varphi_\epsilon(x_k)}.
\end{align*}
We thus have 
\begin{align*}
0\leq \E\left(\left(\sum_{l=1}^n\alpha_lZ_l\right)^2\right)=\sum_{l=1}^n \alpha_l^2 \E(Z_l^2)+2\sum_{1\leq j<k\leq n}\alpha_j\alpha_k\gffcf{\varphi_\epsilon(x_j)\varphi_\epsilon(x_k)},
\end{align*}
or 
\begin{align*}
-\sum_{1\leq j<k\leq n}\alpha_j\alpha_k\gffcf{\varphi_\epsilon(x_j)\varphi_\epsilon(x_k)}\leq \frac{1}{2}\sum_{l=1}^n \alpha_l^2 \E(Z_l^2).
\end{align*}
For $\E(Z_l^2)$, we note from \eqref{eq:gffcov3} that 
\begin{align*}
\E(Z_l^2)=\gffcf{\varphi_{\max(\epsilon,r_l)}(x_l)^2}\leq \frac{1}{2\pi}\log\min(\epsilon^{-1},r_l^{-1})+c\leq \frac{1}{2\pi}\log r_l^{-1}+c.
\end{align*}
Note that for each $K\subset \Omega$ compact‚ we can find a $c_K$ such that for $x_1,...,x_n\in K$
\begin{equation*}
\log r_l^{-1}\leq \log \frac{1}{\min_{j\neq l}|x_l-x_j|}+c_K.
\end{equation*}
We conclude that 
\begin{align*}
-\sum_{1\leq j <k\leq n}\alpha_j\alpha_k\gffcf{\varphi_\epsilon(x_j)\varphi_\epsilon(x_k)}\leq \frac{1}{2}\sum_{l=1}^n \alpha_l^2 \E(Z_l^2)\leq \frac{1}{4\pi}\sum_{l=1}^n \alpha_l^2 \log \frac{1}{\min_{m: m\neq  l}|x_l-x_m|}+n C
\end{align*}
for a suitable $C$ depending only on $K,\rho$.
\item This case follows from the first one by taking $\epsilon\to 0$ and using e.g. \eqref{eq:gffcov1}.
\end{enumerate}
\end{proof}

\subsection{Statements of the moment bounds}

In this section, we record the main technical estimates we need in this article. These are various bounds and regularity estimates for correlation functions involving $\:e^{\pm i \sqrt{\beta}\varphi}\:$. We present the proofs, which rely on a suitable graphical enumeration argument originally due to \cite{GP}, later in Section \ref{sec:mombdproof}.

Our first estimate is a bound on moments of $\:e^{i\alpha \varphi_\epsilon}\:(\psi)$ that is uniform in $\epsilon$, as well as a related bound on moments of $\:e^{i\alpha\varphi}\:(f)$. In addition, these estimates control the growth rate of these moments. This result is not new -- a variant of it is proven in \cite{GP}, and a more general result is given by \cite[Theorem 3.12]{JSW}. We provide a proof for it in Section \ref{sec:mombdproof} as proofs for further estimates we state shortly build on this proof. As the quantity $\max_i |\alpha_i|$ appears repeatedly in what follows, we introduce the following notation for it: for $\alpha_1,...,\alpha_n\in \R$, define 
\begin{equation}\label{eq:alphanorm}
\|\alpha\|:=\max_{1\leq i\leq n}|\alpha_i|.
\end{equation}

\begin{proposition}\label{pr:mombd1}
For $f_1,...,f_n\in C_c^\infty(\Omega)$ and $\alpha_1,...,\alpha_n\in(-\sqrt{4\pi},\sqrt{4\pi})$, there exists a constant $C$ depending only on $\Omega$, $\bigcup_{i=1}^n \mathrm{supp}(f_i)$, $\max_i \|f_i\|_{L^\infty(\Omega)}$, and $\|\alpha\|$ such that for $n\geq 1$, we have 
\begin{align}\label{eq:gpest}
\sup_{\epsilon>0}\gffcf{\left| \prod_{j=1}^n \:e^{ i \alpha_j\varphi_\epsilon}\:(f_j)\right|}\leq C^n n^{\frac{\|\alpha\|^2}{8\pi}n}.
\end{align} 
Moreover, we have 
\begin{align}\label{eq:icmom}
\gffcf{\prod_{j=1}^n \:e^{i\alpha_j\varphi}\:(f_j)}=\int_{\Omega^n}d^{2n}x \prod_{j=1}^n \left(f_j(x_j) e^{-\frac{\alpha_j^2}{2}g_\Omega(x_j,x_j)}\right) e^{-\sum_{1\leq k<l\leq n}\alpha_k\alpha_l G_\Omega(x_k,x_l)}.
\end{align}
\end{proposition}
\begin{remark}\label{re:vitali}
It will feature in our proof as well, but as we refer to this type of argument elsewhere as well, we emphasize here that \eqref{eq:gpest} upgrades the convergence in $L^2(\P)$ provided by Lemma \ref{le:wickconv} to convergence in $L^p(\P)$ for any $p\in[1,\infty)$. More precisely, \eqref{eq:gpest} implies that as a function of $\epsilon$, $\:e^{i\alpha\varphi_\epsilon}\:(f)$ is bounded in $L^p(\P)$ for any $p\in[1,\infty)$. In particular (with an application of Markov's inequality), $|\:e^{i\alpha\varphi_\epsilon}\:(f)|^p$ is uniformly integrable for any $p\in [1,\infty)$. Combining this fact with Lemma \ref{le:wickconv}, we can apply Vitali's convergence theorem to deduce that $\:e^{i\alpha\varphi_\epsilon}\:(f)\to \:e^{i\alpha\varphi}\:(f)$ in $L^p(\P)$ for any $p\in[1,\infty)$. 
\end{remark}

For our next two results, we need to estimate integrals involving $e^{-\sum_{1\leq j<k\leq n}\alpha_j\alpha_k G_\Omega(x_j,x_k)}$, but now with one or two of the variables not being integrated over. The following two results are the fundamental tools we use in the analysis of the sine-Gordon OPEs. We will need the notion of the Lipschitz norm of a Lipschitz continuous function. For $K\subset \Omega$ compact and $f:K\to \C$ Lipschitz continuous, we write 
\begin{equation*}
\|f\|_{\mathrm{Lip}(K)}:=\sup_{x\in K}|f(x)|+\sup_{\substack{x,y\in K:\\ x\neq y}}\frac{|f(x)-f(y)|}{|x-y|}.
\end{equation*}
\begin{proposition}\label{pr:mombd2}
For each $n\geq 2$, $\alpha_1,...,\alpha_n\in (-\sqrt{4\pi},\sqrt{4\pi})$, and $f_2,...,f_n\in C_c^\infty(\Omega)$ (possibly complex valued), let $G_n=G_n^{(\alpha,f)}:\Omega \to \C$, 
\begin{equation*}
G_n(x_1)=\int_{\Omega^{n-1}}d^2 x_2\cdots d^2 x_n \prod_{j=2}^n (f_j(x_j)) e^{-\sum_{1\leq k<l\leq n}\alpha_k \alpha_l G_\Omega(x_k,x_l)}.
\end{equation*}
Then for each $K\subset \Omega$ compact, $G_n$ is Lipschitz continuous on $K$ and there exists a constant $C$ depending only on $K$, $\Omega$, $\|\alpha\|$, $\bigcup_{i=2}^n \mathrm{supp}(f_i)$, and $\max_i \|f_i\|_{\mathrm{Lip}(\Omega)}$ such that 
\begin{equation*}
\|G_n\|_{\mathrm{Lip}(K)}\leq C^n n^{\frac{\|\alpha\|^2}{8\pi}n}.
\end{equation*}
\end{proposition}

Finally, we need an estimate for a similar integral where two variables are not integrated over.  
\begin{proposition}\label{pr:mombd3}
For each $n\geq 3$, $\alpha_1,...,\alpha_n\in (-\sqrt{4\pi},\sqrt{4\pi})$, and $f_3,...,f_n\in C_c^\infty(\Omega)$ (possibly complex valued), let $H_n:\{(x_1,x_2)\in \Omega^2: x_1\neq x_2\}\to \C$, 
\begin{equation*}
H_n(x_1,x_2)=\int_{\Omega^{n-2}}d^2 x_3\cdots d^2 x_n \prod_{j=3}^n (f_j(x_j)) e^{-\sum_{1\leq k<l\leq n}\alpha_k \alpha_l G_\Omega(x_k,x_l)}.
\end{equation*}
Then for each $K\subset \Omega$, there exists a constant $C$ depending only on $K$, $\Omega$, $\|\alpha\|$, $\bigcup_{i=3}^n \mathrm{supp}(f_i)$, and $\max_i \|f_i\|_{L^\infty(\Omega)}$ such that 
\begin{equation*}
|H_n(x_1,x_2)|\leq |x_1-x_2|^{-\frac{\|\alpha\|^2}{2\pi}} C^n n^{\frac{\|\alpha\|^2}{8\pi}n}.
\end{equation*}
for $x_1,x_2\in K$ with $x_1\neq x_2$.
\end{proposition}

Before going into the proofs, we record some of the basic notions and estimates we will need in the proofs. 

\subsection{Key notions and estimates needed for the moment bounds}\label{sec:tools}

The first step in our proofs will typically be to make use of Onsager's inequality, Lemma \ref{le:ons}, and deal with integrals involving 
\begin{equation*}
\prod_{j=1}^n \left(\min_{k\neq j}|x_k-x_j|\right)^{-\frac{\alpha_j^2}{4\pi}}.
\end{equation*} 
Given $\Omega$‚ we can always find some $R>0$ such that $\Omega\subset B(0,R)$, and in the relevant integrals, we will typically make use of this bound.

The minima here make these integrals hard to estimate directly, and the way we deal with this issue is to decompose the integration region in a way that keeps track of the structure of the relative distances of the points. It turns out that this decomposition is most convenient to describe in terms of suitable graphs associated with \emph{nearest neighbor maps}. We now define these maps and review the basic facts about them that we need. 

Given a point configuration $x=(x_1,...,x_{n})\in B(0,R)^{n}$, we define a function $F_x:\{1,...,n\}\to \{1,...,n\}$ by declaring that $F_x(i)=j$ if $x_j$ is the closest point to $x_i$. If there are several points $x_{j_1},...,x_{j_l}$ at the minimal distance from $x_i$, we define $F_x(i)$ to be the smallest of these indices.

\subsubsection{Nearest neighbor maps and their associated graphs}

In our proofs, we will decompose the integration region according to sets where different nearest neighbor mappings are realized, so we will need to sum over the different possibilities. For keeping track of the combinatorics, it is convenient to identify the nearest neighbor mapping with a certain directed graph. This directed graph has as its vertex set $\{1,...,n\}$ and it has edges from $i$ to $F_x(i)$. 

Not all directed graphs are possible since the nearest neighbor structure imposes constraints on the graph. To illustrate this, suppose that such a graph contains edges $(i \to j), (j \to k)$, $i \neq k$. This means that $x_k$ is either strictly closer to $x_j$ than $x_i$, or they are equally close and $k < i$. Now, if the graph contained a loop $(i_1 \to i_2), \ldots, (i_\ell \to i_1)$ of $\ell \geq 3$ edges then either there is a pair of consecutive edges corresponding to the ``strictly closer'' case in the loop, or $i_1 > i_3 > i_5 > \ldots$ and $i_2> i_4 > \ldots$. Chaining comparisons around the loop, one would obtain either that $x_{i_2}$ is strictly closer to $x_{i_1}$ than $x_{i_2}$, or $i_1 > i_1$, both a contradiction. On the other hand, by definition, every vertex has out degree one. It follows that one-directional edges of the graph of $F_x$ form trees rooted at, and directed toward, the vertices of two-loops. See for example Figure \ref{fig:forest}.  We call such graphs labelled directed two-loop rooted forests, and write $\mathcal F_2^{(n)}$ for their set.

Conversely, given such a graph $F\in \mathcal F_2^{(n)}$, there can be at most one nearest neighbor map $F_x$ corresponding to $F$. In other words,  the mapping from nearest-neighbour maps $F_x$ to labelled directed two-loop-rooted forests is an injection. Since we are looking for an upper bound, it will be fine to upper bound combinatorial factors by summing over all of $\mathcal F_2^{(n)}$. When we want to emphasize the number of connected components in the graph, we write $\mathcal F_{2,k}^{(n)}$ for the subset of $\mathcal F_2^{(n)}$ consisting of graphs with $k$ connected components (equivalently, $k$ two-loops). Note that always $k\leq n/2$. For an illustration of such a graph, see Figure \ref{fig:forest}.

\begin{figure}
\begin{center}
\begin{tikzpicture}

    \draw (-1,0) circle (0.2);
    \node[label={$1$}] at (-1,0) {};
    \draw (1,0) circle (0.2);
    \node[label={$2$}] at (1,0) {};
    \draw (-3,0) circle (0.2);
    \node[label={$3$}] at (-3,0) {};
    \draw (-5,1) circle (0.2);
    \node[label={$4$}] at (-5,1) {};
    \draw (-5,-1) circle (0.2);
    \node[label={$5$}] at (-5,-1) {};
    \draw (3,0) circle (0.2);
    \node[label={$6$}] at (3,0) {};
    \draw (5,0) circle (0.2);
    \node[label={$7$}] at (5,0) {};    

   \draw (-1,-3) circle (0.2);
    \node[label={$8$}] at (-1,-3) {};
   \draw (1,-3) circle (0.2);
    \node[label={$9$}] at (1,-3) {}; 
    \draw (-3,-3) circle (0.2);
    \node[label={$10$}] at (-3,-3) {};
    \draw (3,-2) circle (0.2);
    \node[label={$11$}] at (3,-2) {};
    \draw (3,-4) circle (0.2);
    \node[label={$12$}] at (3,-4) {};
    \draw (5,-2) circle (0.2);
    \node[label={$13$}] at (5,-2) {};

    \draw[-{Latex[length=3mm,width=3mm]},very thick] (-0.85,0.12) to[out=30, in=150] (0.85,0.12);
    \draw[-{Latex[length=3mm,width=3mm]},very thick] (0.85,-0.12) to[out=210, in=330] (-0.85,-0.12);
    \draw[-{Latex[length=3mm,width=3mm]},very thick] (-2.8,0)--(-1.2,0);
    \draw[-{Latex[length=3mm,width=3mm]},very thick] (-4.82,0.95)--(-3.2,0);
    \draw[-{Latex[length=3mm,width=3mm]},very thick] (-4.82,-0.95)--(-3.2,0);
    \draw[-{Latex[length=3mm,width=3mm]},very thick] (2.8,0)--(1.2,0);
    \draw[-{Latex[length=3mm,width=3mm]},very thick] (4.8,0)--(3.2,0);
    
    \draw[-{Latex[length=3mm,width=3mm]},very thick] (-0.85,-2.88) to[out=30, in=150] (0.85,-2.88);
    \draw[-{Latex[length=3mm,width=3mm]},very thick] (0.85,-3.12) to[out=210, in=330] (-0.85,-3.12);
    \draw[-{Latex[length=3mm,width=3mm]},very thick] (-2.8,-3)--(-1.2,-3);

    \draw[-{Latex[length=3mm,width=3mm]},very thick] (2.82,-2.05)--(1.2,-3);
    \draw[-{Latex[length=3mm,width=3mm]},very thick] (2.82,-3.95)--(1.2,-3);
    \draw[-{Latex[length=3mm,width=3mm]},very thick] (4.82,-2)--(3.18,-2);
\end{tikzpicture} 
\end{center}
\caption{An illustration of an element of $\mathcal F_{2,2}^{(13)}$.}
\label{fig:forest}
\end{figure}
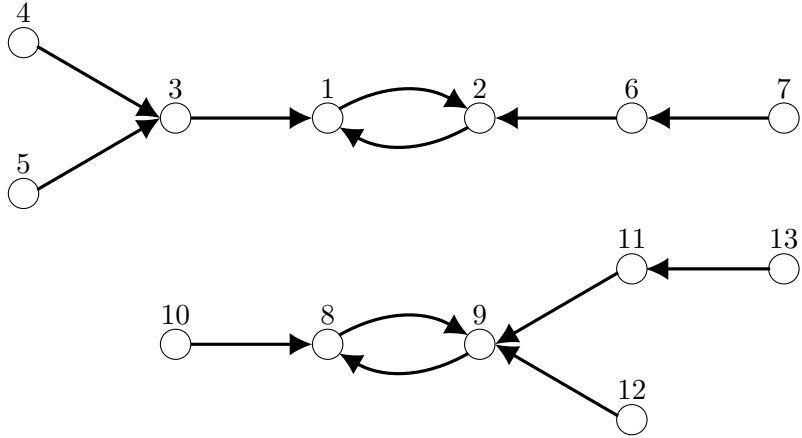

One of the things that makes this graphical decomposition of the integration region tractable is that we can estimate the number of such graphs. The precise fact we need, that has been proven e.g. in \cite[Lemma 3.10]{JSW} and \cite{GP} is recorded in the following lemma. For completeness, we provide a proof in Appendix \ref{app:tools}.
\begin{lemma}\label{le:nnnum}
For $k\leq n/2$, there exists an absolute constant $C>0$ such that 
\begin{equation*}
|\mathcal F_{2,k}^{(n)}|\leq C^n (n-k)!.
\end{equation*}
\end{lemma}

\subsubsection{Estimating an integral}

After performing the decomposition according to nearest neighbor maps, we will encounter integrals of the form 
\begin{equation*}
\int_{U(R)}d^{2(l+m)}u \prod_{j=1}^l |u_j|^{-\frac{\alpha}{2\pi}}\prod_{k=l+1}^{l+m}|u_k|^{-\frac{\alpha}{4\pi}},
\end{equation*}
where 
\begin{equation}\label{eq:UR}
U(R)=\left\{(u_1,...,u_{l+m})\in B(0,2R)^{l+m}: \sum_{j=1}^{l+m}|u_j|^2< (2R)^2\right\}.
\end{equation}
The following lemma records the basic estimate we will need for these integrals. A similar estimate already appears in \cite{GP,JSW}, however, since we state and use the estimate in a slightly different form, we present a proof in Appendix \ref{app:tools}.
\begin{lemma}\label{le:intest}
For $\alpha\in(0,4\pi)$ and $R>0$, there exist a constant $C=C_{\alpha,R}>0$ such that 
\begin{equation*}
\int_{U(R)}d^{2(l+m)}u \prod_{j=1}^l |u_j|^{-\frac{\alpha}{2\pi}}\prod_{k=l+1}^{l+m}|u_k|^{-\frac{\alpha}{4\pi}}\leq C^{l+m}\frac{1}{\Gamma(1+(l+m)(1-\frac{\alpha}{8\pi})-l\frac{\alpha}{8\pi})}.
\end{equation*}
\end{lemma}

We now turn to the proofs of the moment bounds we need.

\subsection{Proofs of the moment bounds}\label{sec:mombdproof}

As suggested above, the proof of Proposition \ref{pr:mombd1} is a minor modification of one due to \cite{GP}, but it serves as a building block for the other two proofs so we carry it out carefully. 

\begin{proof}[Proof of Proposition \ref{pr:mombd1}]
Despite having introduced the main tools, the proof  is still slightly lengthy so we split it into steps.

\medskip

\underline{Step 1: Applying the Onsager inequality}.  We begin by noting that since for any fixed $\epsilon>0$, $|\:e^{i\alpha\varphi_\epsilon}\:(f)| \leq \epsilon^{-\frac{\alpha^2}{4\pi}} e^{\frac{\alpha^2}{2}c_\rho}\int_\Omega d^2x |f(x)|$ (deterministically), it is sufficient to prove the claim for  $0<\epsilon<\mathrm{dist}(\bigcup_j\mathrm{supp}(f_j),\partial \Omega)$. Moreover, since for any random variable $X$, $\E(|X|)\leq \sqrt{\E(|X|^{2})}$,  we have
\begin{align*}
\gffcf{\left| \prod_{j=1}^n \:e^{ i \alpha_j\varphi_\epsilon}\:(f_j)\right|}\leq \left(\gffcf{\prod_{j=1}^n \:e^{ i \alpha_j\varphi_\epsilon}\:(f_j)\prod_{j=1}^n \:e^{- i \alpha_j\varphi_\epsilon}\:(\overline{f_j})}\right)^{1/2}
\end{align*}

Fubini combined with \eqref{eq:gffcov2} and the fact that for centered jointly Gaussian random variables $V_1,...,V_n$, we have  $\E[\prod_{j=1}^n e^{iV_j}]=e^{-\frac{1}{2}\sum_{j,k=1}^n \E(V_jV_k)}$, shows that 
\begin{align}\label{eq:2nmom}
&\gffcf{\prod_{j=1}^n \:e^{ i \alpha_j\varphi_\epsilon}\:(f_j)\prod_{j=1}^n \:e^{- i \alpha_j\varphi_\epsilon}\:(\overline{f_j})}\\
\notag &\quad =\epsilon^{-\sum_{j=1}^n\frac{\alpha_j^2}{2\pi}} e^{\sum_{j=1}^n\alpha_j^2c_\rho }\int_{\Omega^{2n}}d^{2n}x d^{2n}y \prod_{j=1}^n \left[f_j(x_j)\overline{f_j(y_j)} e^{-\frac{\alpha_j^2}{2}\gffcf{\varphi_\epsilon(x_j)^2}-\frac{\alpha_j^2}{2}\gffcf{\varphi_\epsilon(y_j)^2}}\right]\\
\notag &\qquad \times e^{- \sum_{1\leq k<l\leq n}\alpha_k\alpha_l(\gffcf{\varphi_\epsilon(x_k)\varphi_\epsilon(x_l)}+\gffcf{\varphi_\epsilon(y_k)\varphi_\epsilon(y_l)})} e^{\sum_{k,l=1}^n\alpha_k\alpha_l \gffcf{\varphi_\epsilon(x_k)\varphi_\epsilon(y_l)}}\\
&\notag\quad =\int_{\Omega^{2n}}d^{2n}x d^{2n}y \prod_{j=1}^n \left[f_j(x_j)\overline{f_j(y_j)} e^{-\frac{\alpha_j^2}{2}g_\Omega(x_j,x_j)-\frac{\alpha_j^2}{2}g_\Omega(y_j,y_j)}\right]\\
\notag &\qquad \times e^{- \sum_{1\leq k<l\leq n}\alpha_k \alpha_l(\gffcf{\varphi_\epsilon(x_k)\varphi_\epsilon(x_l)}+\gffcf{\varphi_\epsilon(y_k)\varphi_\epsilon(y_l)})} e^{\sum_{k,l=1}^n \alpha_k\alpha_l \gffcf{\varphi_\epsilon(x_k)\varphi_\epsilon(y_l)}}.
\end{align}
Since the $f_j$ are compactly supported, we see that $\max_j \sup_{x\in \cup_k \mathrm{supp}(f_k)}e^{-\frac{\alpha_j^2}{2}g_\Omega(x,x)}$ is bounded. Let $R>0$ be such that $\Omega\subset B(0,R)$. Then applying Lemma \ref{le:ons} (item (1)), renaming $y_1,...,y_n$ to $x_{n+1},...,x_{2n}$, and defining $\alpha_{n+j}=-\alpha_j$, we find that for some constant $C$ depending only on $\beta$, $\|\alpha\|$, $\cup_j\mathrm{supp}(f_j)$ and $\max_i \|f_i\|_{L^\infty(\Omega)}$, we have 
\begin{align*}
\gffcf{\prod_{j=1}^n \:e^{ i \alpha_j\varphi_\epsilon}\:(f_j)\prod_{j=1}^n \:e^{- i \alpha_j\varphi_\epsilon}\:(\overline{f_j})}\leq C^{2n} \int_{B(0,R)^{2n}}d^{4n}x e^{\sum_{k=1}^{2n} \frac{\|\alpha\|^2}{4\pi}\log \frac{1}{\min_{j:j\neq k}|x_j-x_k|}}.
\end{align*}
By rescaling the integration variables, we find for some other constant $\widetilde C$ (still depending only on the allowed quantities)
\begin{align*}
\gffcf{\prod_{j=1}^n \:e^{ i \alpha_j\varphi_\epsilon}\:(f_j)\prod_{j=1}^n \:e^{- i \alpha_j\varphi_\epsilon}\:(\overline{f_j})}\leq \widetilde C^{2n} \int_{B(0,1)^{2n}}d^{4n}x e^{\sum_{k=1}^{2n} \frac{\|\alpha\|^2}{4\pi}\log \frac{1}{\frac{1}{2}\min_{j:j\neq k}|x_j-x_k|}}.
\end{align*}
Introducing the factor $\frac{1}{2}$ inside the logarithm, at the price of changing $\widetilde C$ slightly, will turn out to be convenient later on.

\medskip

\underline{Step 2: Using nearest neighbor maps for a change of variables}. We will now start using the tools from Section \ref{sec:tools}. Given a graph $F\in \mathcal F_2^{(2n)}$, we define a set $U_F\subset B(0,1)^{2n}$ by the condition $U_F=\{(x_{1},...,x_{2n})\in B(0,1)^{2n}: F_x=F\}$. We thus find 
\begin{align*}
\gffcf{\prod_{j=1}^n \:e^{ i \alpha_j\varphi_\epsilon}\:(f_j)\prod_{j=1}^n \:e^{- i \alpha_j\varphi_\epsilon}\:(\overline{f_j})}&\leq \widetilde C^{2n}\sum_{F\in \mathcal F_2^{(2n)}}\int_{U_F} d^{4n}x \prod_{j=1}^{2n}\left(\frac{1}{2}|x_j-x_{F(j)}|\right)^{- \frac{\|\alpha\|^2}{4\pi}}\\
&=\widetilde C^{2n}\sum_{k=1}^{n}\sum_{F\in \mathcal F_{2,k}^{(2n)}}\int_{U_F} d^{4n}x \prod_{j=1}^{2n}\left(\frac{1}{2}|x_j-x_{F(j)}|\right)^{-\frac{\|\alpha\|^2}{4\pi}}\notag
\end{align*}

To estimate the integral, we perform a change of variables that depends on the structure of the graph $F$. By renaming the integration variables, we can assume that $\{1,2\}, \{3,4\},...,\{2k-1,2k\}$ are the vertices of the two-loops. We define for $1\leq j\leq k$, $u_{2j-1}=\frac{1}{2}(x_{2j-1}-x_{2j})$ and $u_{2j}=\frac{1}{2}x_{2j}$. For the remaining vertices ($i>2k$), we define $u_i=\frac{1}{2}(x_i-x_{F(i)})$. Note that this linear change of variables is a bijection, since we can solve each of the variables $x_i$ in terms of the $u$-variables recursively by ``climbing up a tree''.

If we write $\widetilde U_F$ for the image of $U_F$ under this change of variables, and note that the absolute value of the Jacobian determinant of the change of variables is $4^{2n}$ (the linear change of variables is upper triangular in the basis corresponding to the order of the tree), we find that 
\begin{align}\label{eq:comp}
\int_{U_F} d^{4n}x \prod_{j=1}^{2n}\left(\frac{1}{2}|x_j-x_{F(j)}|\right)^{- \frac{\|\alpha\|^2}{4\pi}}=4^{2n}\int_{\widetilde U_F}d^{4n}u \prod_{j=1}^k |u_{2j-1}|^{-\frac{\|\alpha\|^2}{2\pi}} \prod_{j=2k+1}^{2n}|u_j|^{-\frac{\|\alpha\|^2}{4\pi}}.
\end{align}

\medskip

\underline{Step 3: Modifying the integration region}. Apart from the integration region, and number of integration variables, this is of the form of Lemma \ref{le:intest}. We now show that we can enlarge the integration region to be of suitable form. For this purpose, note first of all that $|u_j|<1$ for all $j$ since $x_j\in B(0,1)$ for all $j$. Moreover, we point out that the balls $B_j=\{y\in \R^2: |y-x_j|< |u_j|\}$, with $j=1,3,...,2k-1,2k+1,2k+2,...,2n$ are disjoint. If they were not disjoint, there would be a point $y$ such that $|y-x_j|<|u_j|$ and $|y-x_{j'}|<|u_{j'}|$ for some $j\neq j'$. This would mean that $|x_j-x_{j'}|< 2\max(|u_j|,|u_{j'}|)$. Thus for one of the points $x_j,x_{j'}$, the other point would be closer than its nearest neighbor, contradicting the definition of the nearest neighbor, so we see that the balls are disjoint.

 Note that $B_j\subset B(0,2)$ for all $j$ (since $|x_j|<1$ and $|u_j|<1$). As the balls are disjoint, this means the sum of the areas of the balls $B_j$ must be less than the area of $B(0,2)$, or in other words, 
\begin{equation*}
\sum_{j=1}^k |u_{2j-1}|^2+\sum_{j=2k+1}^{2n}|u_j|^2\leq 4.
\end{equation*}
We conclude that 
\begin{equation*}
\widetilde U_F\subset \left\{u\in (\R^2)^{2n}: \sum_{j=1}^k |u_{2j-1}|^2+\sum_{j=2k+1}^{2n}|u_j|^2\leq 4, |u_2|,|u_4|,...,|u_{2k}|\leq 1\right\}.
\end{equation*}

Integrating out the variables $u_2,u_4,...,u_{2k}$, we recover an integral of the form studied in Lemma \ref{le:intest} (recall \eqref{eq:UR} for the notation $U(1)$):
\begin{align*}
&\gffcf{\prod_{j=1}^n \:e^{ i \alpha_j\varphi_\epsilon}\:(f_j)\prod_{j=1}^n \:e^{- i \alpha_j\varphi_\epsilon}\:(\overline{f_j})}\\
&\quad \leq C^{2n} \sum_{k=1}^n \sum_{F\in \mathcal F_{2,k}^{(2n)}}\int_{U(1)}d^{2(2n-k)}v \prod_{j=1}^k |v_j|^{-\frac{\|\alpha\|^2}{2\pi}}\prod_{j=k+1}^{2n-k}|v_j|^{- \frac{\|\alpha\|^2}{4\pi}},
\end{align*}
where once again, $C$ is a constant depending only on the allowed parameters.
\medskip

\underline{Step 4: Concluding the proof of \eqref{eq:gpest}}. Applying Lemma \ref{le:nnnum} and Lemma \ref{le:intest}, we conclude that there exists a $C$ depending only on the allowed parameters such that uniformly in $\epsilon$
\begin{align*}
\gffcf{\prod_{j=1}^n \:e^{ i \alpha_j\varphi_\epsilon}\:(f_j)\prod_{j=1}^n \:e^{- i \alpha_j\varphi_\epsilon}\:(\overline{f_j})}&\leq C^{2n} \sum_{k=1}^n \frac{(2n-k)!}{\Gamma(1+(2n-k)(1-\frac{\|\alpha\|^2}{8\pi})-k \frac{\|\alpha\|^2}{8\pi})}\\
&=C^{2n} \sum_{k=1}^n \frac{(2n-k)!}{\Gamma(1+2n(1-\frac{\|\alpha\|^2}{8\pi})-k)}.
\end{align*}
By Stirling's approximation for the Gamma function, uniformly in $k\in \{1,...,n\}$, 
\begin{align*}
\log \frac{(2n-k)!}{\Gamma(1+2n(1-\frac{\|\alpha\|^2}{8\pi})-k)}&=(2n-k)\log (2n-k)-(2n-k)\\
&\quad \notag -\left(1+2n\left(1-\frac{\|\alpha\|^2}{8\pi}\right)-k\right)\log \left(1+2n\left(1-\frac{\|\alpha\|^2}{8\pi}\right)-k\right)\\
&\quad \notag +1+2n\left(1-\frac{\|\alpha\|^2}{8\pi}\right)-k+O(\log n)\\
& \notag=(2n-k)\log n-\left(2n\left(1-\frac{\|\alpha\|^2}{8\pi}\right)-k\right)\log n+O(n)\\
& \notag=n \frac{\|\alpha\|^2}{4\pi}\log n+O(n), 
\end{align*}
so by noting that the sum produces a factor of $n$, which we can get rid of by modifying $C$, we conclude that 
\begin{align*}
\gffcf{\prod_{j=1}^n \:e^{ i \alpha_j\varphi_\epsilon}\:(f_j)\prod_{j=1}^n \:e^{- i \alpha_j\varphi_\epsilon}\:(\overline{f_j})}\leq C^{2n} n^{\frac{\|\alpha\|^2}{4\pi}n}
\end{align*}
and 
\begin{align*}
\gffcf{\left|\prod_{j=1}^n \:e^{ i \alpha_j\varphi_\epsilon}\:(f_j)\right|}\leq C^{n} n^{\frac{\|\alpha\|^2}{8\pi}n}
\end{align*}
for some constant $C$ depending only on the allowed parameters. This concludes the proof of \eqref{eq:gpest}.

\medskip

\underline{Completing the proof  of \eqref{eq:icmom}.} Recall from Remark \ref{re:vitali} that \eqref{eq:gpest} implies that $\:e^{i\alpha_j\varphi_\epsilon}\:(f_j)\to \:e^{i\alpha_j\varphi}\:(f_j)$ in $L^p(\P)$ for each $p\in [1,\infty)$. In particular, for each $j$, $\:e^{i\alpha_j\varphi}\:(f_j)\in L^p(\P)$ for all $p\in [1,\infty)$. By a repeated application of Cauchy-Schwarz, this implies that the expectation in \eqref{eq:icmom} is finite and 
\begin{equation*}
\gffcf{\prod_{j=1}^n \:e^{i\alpha_j\varphi}\:(f_j)}=\lim_{\epsilon\to 0}\gffcf{\prod_{j=1}^n \:e^{i\alpha_j\varphi_\epsilon}\:(f_j)}.
\end{equation*}
With a similar argument as in \eqref{eq:2nmom}, we find that for small enough $\epsilon$
\begin{align*}
\gffcf{\prod_{j=1}^n \:e^{i\alpha_j\varphi_\epsilon}\:(f_j)}&=\int_{\Omega^n}d^{2n}x \prod_{j=1}^n \left(f_j(x_j) e^{-\frac{\alpha_j^2}{2}g_\Omega(x_j,x_j)}\right) e^{-\sum_{1\leq k<l\leq n}\alpha_k\alpha_l \gffcf{\varphi_\epsilon(x_k)\varphi_\epsilon(x_l)}}.
\end{align*}
We see that Lemma \ref{le:varphiecov} implies that the covariances $\gffcf{\varphi_\epsilon(x_k)\varphi_\epsilon(x_l)}$ converge to $G_\Omega(x_k,x_l)$ as $\epsilon\to 0$, which means that if we could take the $\epsilon\to 0$ limit inside of the integral, we would have proven \eqref{eq:icmom}. This we can do with the dominated convergence theorem and Lemma \ref{le:ons}: we have 
\begin{align*}
e^{-\sum_{1\leq k<l\leq n}\alpha_k\alpha_l \gffcf{\varphi_\epsilon(x_k)\varphi_\epsilon(x_l)}}\leq e^{cn} e^{\sum_{k=1}^n \frac{\|\alpha\|^2}{4\pi}\log \frac{1}{\min_{l: l\neq k}|x_k-x_l|}}.
\end{align*}
Repeating the proof of \eqref{eq:gpest}, we see this to be integrable. This concludes the proof of \eqref{eq:icmom}.
\end{proof}

\begin{proof}[Proof of Proposition \ref{pr:mombd2}]
Let us begin with bounding $\|G_n\|_{L^\infty(K)}$. As mentioned, the proof builds on that of Proposition \ref{pr:mombd1} so we split the proof here into the same steps. 

\medskip

\underline{Step 1: Applying the Onsager inequality}. Letting $R>0$ be such that $\Omega\subset B(0,R)$, we use Lemma \ref{le:ons}, (item (2)), to find that for some constant $C$ depending only on $\|\alpha\|$, $\bigcup_{i=2}^n \mathrm{supp}(f_i)$, and $\max_i\|f_i\|_{L^\infty(\Omega)}$ 
\begin{align*}
|G_n(x_1)|\leq C^n \int_{B(0,R)^{n-1}}d^2x_2\cdots d^2 x_n e^{\sum_{k=1}^n \frac{\|\alpha\|^2}{4\pi}\log \frac{1}{\frac{1}{2}\min_{j\neq k}|x_j-x_k|}}.
\end{align*}
Note that carrying out the scaling is not so convenient now since we still have dependence on $x_1$.

\medskip

\underline{Step 2: Using nearest neighbor maps for a change of variables}.  The idea here is similar to the proof of Proposition \ref{pr:mombd1}, but we need to be more careful, since we are not integrating over $x_1$. Given a graph $F\in \mathcal F_2^{(n)}$ and a point $x_1\in K$, we define $U_F(x_1)=\{(x_2,...,x_n)\in B(0,R)^{n-1}: F_{(x_1,...,x_n)}=F\}$.

For the change of variables, we need to consider two possibilities: the vertex $1$ (corresponding to $x_1$ which is not integrated over) can be either (i) a vertex of a two-loop or (ii) a non-root vertex of a tree. Let us consider the case (i) first. Again, we can assume that the other vertices of the two-loops are $x_2,...,x_{2k}$ (with say $x_2$ connected to $x_1$). We define then 
\begin{align*}
u_{2j-1}=\frac{1}{2}(x_{2j-1}-x_{2j}) \qquad \text{for} \qquad 1\leq j\leq k
\end{align*}
\begin{align*}
u_{2j}=\frac{1}{2}x_{2j} \qquad \text{for} \qquad 2\leq j\leq k
\end{align*}
and
\begin{align*}
u_j=\frac{1}{2}(x_j-x_{F(j)}) \qquad \text{for} \qquad j>2k.
\end{align*}
Note that $u_2$ is not defined. Again, all of the variables $x_j$ can be recovered from $x_1$ and the variables $u_i$ by ``climbing up the tree'', so this change of variables is a linear bijection.

\medskip

\begin{figure}
\begin{center}
\begin{tikzpicture}

    \draw (-1,0) circle (0.2);
    \node[label={$3$}] at (-1,0) {};
    \draw (1,0) circle (0.2);
    \node[label={$2$}] at (1,0) {};
    \draw (-3,0) circle (0.2);
    \node[label={$4$}] at (-3,0) {};
    \draw (-5,1) circle (0.2);
    \node[label={$5$}] at (-5,1) {};
    \draw (-5,-1) circle (0.2);
    \node[label={$6$}] at (-5,-1) {};
    \draw (3,0) circle (0.2);
    \node[label={$1$}] at (3,0) {};
    \draw (5,0) circle (0.2);
    \node[label={$7$}] at (5,0) {};

    \draw (-1,-3) circle (0.2);
    \node[label={$3$}] at (-1,-3) {};
    \draw (1,-3) circle (0.2);
    \node[label={$2$}] at (1,-3) {};
    \draw (-3,-3) circle (0.2);
    \node[label={$4$}] at (-3,-3) {};
    \draw (-5,-2) circle (0.2);
    \node[label={$5$}] at (-5,-2) {};
    \draw (-5,-4) circle (0.2);
    \node[label={$6$}] at (-5,-4) {};
    \draw (3,-3) circle (0.2);
    \node[label={$1$}] at (3,-3) {};
    \draw (5,-3) circle (0.2);
    \node[label={$7$}] at (5,-3) {};

    \draw[-{Latex[length=3mm,width=3mm]},very thick] (-0.85,0.12) to[out=30, in=150] (0.85,0.12);
    \draw[-{Latex[length=3mm,width=3mm]},very thick] (0.85,-0.12) to[out=210, in=330] (-0.85,-0.12);
    \draw[-{Latex[length=3mm,width=3mm]},very thick] (-2.8,0)--(-1.2,0);
    \draw[-{Latex[length=3mm,width=3mm]},very thick] (-4.82,0.95)--(-3.2,0);
    \draw[-{Latex[length=3mm,width=3mm]},very thick] (-4.82,-0.95)--(-3.2,0);
    \draw[-{Latex[length=3mm,width=3mm]},very thick] (2.8,0)--(1.2,0);
    \draw[-{Latex[length=3mm,width=3mm]},very thick] (4.8,0)--(3.2,0);

    \draw[-{Latex[length=3mm,width=3mm]},very thick] (-0.85,-2.88) to[out=30, in=150] (0.85,-2.88);
    \draw[-{Latex[length=3mm,width=3mm]},very thick] (-0.85,-3.12) to[out=330, in=210] (0.85,-3.12);
    \draw[-{Latex[length=3mm,width=3mm]},very thick] (-2.8,-3)--(-1.2,-3);
    \draw[-{Latex[length=3mm,width=3mm]},very thick] (-4.82,-2.05)--(-3.2,-3);
    \draw[-{Latex[length=3mm,width=3mm]},very thick] (-4.82,-3.95)--(-3.2,-3);
    \draw[-{Latex[length=3mm,width=3mm]},very thick] (1.2,-3)--(2.8,-3);
    \draw[-{Latex[length=3mm,width=3mm]},very thick] (4.8,-3)--(3.2,-3);

\end{tikzpicture} 
\end{center}
\caption{An illustration of reorienting the graph in case ii): the original orientation is above, and the reoriented graph is below.}
\label{fig:treemod}
\end{figure}
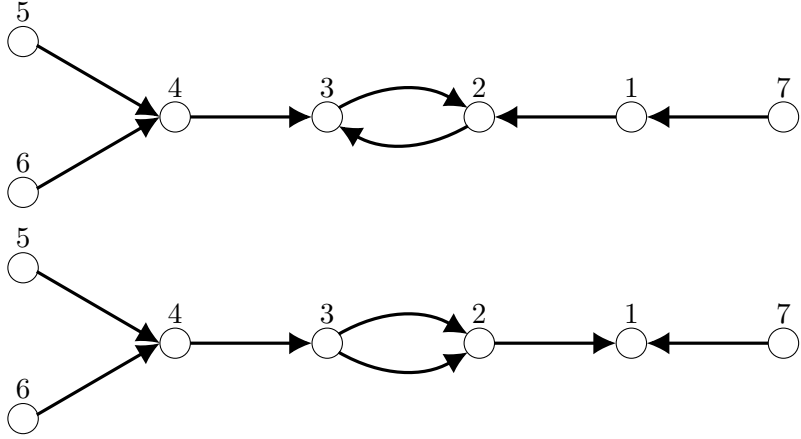

In case (ii), for the components not containing the vertex $1$, we proceed as in Step 2 of the proof of Proposition \ref{pr:mombd1}, but for the component containing the vertex $1$, we need to be more careful. We define a new orientation on this component by requiring that all edges flow toward the vertex $1$. See Figure \ref{fig:treemod} for an illustration of this reorienting. Let us write $\widetilde F$ for the function that takes the vertex $i$ to its neighbor following this flow. We define the new integration variables now as follows: $u_j=\frac{1}{2}(x_j-x_{\widetilde F(j)})$ for $j\neq 1$. 

All of the $x_j$ can now be recovered from the $u_i$ and $x_1$ by exploring the tree starting from the vertex $1$, so again, this is a linear bijection.

\medskip

\underline{The remaining steps}. The remaining steps are essentially the same as in proof of Proposition \ref{pr:mombd1}. The scaling by $R$ now just happens at a later stage, but the upshot is that 
\begin{equation*}
|G_n(x_1)|\leq C^n \sum_{1\leq k\leq n/2}\frac{(n-k)!}{\Gamma(1+n(1- \frac{\|\alpha\|^2}{8\pi})-k)}.
\end{equation*}
This can be estimated as in the previous case with Stirling, and one finds the bound we are after for $\|G_n\|_{L^\infty(K)}$.  

\medskip

For the Lipschitz-part, let $R>0$ be so large that $\Omega\subset B(0,R)$ and let us shift each of the integration variable $x_2,...,x_n$ by $x_1$ so that using \eqref{eq:gomega} we have  
\begin{align*}
G_n(x_1)&=\int_{B(0,2R)^{n-1}}d^2x_2\cdots d^2 x_n \prod_{j=2}^n \left(f_j(x_j+x_1)e^{-\alpha_1\alpha_j g_\Omega(x_1,x_j+x_1)}\right) \prod_{k=2}^n |x_k|^{-\frac{\alpha_1 \alpha_k}{2\pi}} \\
&\notag \qquad \times \prod_{2\leq l<m\leq n}\left(|x_l-x_m|^{-\frac{\alpha_l\alpha_m}{2\pi}} e^{-\alpha_l \alpha_m g_\Omega(x_l+x_1,x_m+x_1)} \right)\\
&\notag =: \int_{B(0,2R)^{n-1}}d^2x_2\cdots d^2 x_n U(x_1;x_2,...,x_n)\prod_{j=2}^n |x_j|^{-\frac{\alpha_1\alpha_j}{2\pi}} \prod_{2\leq l<m\leq n}|x_l-x_m|^{-\frac{\alpha_l\alpha_m}{2\pi}}.
\end{align*}
Let us now pick compact and connected $K'\subset \Omega$ such that $x_1,y_1\in K'$, and let $\gamma$ be a piecewise smooth curve within $K'$ from $x_1$ to $y_1$ for which the length of $\gamma$ is bounded by a constant times $|x_1-y_1|$ (for $|x_1-y_1|$ sufficiently small, we could take this just to be a line segment with an appropriate choice of $K'$, and for larger $|x_1-y_1|$ the existence of this follows by compactness). For such a curve, we have 
\begin{align*}
|U(x_1;x_2,...,x_n)-U(y_1;x_2,...,x_n)|\leq |\nabla_{x_1} U(\xi;x_2,...,x_n)| |x_1-y_1|,
\end{align*}
for some $\xi$ on $\gamma$. By the smoothness of $f_j$ and $g_\Omega$, we see that we can bound 
\begin{equation*}
 |\nabla_{x_1} U(\xi;x_2,...,x_n)|\leq C n^2 e^{-\sum_{j=2}^n \alpha_1\alpha_j g_\Omega(\xi,x_j+\xi)} e^{-\sum_{2\leq j<k\leq n}\alpha_j \alpha_k g_\Omega(x_j+\xi,x_k+\xi)}
\end{equation*}
for some constant $C$ which depends only on $\max_i \|\nabla f_i\|_{L^\infty(\Omega)}$, $\max_i\|f_i\|_{L^\infty(\Omega)}$, $\bigcup_{i=2}^n \mathrm{supp}(f_i)$, and $K$. This means that we have 
\begin{align*}
|G_n(x_1)-G_n(y_1)|&\leq C n^2 |x_1-y_1|\int_{B(0,2R)^{n-1}}d^2x_2\cdots d^2 x_n e^{-\sum_{j=2}^n \alpha_1\alpha_j G_\Omega(\xi,x_j+\xi)}\\
&\notag \qquad \qquad \qquad \times e^{-\sum_{2\leq k<l\leq n}\alpha_k\alpha_l G_\Omega(x_k+\xi,x_l+\xi)}.
\end{align*}
Even though $\xi$ depends on $x_1,y_1,x_2,...,x_n$, we can make use of Lemma \ref{le:ons} (item (2)), and bound this by  
\begin{align*}
|G_n(x_1)-G_n(y_1)|\leq C n^2|x_1-y_1|\int_{B(0,2R)^{n-1}}d^2 u_2\cdots d^2 u_n e^{\sum_{j=1}^n \frac{\|\alpha\|^2}{4\pi}\log \frac{1}{\frac{1}{2}\min_{k:k\neq j}|u_k-u_j|}},
\end{align*}
with the interpretation $u_1=0$. This is essentially the same integral we encountered when bounding $|G_n(x_1)|$, but simply with $x_1=0$. This means that from the bound for $\|G_n\|_{L^\infty(K)}$, we have 
\begin{align*}
|G_n(x_1)-G_n(y_1)|\leq C^n  n^2 n^{\frac{\|\alpha\|^2}{8\pi}n}|x_1-y_1|
\end{align*}
and adjusting the value of $C$, we find the required estimate for $\|G_n\|_{\mathrm{Lip}(K)}$.
\end{proof}

\begin{proof}[Proof of Proposition \ref{pr:mombd3}] $ $

\medskip

\underline{Preliminary considerations}: Again, the initial steps up to the change of variables are similar to the ones we have already encountered and we find
\begin{align*}
|H_n(x_1,x_2)|\leq C^n \sum_{1\leq k\leq n/2}\sum_{F\in \mathcal F_{2,k}^{(n)}} \int_{U_F(x_1,x_2)}d^2x_3\cdots d^2 x_n \prod_{k=1}^n \left(\frac{1}{2}|x_k-x_{F(k)}|\right)^{-\frac{\|\alpha\|^2}{4\pi}},
\end{align*}
where $U_F(x_1,x_2)=\{(x_3,...,x_n)\subset B(0,R)^{n-2}: F_{(x_1,x_2,...,x_n)}=F\}$, $R>0$ is such that $\Omega\subset B(0,R)$, and $C$ depends only on the allowed parameters.

\medskip

\underline{The change of variables and remainder of the analysis}: Here we have several cases we need to consider:
\begin{enumerate}
\item[i)]  $F(1)=2$ and $F(2)=1$ (i.e. $1$ and $2$ are the vertices of one two-loop).
\item[ii)] Either $F(1)=2$ and $F(2)\neq 1$ or vice versa (i.e. neighboring vertices, at least one is not a vertex of the two-loop).
\item[iii)] $1$ and $2$ are in different components.
\item[iv)] $F(1)\neq 2$, $F(2)\neq 1$, but $1$ and $2$ are in the same component.
\end{enumerate}
Let us now go through the change of variables and its consequences in each case -- we again need to be a bit more careful than in Proposition \ref{pr:mombd1} in most of the cases. In fact, in some cases a simple change of variables is not sufficient, but we need to get around this with some estimation. 
\begin{enumerate}
\item[i)] In this case, which we illustrate in Figure \ref{fig:case1}, we make the same change of variables as in Proposition \ref{pr:mombd1}, but ignoring the $1,2$ vertices. Again, climbing up the trees shows that this is a linear bijection, and the relevant integrals\footnote{Note however that in comparison with say \eqref{eq:comp}, there is one less of the integrals with exponent $-\|\alpha\|^2/2\pi$ since we do not integrate over the loop corresponding to $(1,2)$.} can be estimated as before. Carrying out the same arguments as before (applying Lemma \ref{le:nnnum} and Lemma \ref{le:intest}), we can thus upper bound the contribution of this case by 
\begin{equation*}
|x_1-x_2|^{-\frac{\|\alpha\|^2}{2\pi}}C^n \sum_{1\leq k\leq n/2}\frac{(n-k)!}{\Gamma(\frac{\|\alpha\|^2}{4\pi}+n(1-\frac{\|\alpha\|^2}{8\pi})-k)}
\end{equation*}
This sum can be bounded as before (using Stirling), and we arrive at a bound precisely of the form stated in the proposition. 

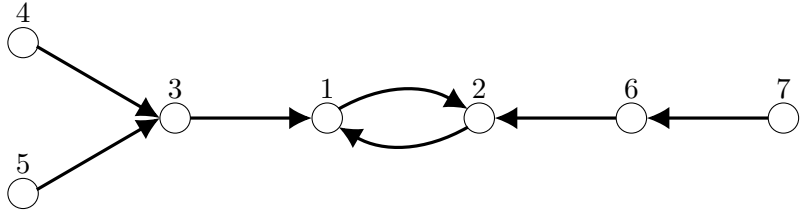
\begin{figure}
\begin{center}
\begin{tikzpicture}

    \draw (-1,0) circle (0.2);
    \node[label={$1$}] at (-1,0) {};
    \draw (1,0) circle (0.2);
    \node[label={$2$}] at (1,0) {};
    \draw (-3,0) circle (0.2);
    \node[label={$3$}] at (-3,0) {};
    \draw (-5,1) circle (0.2);
    \node[label={$4$}] at (-5,1) {};
    \draw (-5,-1) circle (0.2);
    \node[label={$5$}] at (-5,-1) {};
    \draw (3,0) circle (0.2);
    \node[label={$6$}] at (3,0) {};
    \draw (5,0) circle (0.2);
    \node[label={$7$}] at (5,0) {};

    \draw[-{Latex[length=3mm,width=3mm]},very thick] (-0.85,0.12) to[out=30, in=150] (0.85,0.12);
    \draw[-{Latex[length=3mm,width=3mm]},very thick] (0.85,-0.12) to[out=210, in=330] (-0.85,-0.12);
    \draw[-{Latex[length=3mm,width=3mm]},very thick] (-2.8,0)--(-1.2,0);
    \draw[-{Latex[length=3mm,width=3mm]},very thick] (-4.82,0.95)--(-3.2,0);
    \draw[-{Latex[length=3mm,width=3mm]},very thick] (-4.82,-0.95)--(-3.2,0);
    \draw[-{Latex[length=3mm,width=3mm]},very thick] (2.8,0)--(1.2,0);
    \draw[-{Latex[length=3mm,width=3mm]},very thick] (4.8,0)--(3.2,0);

\end{tikzpicture} 
\end{center}
\caption{A possible graph component arising in case i).}
\label{fig:case1}
\end{figure}

\item[ii)] Consider the case $F(1)=2$ and $F(2)\neq 1$ (the other case is treated similarly). This situation is illustrated in Figure \ref{fig:treemod}. For the components that do not contain the vertices $1,2$, we perform the same change of variables as in Step 3 of the proof of Proposition \ref{pr:mombd1}. As in the proof of Proposition \ref{pr:mombd2}, we define a new orientation on the component containing $1,2$ by requiring that all edges flow toward the edge $\{1,2\}$ (we can keep the orientation on this edge as it is or reverse it as in Figure \ref{fig:treemod}). We again call $\widetilde F$ the function that takes the vertex $i$ to its neighbor along this flow. We define the new integration variables on this component by $u_j=\frac{1}{2}(x_j-x_{\widetilde F(j)})$ for $j\neq 1,2$. Again, exploring the component starting from this edge shows that this is a linear bijection. The rest of the estimates can be carried out as in the proof of Proposition \ref{pr:mombd1}, and we arrive at a bound of the form 
\begin{equation*}
|x_1-x_2|^{-\frac{\|\alpha\|^2}{4\pi}}C^n\sum_{1\leq k\leq n/2}\frac{(n-k)!}{\Gamma( \frac{\|\alpha\|^2}{8\pi}+n(1-\frac{\|\alpha\|^2}{8\pi})-k)},
\end{equation*}
where $C$ depends only on the allowed parameters. We can bound this by the quantity in the claim of this proposition (by possibly modifying $C$).
\item[iii)] In this case, the components not containing $1$ or $2$ can be treated as in Step 3 of the proof of Proposition \ref{pr:mombd1} while the components containing $1$ or $2$ can be treated as in the proof of Proposition \ref{pr:mombd2}. We arrive at a bound of the form 
\begin{equation*}
C^n\sum_{1\leq k\leq n/2}\frac{(n-k)!}{\Gamma(1+n(1- \frac{\|\alpha\|^2}{8\pi})-k)}
\end{equation*}
(with the allowed dependence on the parameters in $C$), and this can be bounded by the quantity in the claim of the proposition (by possibly adjusting $C$). 
\item[iv)] This case is more involved and we cannot get away with simple changes of variables, but have to do some estimation. Moreover, we need to split this case into several subcases. Note that since $F(1)\neq F(2)$ and $F(2)\neq F(1)$, at most one of the vertices is a vertex of the two-loop of the component. Let us assume for simplicity that the vertex $1$ is not a vertex of the two-loop. 
\begin{enumerate}
\item[a.] The vertex $2$ is not a vertex of the two-loop, but the path to the two-loop from the vertex $1$ goes through vertex $2$ (or vice versa -- the second case is treated the same way, so we focus on the first one). This case is illustrated in Figure \ref{fig:case4a}.

\begin{figure}
\begin{center}
\begin{tikzpicture}

    \draw (-1,0) circle (0.2);
    \node[label={$6$}] at (-1,0) {};
    \draw (1,0) circle (0.2);
    \node[label={$7$}] at (1,0) {};
    \draw (-3,0) circle (0.2);
    \node[label={$3$}] at (-3,0) {};
    \draw (-5,1) circle (0.2);
    \node[label={$4$}] at (-5,1) {};
    \draw (-5,-1) circle (0.2);
    \node[label={$5$}] at (-5,-1) {};
    \draw (3,0) circle (0.2);
    \node[label={$2$}] at (3,0) {};
    \draw (5,0) circle (0.2);
    \node[label={$8$}] at (5,0) {};    
    \draw (7,0) circle (0.2);
    \node[label={$1$}] at (7,0) {};

    \draw[-{Latex[length=3mm,width=3mm]},very thick] (-0.85,0.12) to[out=30, in=150] (0.85,0.12);
    \draw[-{Latex[length=3mm,width=3mm]},very thick] (0.85,-0.12) to[out=210, in=330] (-0.85,-0.12);
    \draw[-{Latex[length=3mm,width=3mm]},very thick] (-2.8,0)--(-1.2,0);
    \draw[-{Latex[length=3mm,width=3mm]},very thick] (-4.82,0.95)--(-3.2,0);
    \draw[-{Latex[length=3mm,width=3mm]},very thick] (-4.82,-0.95)--(-3.2,0);
    \draw[-{Latex[length=3mm,width=3mm]},very thick] (2.8,0)--(1.2,0);
    \draw[-{Latex[length=3mm,width=3mm]},very thick] (4.8,0)--(3.2,0);
    \draw[-{Latex[length=3mm,width=3mm]},very thick] (6.8,0)--(5.2,0);
    
\end{tikzpicture} 
\end{center}
\caption{A possible graph component arising in case iv) a.}
\label{fig:case4a}
\end{figure}
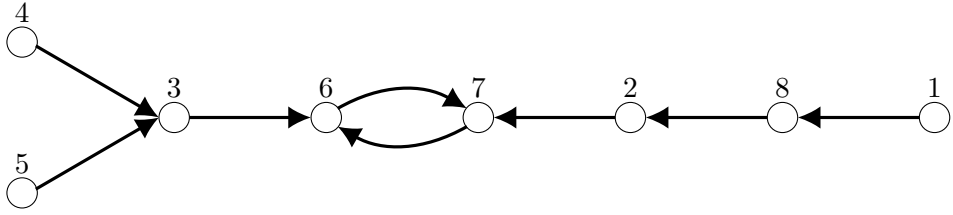

\medskip

Note that by definition, the distances between neighboring points grow along a (reversed) path from the two-loop to any tree vertex. This means that in our case, $|x_2-x_{F(2)}|\leq |x_1-x_{F(1)}|$. This implies that 
\begin{equation*}
\prod_{i=1}^n |x_i-x_{F(i)}|^{-\frac{\|\alpha\|^2}{4\pi}} \leq |x_2-x_{F(2)}|^{-\frac{\|\alpha\|^2}{2\pi}}\prod_{i=3}^n |x_i-x_{F(i)}|^{-\frac{\|\alpha\|^2}{4\pi}}. 
\end{equation*}
The effect of removing the $i=1$-term from the product is that it disconnects the component containing the vertices $1$ and $2$. These two components can now be treated as in the proof of Proposition \ref{pr:mombd2} (reorienting the graph), and the rest of the argument goes through analogously to the proof of Proposition \ref{pr:mombd1}. The upshot is that we get a bound of the form 
\begin{equation*}
C^n n^{\frac{\|\alpha\|^2}{8\pi}n}.
\end{equation*}

\item [b.] Vertex $2$ is a vertex of the two-loop and $F(1)\neq F(2)$. 

\medskip

\begin{figure}
\begin{center}
\begin{tikzpicture}

    \draw (-1,0) circle (0.2);
    \node[label={$3$}] at (-1,0) {};
    \draw (1,0) circle (0.2);
    \node[label={$2$}] at (1,0) {};
    \draw (-3,0) circle (0.2);
    \node[label={$4$}] at (-3,0) {};
    \draw (-5,1) circle (0.2);
    \node[label={$5$}] at (-5,1) {};
    \draw (-5,-1) circle (0.2);
    \node[label={$6$}] at (-5,-1) {};
    \draw (3,0) circle (0.2);
    \node[label={$7$}] at (3,0) {};
    \draw (5,0) circle (0.2);
    \node[label={$1$}] at (5,0) {};

    \draw[-{Latex[length=3mm,width=3mm]},very thick] (-0.85,0.12) to[out=30, in=150] (0.85,0.12);
    \draw[-{Latex[length=3mm,width=3mm]},very thick] (0.85,-0.12) to[out=210, in=330] (-0.85,-0.12);
    \draw[-{Latex[length=3mm,width=3mm]},very thick] (-2.8,0)--(-1.2,0);
    \draw[-{Latex[length=3mm,width=3mm]},very thick] (-4.82,0.95)--(-3.2,0);
    \draw[-{Latex[length=3mm,width=3mm]},very thick] (-4.82,-0.95)--(-3.2,0);
    \draw[-{Latex[length=3mm,width=3mm]},very thick] (2.8,0)--(1.2,0);
    \draw[-{Latex[length=3mm,width=3mm]},very thick] (4.8,0)--(3.2,0);

\end{tikzpicture} 
\end{center}
\caption{An illustration of the graph in case iv) b.}
\label{fig:case4b}
\end{figure}
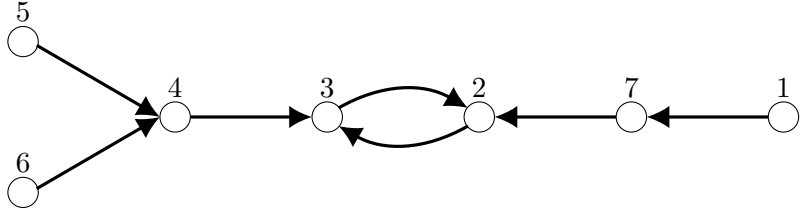

In this case, we can proceed in a similar manner as case a: we write 
\begin{equation*}
|x_1-x_{F(1)}|^{-\frac{\|\alpha\|^2}{4\pi}}\leq |x_{F(1)}-x_{F(F(1))}|^{-\frac{\|\alpha\|^2}{4\pi}}
\end{equation*}
and carry on as in case a. This produces the same bound. 
\item [c.] Vertex $2$ is a vertex of the two-loop and $F(1)=F(2)$. See Figure \ref{fig:case4c}.

\medskip

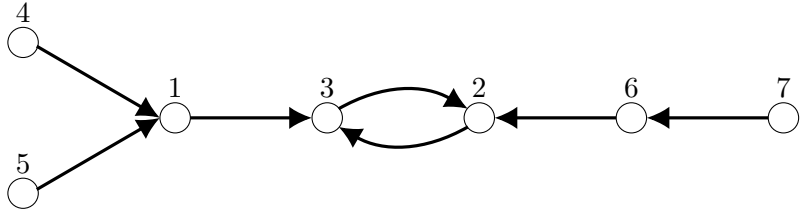
\begin{figure}
\begin{center}
\begin{tikzpicture}

    \draw (-1,0) circle (0.2);
    \node[label={$3$}] at (-1,0) {};
    \draw (1,0) circle (0.2);
    \node[label={$2$}] at (1,0) {};
    \draw (-3,0) circle (0.2);
    \node[label={$1$}] at (-3,0) {};
    \draw (-5,1) circle (0.2);
    \node[label={$4$}] at (-5,1) {};
    \draw (-5,-1) circle (0.2);
    \node[label={$5$}] at (-5,-1) {};
    \draw (3,0) circle (0.2);
    \node[label={$6$}] at (3,0) {};
    \draw (5,0) circle (0.2);
    \node[label={$7$}] at (5,0) {};

    \draw[-{Latex[length=3mm,width=3mm]},very thick] (-0.85,0.12) to[out=30, in=150] (0.85,0.12);
    \draw[-{Latex[length=3mm,width=3mm]},very thick] (0.85,-0.12) to[out=210, in=330] (-0.85,-0.12);
    \draw[-{Latex[length=3mm,width=3mm]},very thick] (-2.8,0)--(-1.2,0);
    \draw[-{Latex[length=3mm,width=3mm]},very thick] (-4.82,0.95)--(-3.2,0);
    \draw[-{Latex[length=3mm,width=3mm]},very thick] (-4.82,-0.95)--(-3.2,0);
    \draw[-{Latex[length=3mm,width=3mm]},very thick] (2.8,0)--(1.2,0);
    \draw[-{Latex[length=3mm,width=3mm]},very thick] (4.8,0)--(3.2,0);

\end{tikzpicture} 
\end{center}
\caption{An illustration of the graph in case iv) c.}
\label{fig:case4c}
\end{figure}

 Now we cannot use the same trick, since this would produce a term of the form $|x_{2}-x_{F(2)}|^{-\frac{3\|\alpha\|^2}{4\pi}}$ which may not be integrable. Now using the fact that $F(1)=F(2)$ (or in other words, $2=F(F(1))$), we can write 
\begin{equation*}
|x_1-x_2|=|x_1-x_{F(F(1))}|\leq |x_1-x_{F(1)}|+|x_{F(1)}-x_{F(F(1))}|\leq 2|x_1-x_{F(1)}|
\end{equation*}
implying that 
\begin{equation*}
|x_1-x_{F(1)}|^{-\frac{\|\alpha\|^2}{4\pi}}\leq \left(\frac{|x_1-x_2|}{2}\right)^{-\frac{\|\alpha\|^2}{4\pi}}.
\end{equation*}
Using this estimate in the product $\prod_i|x_i-x_{F(i)}|^{-\frac{\|\alpha\|^2}{4\pi}}$ allows us to again split the component containing $1$ and $2$ into two disjoint graphs, where we can use the approach from the proof of Proposition \ref{pr:mombd2}. This produces a bound of the form 
\begin{equation*}
C^n n^{\frac{\|\alpha\|^2}{8\pi}} |x_1-x_2|^{-\frac{\|\alpha\|^2}{4\pi}}.
\end{equation*}
\item [d.] Vertex $2$ is not a vertex of the two-loop, and neither vertex lies on the path from the two-loop to the other vertex. See Figure \ref{fig:case4d}.

\medskip

\begin{figure}
\begin{center}
\begin{tikzpicture}

    \draw (-1,0) circle (0.2);
    \node[label={$3$}] at (-1,0) {};
    \draw (1,0) circle (0.2);
    \node[label={$5$}] at (1,0) {};
    \draw (-3,0) circle (0.2);
    \node[label={$4$}] at (-3,0) {};
    \draw (-5,1) circle (0.2);
    \node[label={$1$}] at (-5,1) {};
    \draw (-5,-1) circle (0.2);
    \node[label={$2$}] at (-5,-1) {};
    \draw (3,0) circle (0.2);
    \node[label={$6$}] at (3,0) {};
    \draw (5,0) circle (0.2);
    \node[label={$7$}] at (5,0) {};

    \draw[-{Latex[length=3mm,width=3mm]},very thick] (-0.85,0.12) to[out=30, in=150] (0.85,0.12);
    \draw[-{Latex[length=3mm,width=3mm]},very thick] (0.85,-0.12) to[out=210, in=330] (-0.85,-0.12);
    \draw[-{Latex[length=3mm,width=3mm]},very thick] (-2.8,0)--(-1.2,0);
    \draw[-{Latex[length=3mm,width=3mm]},very thick] (-4.82,0.95)--(-3.2,0);
    \draw[-{Latex[length=3mm,width=3mm]},very thick] (-4.82,-0.95)--(-3.2,0);
    \draw[-{Latex[length=3mm,width=3mm]},very thick] (2.8,0)--(1.2,0);
    \draw[-{Latex[length=3mm,width=3mm]},very thick] (4.8,0)--(3.2,0);

\end{tikzpicture} 
\end{center}
\caption{An illustration of the graph in case iv) d.}
\label{fig:case4d}
\end{figure}
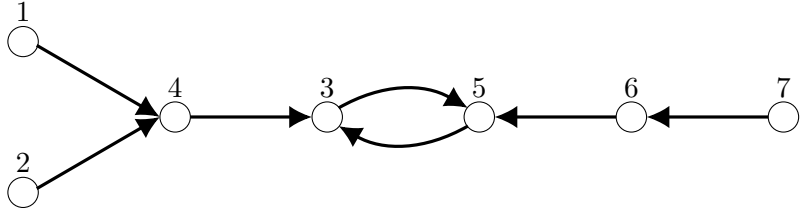

We now use the inequality 
\begin{equation*}
|x_1-x_{F(1)}|^{-\frac{\|\alpha\|^2}{4\pi}}|x_2-x_{F(2)}|^{-\frac{\|\alpha\|^2}{4\pi}}\leq |x_1-x_{F(1)}|^{-\frac{\|\alpha\|^2}{2\pi}}+|x_2-x_{F(2)}|^{-\frac{\|\alpha\|^2}{2\pi}}.
\end{equation*}
This allows disconnecting the component containing $1$ and $2$ as before, and the same argument amounts in 
\begin{align*}
C^n n^{\frac{\max_i|\alpha_i|^2}{8\pi}n}.
\end{align*}
\end{enumerate}
In total, all of these cases provide an estimate that fits in the bound claimed in the proposition. 
\end{enumerate}
\end{proof}

\section{The sine-Gordon model and its correlation functions}\label{sec:sG}

In this section, we define the sine-Gordon model and its correlation functions that we analyze for the OPEs. 

\subsection{The sine-Gordon model}

Recall from Section \ref{sec:intro} that we would like to define our model as a probability measure (on a space of generalized functions) which is absolute continuous with respect to the law of the GFF: for $\beta\in(0,4\pi)$, $\mu\in \R$ and $\psi\in C_c^\infty(\Omega)$ (real valued),
\begin{equation}\label{eq:sgedef4}
\P_{\sg(\beta,\mu,\psi)}(d\varphi)=\frac{1}{Z_{\sg(\beta,\mu,\psi)}}e^{\mu \int_{\Omega}d^2 x \psi(x)\:\cos(\sqrt{\beta}\varphi(x))\:}\P_{\gff}(d\varphi),
\end{equation}
where 
\begin{equation*}
Z_{\sg(\beta,\mu,\psi)}=\gffcf{e^{\mu \int_{\Omega}d^2 x \psi(x)\:\cos(\sqrt{\beta}\varphi(x))\:}}.
\end{equation*}

Having made sense of the Wick ordered exponentials, we see that the natural interpretation for a rigorous definition of the cosine term in \eqref{eq:sgedef4} (for real valued $\psi$) is  
\begin{equation*}
\int_{\Omega}d^2x \psi(x)\:\cos(\sqrt{\beta}\varphi(x))\:=\Re(\:e^{i\sqrt{\beta}\varphi}\:(\psi)).
\end{equation*}
This suggests trying to define the measure \eqref{eq:sgedef4}  rigorously as
\begin{equation*}
\P_{\sg(\beta,\mu,\psi)}(d\varphi)=\frac{1}{Z_{\sg(\beta,\mu,\psi)}}e^{\mu \Re(\:e^{i\sqrt{\beta}\varphi}\:(\psi))}\P_{\gff}(d\varphi),
\end{equation*}
where 
\begin{equation*}
Z_{\sg(\beta,\mu,\psi)}=\gffcf{e^{\mu \Re(\:e^{i\sqrt{\beta}\varphi}\:(\psi))}}.
\end{equation*}
This definition requires of course that $Z_{\sg(\beta,\mu,\psi)}$ is finite (and non-zero). This is the content of the following lemma. 
\begin{lemma}\label{le:sgpf}
For each $\mu\in \R$, $\beta\in(0,4\pi)$ and $\psi\in C_c^\infty(\Omega)$, $Z_{\sg(\beta,\mu,\psi)}\in(0,\infty)$.
\end{lemma}
\begin{proof}
Note that it is sufficient to prove that 
\begin{equation*}
\gffcf{e^{\mu |\Re(\:e^{i\sqrt{\beta}\varphi}\:(\psi))|}}<\infty
\end{equation*}
for each $\mu>0$. For this in turn, it is sufficient to prove that 
\begin{align*}
\gffcf{e^{\mu |\:e^{i\sqrt{\beta}\varphi}\:(\psi)|}}<\infty.
\end{align*}
Using Fubini and Proposition \ref{pr:mombd1},\footnote{To be precise, we use Vitali's convergence theorem to extend the bounds of Proposition \ref{pr:mombd1} to $\:e^{i\sqrt{\beta}\varphi}\:(\psi)$.} we see that 
\begin{align*}
\gffcf{e^{\mu |\:e^{i\sqrt{\beta}\varphi}\:(\psi)|}}\leq \sum_{n=0}^\infty \frac{\mu^n}{n!}C^n n^{\frac{\beta}{8\pi}n}, 
\end{align*}
for some $C$ depending only on $\beta$ and $\psi$. For example, by Stirling's approximation, this series converges for any $\mu>0$ so we are done.
\end{proof}
 
 We have thus constructed the sine-Gordon model for $\beta<4\pi$, and for example for $\O=e^{i\varphi(f)}$ with $f\in C_c^\infty(\Omega)$, 
 \begin{align*}
 \sgcf{\O}=\frac{\gffcf{e^{i\varphi(f)}e^{\mu \Re(\:e^{i\sqrt{\beta}\varphi}\:(\psi))}}}{\gffcf{e^{\mu \Re(\:e^{i\sqrt{\beta}\varphi}\:(\psi))}}}
 \end{align*}
 is well defined.
 
  Note that with a little bit of effort (using Lemma \ref{le:wickconv} and Proposition \ref{pr:mombd1} to justify Vitali's convergence theorem), one could also show that 
 \begin{equation*}
 \P_{\sg(\beta,\mu,\psi)}(d\varphi)=\lim_{\epsilon\to 0} \frac{1}{Z_{\sg(\beta,\mu,\psi)}(\epsilon)} e^{\Re(\:e^{i\sqrt{\beta}\varphi_\epsilon}\:(\psi))}\P_{\gff}(d\varphi),
 \end{equation*}
 where $Z_{\sg(\beta,\mu,\psi)}(\epsilon)=\gffcf{e^{\Re(\:e^{i\sqrt{\beta}\varphi_\epsilon}\:(\psi))}}$,
 as suggested in Section \ref{sec:intro}, but as we have already constructed the model, we will not elaborate on this further. 
 
 \subsection{The relevant smeared correlation functions for the sine-Gordon model}
 
 To study the OPEs for Theorem \ref{th:main}, we introduce smeared correlation functions and find integral kernels for them.

The correlation functions we need are as follows: for $f,g,h,\eta\in C_c^\infty(\Omega)$, we consider 
\begin{equation*}
\sgcf{\partial \varphi(g)\partial\varphi(h)\O}:=\frac{\gffcf{\partial \varphi(g)\partial\varphi(h)\O e^{\mu \Re(\:e^{i\sqrt{\beta}\varphi}\:(\psi))}}}{Z_{\sg(\beta,\mu,\psi)}}
\end{equation*}
and 
\begin{equation*}
\sgcf{\:e^{\pm i\sqrt{\beta}\varphi}\:(\eta)\O}:=\frac{\gffcf{\:e^{\pm i \sqrt{\beta}\varphi}\:(\eta)\O e^{\mu \Re(\:e^{i\sqrt{\beta}\varphi}\:(\psi))}}}{Z_{\sg(\beta,\mu,\psi)}},
\end{equation*}
where $\O=e^{i\varphi(f)}$ as before, and we also consider correlation functions with one or both of the $\partial$s replaced by $\bar\partial$. Note that by the proof of Lemma \ref{le:sgpf}, Proposition \ref{pr:mombd1}, and Cauchy-Schwarz, these are all finite.

We will shortly need integral kernels for these correlation functions. We will express these through a series expansion in $\mu$ for the numerator. The convergence of this series expansion is provided by the following lemma. 
\begin{lemma}\label{le:sgsmeared}
The functions 
\begin{equation}\label{eq:zmugh}
Z_\mu(g,h)=\gffcf{\partial \varphi(g)\partial\varphi(h)\O e^{\mu \Re(\:e^{i\sqrt{\beta}\varphi}\:(\psi))}},
\end{equation} 
\begin{equation*}
Z_\mu(\eta)=\gffcf{\:e^{\pm i \sqrt{\beta}\varphi}\:(\eta)\O e^{\mu \Re(\:e^{i\sqrt{\beta}\varphi}\:(\psi))}},
\end{equation*}
and 
\begin{equation*}
Z_\mu(\O)=\gffcf{\O e^{\mu \Re(\:e^{i\sqrt{\beta}\varphi}\:(\psi))}}
\end{equation*}
are entire functions of $\mu$ with (everywhere convergent) series expansions
\begin{align}\label{eq:zmuseries}
Z_\mu(g,h)=\sum_{n=0}^\infty \frac{\mu^n}{2^n n!}\sum_{\sigma_1,...,\sigma_n\in \{-1,1\}}\gffcf{\partial \varphi(g)\partial\varphi(h)\O\prod_{j=1}^n \:e^{i\sigma_j\sqrt{\beta}\varphi}\:(\psi)},
\end{align}
\begin{align}\label{eq:zmuseries2}
Z_\mu(\eta)=\sum_{n=0}^\infty \frac{\mu^n}{2^n n!}\sum_{\sigma_1,...,\sigma_n\in \{-1,1\}}\gffcf{\:e^{\pm i \sqrt{\beta}\varphi}\:(\eta)\O\prod_{j=1}^n \:e^{i\sigma_j\sqrt{\beta}\varphi}\:(\psi)},
\end{align}
and 
\begin{align}\label{eq:zmuseries3}
Z_\mu(\O)=\sum_{n=0}^\infty \frac{\mu^n}{2^n n!}\sum_{\sigma_1,...,\sigma_n\in \{-1,1\}}\gffcf{\O\prod_{j=1}^n \:e^{i\sigma_j\sqrt{\beta}\varphi}\:(\psi)}
\end{align}
Similar claims are true if one or both of the $\partial$s are replaced by $\bar\partial$.
\end{lemma}
\begin{proof}
Let us first give a heuristic argument and then justify why the steps we take are allowed. Formally, if we series expand the exponential in \eqref{eq:zmugh} and swap the order of the expectation and summation, we find that $Z_{\mu}(g,h)$ would be given by 
\begin{align*}
\sum_{n=0}^\infty \frac{\mu^n}{n!}\gffcf{\partial \varphi(g)\partial\varphi(h)\O\left(\Re(\:e^{i\sqrt{\beta}\varphi}\:(\psi))\right)^n}.
\end{align*}
This is almost of the form claimed in the statement of the lemma. To massage it into the required form, note that for example from Lemma \ref{le:wickconv} one can show that for real valued $\psi$, $\overline{\:e^{i\sqrt{\beta}\varphi}\:(\psi)}=\:e^{-i\sqrt{\beta}\varphi}\:(\psi)$, so 
\begin{equation*}
\Re\left(\:e^{i\sqrt{\beta}\varphi}\:(\psi)\right)=\frac{1}{2}\left(\:e^{i\sqrt{\beta}\varphi}\:(\psi)+\:e^{-i\sqrt{\beta}\varphi}\:(\psi)\right)=\frac{1}{2}\sum_{\sigma\in \{-1,1\}}\:e^{i\sigma\sqrt{\beta}\varphi}\:(\psi).
\end{equation*}
Thus 
\begin{align*}
\left(\Re\left(\:e^{i\sqrt{\beta}\varphi}\:(\psi)\right)\right)^n =\frac{1}{2^n}\sum_{\sigma_1,...,\sigma_n\in \{-1,1\}}\prod_{j=1}^n \:e^{i\sigma_j \sqrt{\beta}\varphi}\:(\psi),
\end{align*}
which suggests that $Z_{\mu}(g,h)$ should be given by the series expansion \eqref{eq:zmuseries} as claimed.

The only step in this argument that was not properly justified was interchanging the order of expectation and the series expansion. For this we naturally want to use Fubini. It is thus sufficient to show that  
\begin{align*}
\sum_{n=0}^\infty \frac{|\mu|^n}{2^n n!}\gffcf{\left|\partial \varphi(g)\partial \varphi(h)\prod_{j=1}^n \:e^{i\sigma_j\sqrt{\beta}\varphi}\:(\psi)\right|}<\infty.
\end{align*}
Using Cauchy-Schwarz (and the fact that $\gffcf{|\partial \varphi(g)\partial \varphi(f)|^2}<\infty$ since we are dealing with Gaussian random variables) and our remark that $\overline{\:e^{i\sqrt{\beta}\varphi}\:(\psi)}=\:e^{-i\sqrt{\beta}\varphi}\:(\psi)$, we see that this is equivalent to
\begin{align*}
\sum_{n=0}^\infty \frac{|\mu|^n}{2^n n!}\left(\gffcf{\left|\:e^{i\sqrt{\beta}\varphi}\:(\psi)\right|^{2n}}\right)^{1/2}<\infty.
\end{align*}
Due to our uniform integrability and $L^p$ convergence (see Remark \ref{re:vitali}), we can use Proposition \ref{pr:mombd1} to bound  
\begin{equation*}
\left(\gffcf{\left|\:e^{i\sqrt{\beta}\varphi}\:(\psi)\right|^{2n}}\right)^{1/2}\leq C_{\psi,\beta}n^{\frac{\beta}{8\pi}n}
\end{equation*}
so the question is now convergence of the series 
\begin{align*}
\sum_{n=0}^\infty \frac{|\mu|^n}{2^n n!}n^{\frac{\beta}{8\pi}n}.
\end{align*}
By Stirling's approximation, this would be convergent for every $\mu\in \R$ even for $\beta<8\pi$, so it certainly converges for $\beta<4\pi$. Thus our formal argument is justified and we have derived the series expansion \eqref{eq:zmuseries}.

The proofs of \eqref{eq:zmuseries2} and \eqref{eq:zmuseries3} are analogous and we omit the details.
\end{proof}

\subsection{Pointwise correlation functions for the sine-Gordon model}

The main goal of this section is the proof of the following proposition. 
\begin{proposition}\label{pr:sgpw}
There exist  functions 
\begin{equation}\label{eq:func1}
(x,y)\mapsto \gffcf{\partial \varphi(x)\partial \varphi(y)\O e^{\mu \Re(\:e^{i\sqrt{\beta}\varphi}\:(\psi))}}
\end{equation}
and 
\begin{equation}\label{eq:func2}
x\mapsto \gffcf{\:e^{\pm i \sqrt{\beta}\varphi(x)}\:\O e^{\mu \Re(\:e^{i\sqrt{\beta}\varphi}\:(\psi))}}
\end{equation}
such that 
\begin{equation*}
Z_\mu(g,h)=\int_{\Omega^2}d^2x d^2y g(x)h(y)\gffcf{\partial \varphi(x)\partial \varphi(y)\O e^{\mu \Re(\:e^{i\sqrt{\beta}\varphi}\:(\psi))}}
\end{equation*}
and 
\begin{equation*}
Z_\mu(\eta)=\int_{\Omega}d^2x \eta(x)\gffcf{\:e^{\pm i \sqrt{\beta}\varphi(x)}\:\O e^{\mu \Re(\:e^{i\sqrt{\beta}\varphi}\:(\psi))}}
\end{equation*}
for $g,h,\eta\in C_c^\infty(\Omega)$ and $g,h$ having disjoint supports.

The function \eqref{eq:func1} is continuous on $\{(x,y)\in \Omega^2: x\neq y\}$ and the function \eqref{eq:func2} is continuous on $\Omega$.

Again, analogous claims are true for correlation functions involving $\bar\partial$.
\end{proposition}

We split the proof into parts, but the basic idea is to use the series expansion in $\mu$ from Lemma \ref{le:sgsmeared}. Our first step is to analyze the coefficients of the series expansion more precisely. We consider first the derivative case.
\begin{lemma}\label{le:coef1}
For $g,h\in  C_c^\infty(\Omega)$ with disjoint supports, $\O=e^{i\varphi(f)}$ with $f\in C_c^\infty(\Omega)$, $\psi\in C_c^\infty(\Omega)$, $\sigma_1,...,\sigma_n\in \{-1,1\}$, and $\beta\in(0,4\pi)$,  we have 
\begin{equation*}
\gffcf{\partial \varphi(g)\partial\varphi(h)\O\prod_{j=1}^n \:e^{i\sigma_j\sqrt{\beta}\varphi}\:(\psi)}=\int_{\Omega^2} d^2x d^2 y g(x)h(y)\gffcf{\partial\varphi(x)\partial \varphi(y)\O\prod_{j=1}^n \:e^{i\sigma_j\sqrt{\beta}\varphi}\:(\psi)},
\end{equation*}
where 
\begin{align*}
&\gffcf{\partial\varphi(x)\partial \varphi(y)\O\prod_{j=1}^n \:e^{i\sigma_j\sqrt{\beta}\varphi}\:(\psi)}\\
&\quad=\gffcf{\O}\int_{\Omega^n} d^{2n}u\prod_{j=1}^n(\psi_j(u_j)) e^{-\beta \sum_{1\leq k<l\leq n}\sigma_k\sigma_lG_\Omega(u_k,u_l)}  \left(\partial_x\partial_y G_\Omega(x,y)-V(x;u)V(y;u)\right),
\end{align*}
with
\begin{equation*}
\psi_j(u)=\psi(u)e^{-\frac{\beta}{2} g_\Omega(u,u)} e^{-\sqrt{\beta}\sigma_j\int_\Omega d^2v f(v)G_\Omega(v,u)}
\end{equation*}
and
\begin{equation*}
V(x;u)=\int_\Omega d^2v f(v)\partial_x G_\Omega(v,x)+\sqrt{\beta}\sum_{j=1}^n \sigma_j \partial_x G_\Omega(x,u_j).
\end{equation*}
All of these integrals are absolutely convergent, and similar claims are true if one or both of the holomorphic derivatives $\partial$ are replaced by an antiholomorphic derivative $\bar\partial$.

\end{lemma}
\begin{proof}
Our starting point is to write 
\begin{align*}
\gffcf{\partial \varphi(g)\partial\varphi(h)\O\prod_{j=1}^n \:e^{i\sigma_j\sqrt{\beta}\varphi}\:(\psi)}&=\lim_{\epsilon\to 0}\gffcf{\partial \varphi(g)\partial\varphi(h)\O\prod_{j=1}^n \:e^{i\sigma_j\sqrt{\beta}\varphi_\epsilon}\:(\psi)}\\
&=\lim_{\epsilon\to 0}\int d^{2n}u \prod_{j=1}^n (\psi(u_j))\gffcf{\partial \varphi(g)\partial \varphi(h)\O \prod_{k=1}^n \:e^{i\sigma_k \sqrt{\beta}\varphi_\epsilon(u_k)}\:}
\end{align*}
which follows from Lemma \ref{le:wickconv}, Cauchy-Schwarz, Proposition \ref{pr:mombd1}, Vitali's convergence theorem (see Remark \ref{re:vitali}), and Fubini. We can now make use of\footnote{To be precise, we define say $V_1=\varphi(f)+\sum_{j=1}^n \sigma_j \sqrt{\beta}\varphi_\epsilon(u_j)$, $V_2=\Re(\partial \varphi(g))$, $V_3=\Im(\partial \varphi(g))$, $V_4=\Re(\partial \varphi(h))$, $V_5=\Im(\partial \varphi(h))$, $P(V_1,...,V_5)=(V_2+iV_3)(V_4+iV_5)$, and $z=i$ in Lemma \ref{le:girsa}.} Lemma \ref{le:girsa} and find
\begin{align*}
&\gffcf{\partial \varphi(g)\partial \varphi(h)\O \prod_{j=1}^n \:e^{i\sigma_j \sqrt{\beta}\varphi_\epsilon(u_j)}\:}\\
&\quad=\gffcf{\O \prod_{j=1}^n \:e^{i\sigma_j \sqrt{\beta}\varphi_\epsilon(u_j)}\:} \gffcf{\left(\partial \varphi(g)+iV_\epsilon(g,u)\right)\left(\partial \varphi(h)+iV_{\epsilon}(h,u)\right)}\\
&\quad =\gffcf{\partial \varphi(g)\partial \varphi(h)}\gffcf{\O \prod_{j=1}^n \:e^{i\sigma_j \sqrt{\beta}\varphi_\epsilon(u_j)}\:}-V_\epsilon(g,u)V_\epsilon(h,u)\gffcf{\O \prod_{j=1}^n \:e^{i\sigma_j \sqrt{\beta}\varphi_\epsilon(u_j)}\:}
\end{align*}
where 
\begin{equation*}
V_\epsilon(g,u)=\gffcf{\partial \varphi(g)\varphi(f)}+\sqrt{\beta}\sum_{j=1}^n \sigma_j\gffcf{\partial \varphi(g)\varphi_\epsilon(u_j)}. 
\end{equation*}
Next we note that since we are dealing with Gaussian random variables (so that the moment generating function is explicit in terms of the covariances), \eqref{eq:gffcov2} implies that for small enough $\epsilon$,
\begin{align}\label{eq:for51}
\gffcf{\O \prod_{j=1}^n \:e^{i\sigma_j \sqrt{\beta}\varphi_\epsilon(u_j)}\:}&= \gffcf{\O} \prod_{j=1}^n\left(e^{-\sqrt{\beta}\sigma_j \gffcf{\varphi(f)\varphi_\epsilon(u_j)}} e^{-\frac{\beta}{2}g_\Omega(u_j,u_j)}\right)\\
&\quad \times e^{-\beta \sum_{1\leq k< l\leq n}\sigma_k\sigma_l \gffcf{\varphi_\epsilon(u_k)\varphi_\epsilon(u_l)}}.\notag
\end{align}

Lemma \ref{le:varphiecov} allows us to bound quantities like $\gffcf{\varphi(f)\varphi_\epsilon(u)}$ (and hence $\gffcf{\partial \varphi(g)\varphi_\epsilon(u)}$) uniformly in $\epsilon$ (and locally uniformly in $u$), so using Proposition \ref{pr:mombd1}, we see that we can use the dominated convergence theorem to take the $\epsilon\to 0$ limit under the integral and find 
\begin{align}\label{eq:intermed}
&\gffcf{\partial \varphi(g)\partial\varphi(h)\O\prod_{j=1}^n \:e^{i\sigma_j\sqrt{\beta}\varphi}\:(\psi)}\\
\notag&=\gffcf{\O}\int_{\Omega^{n}}d^{2n}u \prod_{j=1}^n (\psi_j(u_j))\left(\gffcf{\partial \varphi(g)\partial \varphi(h)}-V(g;u)V(h;u)\right) e^{-\beta\sum_{1\leq k<l\leq n}\sigma_k\sigma_l G_\Omega(u_k,u_l)},
\end{align}
where  
\begin{align*}
\psi_j(u)=\psi(u)e^{-\sqrt{\beta}\sigma_j\int_\Omega d^2 vf(v)G_\Omega(v,u)}e^{-\frac{\beta}{2}g_\Omega(u,u)}
\end{align*}
(as in the statement of the lemma) and 
\begin{align*}
V(g;u)&=\gffcf{\partial \varphi(g)\varphi(f)}+\sqrt{\beta}\sum_{j=1}^n \sigma_j \int_\Omega d^2 v g(v)\partial_v G_\Omega(u_j,v)\\
&=\int_{\Omega^2}d^2 x d^2v g(x)f(v)\partial_x G_\Omega(x,v)+\sqrt{\beta}\sum_{j=1}^n \sigma_j \int_\Omega d^2 x g(x)\partial_x G_\Omega(x,u_j),
\end{align*}
where we used the fact (following from \eqref{eq:gomega}) that first order derivatives of $G_\Omega$ are locally integrable.

Since $g$ and $h$ have disjoint supports, we can write in \eqref{eq:intermed} 
\begin{align}\label{eq:int1}
\gffcf{\partial \varphi(g)\partial \varphi(h)}=\int_{\Omega^2}d^2x d^2 y \partial g(x)\partial h(y)G_\Omega(x,y)=\int_{\Omega^2}d^2x d^2 y  g(x) h(y)\partial_x \partial_yG_\Omega(x,y).
\end{align}
For the $V(g;u)V(h;u)$-term, we would like to use Fubini to write 
\begin{align}\label{eq:int2}
&\int_{\Omega^{n}}d^{2n}u \prod_{j=1}^n (\psi_j(u_j))V(g;u)V(h;u) e^{-\beta\sum_{1\leq k<l\leq n}\sigma_k\sigma_l G_\Omega(u_k,u_l)}\\
\notag &\stackrel{\text{to show}}{=}\int_{\Omega^2}d^2x d^2 y  g(x)h(y) \int_{\Omega^n}d^{2n}u \left(\int_\Omega d^2 v f(v)\partial_x G_\Omega(x,v)+\sqrt{\beta}\sum_{p=1}^n\sigma_p \partial_x G_\Omega(x,u_p)\right)\\
\notag &\qquad \qquad \times  \left(\int_\Omega d^2 v' f(v')\partial_y G_\Omega(y,v')+\sqrt{\beta}\sum_{q=1}^n \sigma_q \partial_y G_\Omega(y,u_q)\right) \\
\notag &\qquad \qquad \times \prod_{j=1}^n (\psi_j(u_j)) e^{-\beta\sum_{1\leq k<l\leq n}\sigma_k\sigma_l G_\Omega(u_k,u_l)}.
\end{align}
If this were justified, combining with \eqref{eq:int1} and rearranging slightly would yield the claim. The issue is to check that Fubini here is really justified. As we saw in the proof of Lemma \ref{le:gffpw}, the $f$-terms
\begin{equation*}
x\mapsto \int_{\Omega}d^2 v\, f(v)\partial_x G_\Omega(x,v)
\end{equation*}
are actually smooth functions, so the actual question is whether functions of the form 
\begin{align}
\label{eq:term1}&\partial_x G_\Omega(x,u_p)e^{-\beta \sum_{1\leq k<l\leq n}\sigma_k \sigma_lG_\Omega(u_k,u_l)},\\
\label{eq:term2}&\partial_x G_\Omega(x,u_p)\partial_y G_\Omega(y,u_p)e^{-\beta \sum_{1\leq k<l\leq n}\sigma_k \sigma_lG_\Omega(u_k,u_l)},\\
\label{eq:term3}&\partial_x G_\Omega(x,u_p)\partial_y G_\Omega(y,u_q)e^{-\beta \sum_{1\leq k<l\leq n}\sigma_k \sigma_lG_\Omega(u_k,u_l)},
\end{align}
(note that the last two terms differ in the $\partial_y G_\Omega$-term) with $p\neq q$ are locally integrable in $u$ when $|x-y|$ is bounded from below (since $g$ and $h$ have disjoint supports). For this, we make use of Proposition \ref{pr:mombd2} and Proposition \ref{pr:mombd3}.

In the notation of Proposition \ref{pr:mombd2}, we note that the $u$-integral of \eqref{eq:term1} over $K^n$ for some compact $K\subset \Omega$ can be bounded by  (a constant times)
\begin{equation*}
\int_K d^2u \frac{1}{|x-u|}|G_n(u)|,
\end{equation*}
which is convergent and bounded in $x$ due to boundedness of $G_n$. Similarly the corresponding integral for \eqref{eq:term2} can be bounded by 
\begin{equation*}
\int_K d^2 u \frac{1}{|x-u||y-u|}|G_n(u)|,
\end{equation*}
which again is convergent and bounded in $x,y$ (with $|x-y|$ bounded from below) by boundedness of $G_n$. Finally for \eqref{eq:term3}, we have an upper bound of the form 
\begin{equation*}
\int_{K^2}d^2 u d^2 v\frac{1}{|x-u|}\frac{1}{|y-v|}|H_n(u,v)|,
\end{equation*}
where $H_n$ is as in Proposition \ref{pr:mombd3}. Using Proposition \ref{pr:mombd3} and Young's convolution inequality, we can bound this by an $n$-dependent constant times
\begin{equation*}
\left(\int_{B(0,R)}d^2 u \frac{1}{|u-x|^p}\right)^{1/p}\left(\int_{B(0,R)}d^2 u \frac{1}{|u|^{\frac{\beta}{2\pi}q}}\right)^{1/q}\left(\int_{B(0,R)}d^2 u \frac{1}{|u-y|^r}\right)^{1/r}
\end{equation*}
for $p,q,r\geq 1$ satisfying $\frac{1}{p}+\frac{1}{q}+\frac{1}{r}=2$, and $R$ large enough, but fixed. We choose $q>1$ small enough such that $q\frac{\beta}{2\pi}<2$ (and the middle integral remains finite). We also choose $p=r$ so that the remaining integrals are convergent and bounded in $x,y$ (since $p,r<2$ in this case). This justifies the use of Fubini in \eqref{eq:int2}, and concludes the proof.
\end{proof}

We next record a similar claim for the correlation function involving a Wick exponential.
\begin{lemma}\label{le:coef2}
For $\eta\in C_c^\infty(\Omega)$, $\O=e^{i\varphi(f)}$ with $f\in C_c^\infty(\Omega)$, $\psi\in C_c^\infty(\Omega)$, $\sigma_1,...,\sigma_n\in \{-1,1\}$, and $\beta\in (0,4\pi)$, we have 
\begin{align*}
\gffcf{\:e^{\pm i \sqrt{\beta}\varphi}\:(\eta)\O\prod_{j=1}^n \:e^{i\sigma_j \sqrt{\beta}\varphi}\:(\psi)}=\int_{\Omega}d^2x \eta(x)\gffcf{\:e^{\pm i \sqrt{\beta}\varphi(x)}\: \O \prod_{j=1}^n \:e^{i\sigma_j \sqrt{\beta}\varphi}\:(\psi)},
\end{align*}
where 
\begin{align}
\notag &\gffcf{\:e^{\pm i \sqrt{\beta}\varphi(x)}\: \O \prod_{j=1}^n \:e^{i\sigma_j \sqrt{\beta}\varphi}\:(\psi)}\\
&\quad =\gffcf{\O}e^{-\frac{\beta}{2}g_\Omega(x,x)}e^{\mp \sqrt{\beta}\int_{\Omega}d^2 v f(v)G_\Omega(v,x)}\notag\\
&\qquad \times \int_{\Omega^n}d^{2n}u \prod_{j=1}^n (\psi_j(u_j))e^{-\beta\sum_{1\leq k<l\leq n}\sigma_k\sigma_l G_\Omega(u_k,u_l)}  e^{\mp \beta\sum_{p=1}^n \sigma_p G_\Omega(x,u_p)},\label{eq:1ptexprep}
\end{align}
and
\begin{equation*}
\psi_j(u)=\psi(u)e^{-\frac{\beta}{2}g_\Omega(u,u)}e^{-\sqrt{\beta}\sigma_j\int_\Omega d^2 v f(v)G_\Omega(v,u)}. 
\end{equation*}
All the integrals are absolutely convergent.
\end{lemma}
\begin{proof}
The proof is similar to the previous one. By Lemma \ref{le:wickconv}, Proposition \ref{pr:mombd1} (and Vitali's convergence theorem -- see Remark \ref{re:vitali}), and a routine Gaussian calculation (namely expressing the moment generating function in terms of covariances), we can write 
\begin{align*}
&\gffcf{\:e^{\pm i \sqrt{\beta}\varphi}\:(\eta)\O\prod_{j=1}^n \:e^{i\sigma_j \sqrt{\beta}\varphi}\:(\psi)}\\
&=\lim_{\epsilon\to 0}\gffcf{\:e^{\pm i \sqrt{\beta}\varphi_\epsilon}\:(\eta)\O\prod_{j=1}^n \:e^{i\sigma_j \sqrt{\beta}\varphi_\epsilon}\:(\psi)}\\
&=\gffcf{\O}\lim_{\epsilon\to 0}\int_{\Omega} d^2x \eta(x) e^{\mp \sqrt{\beta}\gffcf{\varphi_\epsilon(x)\varphi(f)}} e^{-\frac{\beta}{2}g_\Omega(x,x)}\\
&\qquad \times \int_{\Omega^n}d^{2n}u \prod_{j=1}^n \left(\psi(u_j) e^{-\frac{\beta}{2}g_\Omega(u_j,u_j)} e^{-\sigma_j \sqrt{\beta}\gffcf{\varphi_\epsilon(u_j)\varphi(f)}}e^{\mp \sigma_j \beta \gffcf{\varphi_\epsilon(x)\varphi_\epsilon(u_j)}}\right)\\
&\qquad \qquad \times e^{-\beta \sum_{1\leq k<l\leq n}\sigma_k \sigma_l \gffcf{\varphi_\epsilon(u_k)\varphi_\epsilon(u_l)}}.
\end{align*}
As in the previous proof, we can use Lemma \ref{le:varphiecov} and Proposition \ref{pr:mombd2} to justify the use of the dominated convergence theorem, take $\epsilon\to 0$ inside of the integrals, and to recover the claim of the lemma.
\end{proof}

To be able to use these pointwise correlation functions to prove Proposition \ref{pr:sgpw}, we need continuity estimates for them. We begin with the derivatives again. 
\begin{lemma}\label{le:cont1}
The function
\begin{equation*}
(x,y)\mapsto \gffcf{\partial \varphi(x)\partial\varphi(y)\O\prod_{j=1}^n \:e^{i\sigma_j\sqrt{\beta}\varphi}\:(\psi)}
\end{equation*}
is continuous on the set $\{(x,y)\in \Omega^2: x\neq y\}$, and for any $K\subset \{(x,y)\in \Omega^2: x\neq y\}$ compact, we have the estimate
\begin{equation}\label{eq:2ptbound}
\sup_{(x,y)\in K}\left|\gffcf{\partial \varphi(x)\partial\varphi(y)\O\prod_{j=1}^n \:e^{i\sigma_j\sqrt{\beta}\varphi}\:(\psi)}\right|\leq C^n n^{\frac{\beta}{8\pi}n}
\end{equation}
for some constant $C$ depending only on $K,\psi,f$ and $\beta$.

Analogous claims are true if one or both of the holomorphic derivatives $\partial$ are replaced by an antiholomorphic one $\bar\partial$.
\end{lemma}
\begin{proof}
We begin with the definition of the function: according to Lemma \ref{le:coef1}:
\begin{align*}
&\gffcf{\partial\varphi(x)\partial\varphi(y)\O\prod_{j=1}^n \:e^{i\sigma_j\sqrt{\beta}\varphi}\:(\psi)}\\
&\quad =\partial_x \partial_y G_\Omega(x,y)\gffcf{\O}\int_{\Omega^n}d^{2n}u\prod_{j=1}^n (\psi_j(u_j) )e^{-\beta \sum_{1\leq k<l\leq n}\sigma_k\sigma_lG_\Omega(u_k,u_l)}\\
&\qquad -\gffcf{\O}\int_{\Omega^n}d^{2n}u\prod_{j=1}^n (\psi_j(u_j))e^{-\beta \sum_{1\leq k<l\leq n}\sigma_k\sigma_lG_\Omega(u_k,u_l)}V(x;u)V(y;u).
\end{align*}
We see that $\partial_x\partial_y G_\Omega(x,y)$ is continuous on any $K\subset \{(x,y)\in \Omega^2: x\neq y\}$ while the constant multiplying it satisfies the bound in the statement according to Proposition \ref{pr:mombd1} (note that $\psi_j\in C_c^\infty(\Omega)$). We can thus focus on the $V$-term.

Let us recall that 
\begin{align*}
V(x;u)=\int_\Omega d^2v f(v)\partial_x G_\Omega(v,x)+\sqrt{\beta}\sum_{j=1}^n \sigma_j \partial_x G_\Omega(x,u_j).
\end{align*}
Since $x\mapsto \int_{\Omega}d^2 vf(v)\partial_x G_\Omega(x,v)$ is smooth (see the proof of Lemma \ref{le:gffpw}), it is sufficient for us to obtain the required control for terms of the following type
\begin{align}
\label{eq:boundterm1} I_1(x)&=\int_{\Omega^n}d^{2n}u\prod_{j=1}^n (\psi_j(u_j))e^{-\beta \sum_{1\leq k<l\leq n}\sigma_k\sigma_lG_\Omega(u_k,u_l)} \partial_x G_\Omega(x,u_1),\\
\label{eq:boundterm2}I_2(x,y) &= \int_{\Omega^n}d^{2n}u\prod_{j=1}^n (\psi_j(u_j))e^{-\beta \sum_{1\leq k<l\leq n}\sigma_k\sigma_lG_\Omega(u_k,u_l)} \partial_x G_\Omega(x,u_1)\partial_y G_\Omega(y,u_1),\\
\label{eq:boundterm3} I_3(x,y)&= \int_{\Omega^n}d^{2n}u\prod_{j=1}^n (\psi_j(u_j))e^{-\beta \sum_{1\leq k<l\leq n}\sigma_k\sigma_lG_\Omega(u_k,u_l)} \partial_x G_\Omega(x,u_1)\partial_y G_\Omega(y,u_2).
\end{align}
We start by looking at $I_1$. Using \eqref{eq:gomega} and the notation of Proposition \ref{pr:mombd2}, we can write 
\begin{equation*}
I_1(x)=\int_{\Omega}d^2u \psi_1(u)\left[\frac{1}{4\pi}\frac{1}{u-x}+\partial_x g_\Omega(x,u)\right]G_n(u).
\end{equation*}
Continuity and bounds for the $g_\Omega$-term follow readily from the smoothness of $g_\Omega$ and Proposition \ref{pr:mombd2}. For the $(u-x)^{-1}$, term, to estimate continuity, we write 
\begin{align*}
\left|\int_\Omega d^2 u \psi_1(u)\frac{G_n(u)}{u-x}-\int_\Omega d^2 u \psi_1(u)\frac{G_n(u)}{u-x'}\right|&\leq |x-x'| \int_\Omega d^2 u |\psi_1(u)|\frac{|G_n(u)|}{|u-x||u-x'|}\\
&\leq |x-x'| \|\psi_1 G_n\|_{L^\infty(\Omega)}\int_{B(0,R)}d^2 u \frac{1}{|u-x||u-x'|},
\end{align*}
where $R$ is large enough. One can readily\footnote{By translating and scaling, the question is equivalent to bounding $\int_{B(0,1)}d^2 u \frac{1}{|u||u-(x-x')|}$. By rotating $u$ by the phase of $x-x'$ and scaling by $|x-x'|$, we see that this integral is $\int_{|u|\leq |x-x'|^{-1}}d^2 u \frac{1}{|u||u-1|}$. For small $|x-x'|$, we can split the integration region into $B(0,10)\cup \{u: 10<|u|\leq |x-x'|^{-1}\}$. In the first region, we simply get a constant, and in the latter domain, we can bound $\frac{1}{|u||u-1|}\leq \frac{C}{|u|^2}$ for a universal $C$, so the question becomes understanding asymptotics of $\int_{10<|u|\leq |x-x'|^{-1}}\frac{d^2u}{|u|^2}$. Going into polar coordinates, we see that this integral can be bounded by a constant times $1+|\log |x-x'||$.} bound this integral by a constant times $1+|\log |x-x'||$, so we see that $I_1$ is at least Hölder continuous. For the quantitative bound, we note that 
\begin{align*}
\left|\int_\Omega d^2 u \psi_1(u) \frac{G_n(u)}{|u-x|}\right|\leq \|\psi_1 G_n\|_{L^\infty(\Omega)} \int_{\Omega} \frac{d^2 u}{|u-x|},
\end{align*}
for which Proposition \ref{pr:mombd2} provides an estimate consistent with \eqref{eq:2ptbound}. Thus we conclude that $I_1$ satisfies the bound of \eqref{eq:2ptbound}. 

For $I_2$, similar considerations show that the relevant quantity to bound and control is 
\begin{align*}
\int_{\Omega}d^2 u\psi_1(u)\frac{1}{(u-x)(u-y)}G_n(u).
\end{align*}
Using that $|x-y|$ is bounded from below, one can use similar arguments to prove that this is Hölder continuous on any compact subset of $\{(x,y)\in \Omega^2: x\neq y\}$, and for the $L^\infty$-bound we have 
\begin{align*}
\left|\int_{\Omega}d^2 u\psi_1(u)\frac{1}{(u-x)(u-y)}G_n(u)\right|\leq \|\psi_1G_n\|_{L^\infty(\Omega)}\int_\Omega \frac{d^2 u}{|u-x||u-y|}.
\end{align*}
which satisfies the bound of \eqref{eq:2ptbound} by Proposition \ref{pr:mombd2}.

For $I_3$, we see that in the notation of Proposition \ref{pr:mombd3}, the relevant quantity to estimate is 
\begin{equation*}
L(x,y):=\int_{\Omega^2}d^2 u_1 d^2 u_2 \psi_1(u_1)\psi_2(u_2)\frac{1}{u_1-x}\frac{1}{u_2-y}H_n(u_1,u_2).
\end{equation*}
For continuity, we note for example that Proposition \ref{pr:mombd3} implies that 
\begin{align*}
|L(x,y)-L(x',y)|&\leq |x-x'| C_n \int_{\Omega^2}d^2 u_1 d^2 u_2\frac{1}{|u_1-x||u_1-x'|}\frac{1}{|u_2-y|}\frac{1}{|u_1-u_2|^{\frac{\beta}{2\pi}}}
\end{align*}
for a suitable constant $C_n$. This integral can be estimated with Young's convolution inequality, and we have for a suitable $R>0$
\begin{align*}
|L(x,y)-L(x',y)|&\leq |x-x'|C_n \left(\int_{B(0,R)} d^2u \frac{1}{|u-x|^p|u-x'|^p}\right)^{1/p} \left(\int_{B(0,R)}d^2 u \frac{1}{|u|^{q\frac{\beta}{2\pi}}}\right)^{1/q}\\
&\quad \times \left(\int_{B(0,R)}d^2 u \frac{1}{|u-y|^{r}}\right)^{1/r}
\end{align*} 
with $\frac{1}{p}+\frac{1}{q}+\frac{1}{r}=2$. Choosing $q>1$ small enough to satisfy $q\frac{\beta}{2\pi}<2$ and $r=p$, we see that the second integral is finite and last one is bounded in $y$ (since $r<2$ as well). For the first integral, we see that small $|x-x'|$ behavior of the integral is the same as the small $|x|$ behavior of the integral 
\begin{equation*}
\int_{B(0,1)}d^2u \frac{1}{|u-x|^p|u|^p}=|x|^{2-2p}\int_{B(0,|x|^{-1})}d^2u \frac{1}{|u-1|^p|u|^p}.
\end{equation*}
We see that 
\begin{align*}
|L(x,y)-L(x',y)|\leq \widetilde C_n |x-x'|^{1+\frac{2}{p}-2}.
\end{align*}
A similar argument provides the corresponding estimate when varying the second variable, and we see that since $p<2$, $L$ is Hölder continuous.

A similar argument relying on Young's convolution inequality and Proposition \ref{pr:mombd3} shows that $L$ also satisfies the required $L^\infty$-bound. Thus we have established that all terms are continuous and enjoy the $L^\infty$-bound in the statement of the lemma. 
\end{proof}

We have a similar claim for the exponential case. 
\begin{lemma}\label{eq:cont2}
The function 
\begin{equation*}
x\mapsto \gffcf{\:e^{\pm i \sqrt{\beta}\varphi(x)}\: \O \prod_{j=1}^n \:e^{i\sigma_j \sqrt{\beta}\varphi}\:(\psi)}
\end{equation*}
is continuous on $\Omega$ and for any $K\subset \Omega$ compact, we have the estimate
\begin{align*}
\sup_{x\in K}\left|\gffcf{\:e^{\pm i \sqrt{\beta}\varphi(x)}\: \O \prod_{j=1}^n \:e^{i\sigma_j \sqrt{\beta}\varphi}\:(\psi)}
\right|\leq C^n n^{\frac{\beta}{8\pi}n}
\end{align*}
for some constant $C$ depending only on $\beta,K,f,\psi$.
\end{lemma}
\begin{proof}
This follows by applying Proposition \ref{pr:mombd2} to \eqref{eq:1ptexprep} (with $x$ playing the role of $x_1$ there). Note that in this case, we apply Proposition \ref{pr:mombd2} to $G_{n+1}$, so we get a bound of the form $C^{n+1}(n+1)^{\frac{\beta}{8\pi}(n+1)}$, but this can be bounded by $\widetilde C^n n^{\frac{\beta}{8\pi}n}$ for a suitable $\widetilde C$.
\end{proof}

We can now turn to the proof of Proposition \ref{pr:sgpw}.
\begin{proof}[Proof of Proposition \ref{pr:sgpw}]
By Lemma \ref{le:sgsmeared} and Lemma \ref{le:coef1}, we can write 
\begin{align*}
Z_\mu(g,h)&=\sum_{n=0}^\infty  \frac{\mu^n}{2^n n!}\sum_{\sigma_1,...,\sigma_n\in \{-1,1\}}\gffcf{\partial \varphi(g)\partial \varphi(h)\O\prod_{j=1}^n \:e^{i\sigma_j \sqrt{\beta}\varphi}\:(\psi)}\\
&=\sum_{n=0}^\infty \int_{\Omega^2}d^2x d^2 y g(x)h(y)\frac{\mu^n}{2^n n!}\sum_{\sigma_1,...,\sigma_n\in \{-1,1\}}\gffcf{\partial \varphi(x)\partial \varphi(y)\O\prod_{j=1}^n \:e^{i\sigma_j \sqrt{\beta}\varphi}\:(\psi)}.
\end{align*}
By Fubini and Lemma \ref{le:cont1}, we can interchange the order of the sum over $n$ and the integral over $x$ and $y$ (since $g,h$ have disjoint supports). This means that we have 
\begin{align*}
Z_\mu(g,h)&= \int_{\Omega^2}d^2x d^2 y g(x)h(y)\sum_{n=0}^\infty\frac{\mu^n}{2^n n!}\sum_{\sigma_1,...,\sigma_n\in \{-1,1\}}\gffcf{\partial \varphi(x)\partial \varphi(y)\O\prod_{j=1}^n \:e^{i\sigma_j \sqrt{\beta}\varphi}\:(\psi)}.
\end{align*}
The kernel on the right hand side is our candidate for the function 
\begin{equation*}
\gffcf{\partial \varphi(x)\partial \varphi(y)\O e^{\mu \Re(\:e^{i\sqrt{\beta}\varphi}\:(\psi))}}.
\end{equation*}
It remains to prove the continuity of this on the required domain. Each summand is continuous by Lemma \ref{le:cont1}, and for any compact $K\subset \{(x,y)\in \Omega^2: x\neq y\}$, the series converges absolutely in $L^\infty(K)$. Thus by completeness of $(C(K),\|\cdot\|_\infty)$, our series defines a continuous function on each compact $K\subset \{(x,y)\in \Omega^2: x\neq y\}$ and thus defines a continuous function on $\{(x,y)\in \Omega^2: x\neq y\}$. The proofs for the exponential and other derivative cases are analogous and we omit the details. 
\end{proof}

This concludes our construction of the sine-Gordon correlation functions featuring in Theorem \ref{th:main} so we can turn to the proof of the OPEs.

\section{The OPE for derivative fields for the sine-Gordon model -- Proof of Theorem \ref{th:main}}\label{sec:sGope}

In this section, we prove Theorem \ref{th:main}. This will follow immediately from the following two results. 
\begin{proposition}\label{pr:ope1}
For the functions defined in Proposition \ref{pr:sgpw}, we have that the following function
\begin{align*}
\gffcf{\partial\varphi(x)\partial\varphi(y)\O e^{\mu \Re(\:e^{i\sqrt{\beta}\varphi}\:(\psi))}}&-\gffcf{\partial \varphi(x)\partial \varphi(y)}\gffcf{\O e^{\mu \Re(\:e^{i\sqrt{\beta}\varphi}\:(\psi))}}\\
&\quad -\mu \frac{\beta}{16\pi} \frac{|x-y|^2}{(x-y)^2}\psi(y)\gffcf{\:\cos (\sqrt{\beta}\varphi(y))\:\O e^{\mu \Re(\:e^{i\sqrt{\beta}\varphi}\:(\psi))}}
\end{align*}
has a limit as $x\to y$.
\end{proposition}
\begin{proof}
We start with the representation (see the proof of Proposition \ref{pr:sgpw})
\begin{align*}
\gffcf{\partial\varphi(x)\partial\varphi(y)\O e^{\mu \Re(\:e^{i\sqrt{\beta}\varphi}\:(\psi))}}&=\sum_{n=0}^\infty \frac{\mu^n}{2^n n!}\sum_{\sigma_1,...,\sigma_n\in \{-1,1\}}\gffcf{\partial \varphi(x)\partial \varphi(y)\O \prod_{j=1}^n \:e^{i\sigma_j \sqrt{\beta}\varphi}\:(\psi)}.
\end{align*}
From Lemma \ref{le:coef1} and Lemma \ref{le:cont1}, we can write 
\begin{align}\label{eq:thiseq}
\gffcf{\partial \varphi(x)\partial \varphi(y)\O \prod_{j=1}^n \:e^{i\sigma_j \sqrt{\beta}\varphi}\:(\psi)}&=\gffcf{\O}\int_{\Omega^n}d^{2n}u\prod_{j=1}^n (\psi_j(u_j))e^{-\beta \sum_{1\leq k<l\leq n}\sigma_k\sigma_l G_\Omega(u_k,u_l)}\\
&\qquad \times (\partial_x\partial_y G_\Omega(x,y)-V(x;u)V(y;u)),\notag 
\end{align}
where 
\begin{equation*}
\psi_j(u)=\psi(u)e^{-\frac{\beta}{2}g_\Omega(u,u)}e^{-\sqrt{\beta}\sigma_j \int_\Omega d^2 v f(v)G_\Omega(v,u)}
\end{equation*}
and 
\begin{align*}
V(x;u)=\int_\Omega d^2 vf(v)\partial_x G_\Omega(v,x)+\sqrt{\beta}\sum_{j=1}^n \sigma_j \partial_x G_\Omega(x,u_j).
\end{align*}
We claim that the first term in \eqref{eq:thiseq} satisfies
\begin{align*}
\gffcf{\O}\int_{\Omega^n}d^{2n} u\, \prod_{j=1}^n(\psi_j(u_j)) e^{-\beta\sum_{1\leq l<k\leq n}\sigma_l \sigma_k G_\Omega(u_l,u_k)}=\gffcf{\O \prod_{j=1}^n \:e^{i\sqrt{\beta}\sigma_j\varphi}\:(\psi)}.
\end{align*}
This can be verified by starting from the right hand side and writing it as $\lim_{\epsilon\to 0}\gffcf{\O \prod_{j=1}^n \:e^{i\sqrt{\beta}\sigma_j\varphi_\epsilon}\:(\psi)}$, and arguing as in the proof of Lemma \ref{le:coef1} (see in particular \eqref{eq:for51}).

Recalling that $\gffcf{\partial\varphi(x)\partial \varphi(y)}=\partial_x\partial_y G_\Omega(x,y)$ (see e.g. the proof of Lemma \ref{le:gffpw} with $f=0$), we thus see from Lemma \ref{le:sgsmeared} 
\begin{align*}
&\gffcf{\partial \varphi(x)\partial\varphi(y)\O e^{\mu \Re(\:e^{i\sqrt{\beta}\varphi}\:(\psi))}}-\gffcf{\partial\varphi(x)\partial\varphi(y)}\gffcf{\O e^{\mu \Re(\:e^{i\sqrt{\beta}\varphi}\:(\psi))}}\\
&=-\gffcf{\O}\sum_{n=0}^\infty \frac{\mu^n}{2^n n!}\sum_{\sigma_1,...,\sigma_n\in \{-1,1\}}\int_{\Omega^n}d^{2n}u\, \prod_{j=1}^n (\psi_j(u_j))e^{-\beta\sum_{1\leq k<l\leq n}\sigma_k\sigma_l G_\Omega(u_k,u_l)}V(x;u)V(y;u).
\end{align*}
Our task is now to understand the $x\to y$ asymptotics of this sum (up to a regular term)

For this, we recall the definition of $V$ from Lemma \ref{le:coef1} and write 
\begin{align*}
&\int_{\Omega^n}d^{2n}u\, \prod_{j=1}^n (\psi_j(u_j))e^{-\beta\sum_{1\leq k<l\leq n}\sigma_k\sigma_l G_\Omega(u_k,u_l)}V(x;u)V(y;u)\\
&=\int_{\Omega}d^2 v f(v)\partial_x G_\Omega(v,x)\int_{\Omega}d^2 v' f(v')\partial_y G_\Omega(v',y)\int_{\Omega^n}d^{2n}u\, \prod_{j=1}^n (\psi_j(u_j))e^{-\beta\sum_{1\leq k<l\leq n}\sigma_k\sigma_l G_\Omega(u_k,u_l)} \\
&\quad +\int_{\Omega}d^2 v f(v)\partial_x G_\Omega(v,x) \sqrt{\beta}\sum_{j=1}^n \sigma_j \int_{\Omega^n}d^{2n}u\, \prod_{j=1}^n (\psi_j(u_j))e^{-\beta\sum_{1\leq k<l\leq n}\sigma_k\sigma_l G_\Omega(u_k,u_l)}  \partial_y G_\Omega(y,u_j)\\
&\quad +\int_{\Omega}d^2 v f(v)\partial_y G_\Omega(v,y) \sqrt{\beta}\sum_{j=1}^n \sigma_j \int_{\Omega^n}d^{2n}u\, \prod_{j=1}^n (\psi_j(u_j))e^{-\beta\sum_{1\leq k<l\leq n}\sigma_k\sigma_l G_\Omega(u_k,u_l)}  \partial_x G_\Omega(x,u_j)\\
&\quad +\beta \sum_{p,q=1}^n \sigma_p\sigma_q\int_{\Omega^n}d^{2n}u\, \prod_{j=1}^n (\psi_j(u_j))e^{-\beta\sum_{1\leq k<l\leq n}\sigma_k\sigma_l G_\Omega(u_k,u_l)} \partial_x G_\Omega(x,u_p) \partial_y G_\Omega(y,u_q)\\
&=:\sum_{j=1}^4I_{n,j}(x,y;\sigma).
\end{align*}

Let us now look at the sums of these various terms.

\medskip

\underline{The sum of the $I_{n,1}$ terms:} Here we note that the prefactor 
\begin{equation*}
\int_{\Omega}d^2v\, f(v)\partial_x G_\Omega(v,x)\int_{\Omega}d^2 v' \, f(v')\partial_y G_\Omega(v',y)
\end{equation*}
 does not depend on $n$ or $\sigma$, so we can perform the $\sigma$ and $n$ sums (as we did for the $\partial_x \partial_yG_\Omega(x,y)$ term) to find that 
\begin{align*}
&\gffcf{\O}\sum_{n=0}^\infty \frac{\mu^n}{2^n n!}\sum_{\sigma_1,...,\sigma_n \in \{-1,1\}}I_{n,1}(x,y;\sigma)\\
&=\int_{\Omega}d^2v\, f(v)\partial_x G_\Omega(v,x)\int_{\Omega}d^2 v' \, f(v')\partial_y G_\Omega(v',y) \gffcf{\O e^{\mu \Re(\:e^{i\sqrt{\beta}\varphi}\:(\psi))}}.
\end{align*}
We know from the proof of Lemma \ref{le:gffpw} that $\int_{\Omega}d^2v\, f(v)\partial_x G_\Omega(v,x)$ is a smooth function of $x$, so we see that 
\begin{align*}
\lim_{x\to y}\gffcf{\O}\sum_{n=0}^\infty \frac{\mu^n}{2^n n!}\sum_{\sigma_1,...,\sigma_n \in \{-1,1\}}I_{n,1}(x,y;\sigma)
\end{align*}
exists. 

\medskip

\underline{The sum of the $I_{n,2}$ terms:} Here we note that again the prefactor $\int_{\Omega}d^2 v\, f(v)\partial_x G_\Omega(x,v)$ is $n$-independent and a smooth function, but to see that the full sum is smooth in $x$, we need to show that it is convergent. For this, we note that the integral 
\begin{equation*}
\int_{\Omega^n}d^{2n}u \, \prod_{j=1}^n (\psi_j(u_j))e^{-\beta\sum_{1\leq k<l\leq n}\sigma_k\sigma_l G_\Omega(u_k,u_l)}\partial_y G_\Omega(y,u_j),
\end{equation*}
is precisely of the form of $I_1(y)$ (with a suitable permutation of the indices) from \eqref{eq:boundterm1}. With the same argument as in the proof of Lemma \ref{le:cont1}, we see that using Proposition \ref{pr:mombd2}, for any compact $K\subset \Omega$ there exists a $C_K>0$ such that
\begin{align*}
\sup_{y\in K}\left|\int_{\Omega^n}d^{2n}u \, \prod_{j=1}^n (\psi_j(u_j))e^{-\beta\sum_{1\leq k<l\leq n}\sigma_k\sigma_l G_\Omega(u_k,u_l)}\partial_y G_\Omega(y,u_j)\right|\leq C_K^n n^{\frac{\beta}{8\pi}n}.
\end{align*}
We conclude (once again by Stirling) that 
\begin{align*}
\gffcf{\O}\sum_{n=0}^\infty \frac{\mu^n}{2^n n!}\sum_{\sigma_1,...,\sigma_n \in \{-1,1\}}I_{n,2}(x,y;\sigma)
\end{align*}
converges and defines a smooth function of $x$ (in particular, it has a limit as $x\to y$).

\medskip

\underline{The sum of the $I_{n,3}$ terms:} This is similar to the $I_{n,2}$-case, but now the $x$-dependence does not factor out. However, we still have a sum involving terms of the form $I_1$ from Lemma \ref{le:cont1}. In the proof of Lemma \ref{le:cont1}, we saw these to be continuous for each $n$ and due to the $L^\infty$–bound we already mentioned (and the completeness of the space $C(K,\|\cdot\|_{L^\infty(K)})$), we see that also 
\begin{align*}
\gffcf{\O}\sum_{n=0}^\infty \frac{\mu^n}{2^n n!}\sum_{\sigma_1,...,\sigma_n \in \{-1,1\}}I_{n,3}(x,y;\sigma)
\end{align*}
is a continuous function (and thus has a limit as $x\to y$).

\medskip

\underline{The sum of the $I_{n,4}$ terms:}  Here we split $I_{n,4}$ into two parts:
\begin{align*}
I_{n,4}^1(x,y;\sigma)=\beta \sum_{p\neq q}\sigma_p \sigma_q \int_{\Omega^n}d^{2n}u\, \prod_{j=1}^n (\psi_j(u_j))e^{-\beta \sum_{1\leq k<l\leq n}\sigma_k \sigma_l G_\Omega(u_k,u_l)}\partial_x G_\Omega(x,u_p)\partial_y G_\Omega(y,u_q)
\end{align*}
and 
\begin{align*}
I_{n,4}^2(x,y;\sigma)=\beta \sum_{p=1}^n\int_{\Omega^n}d^{2n}u\, \prod_{j=1}^n (\psi_j(u_j))e^{-\beta \sum_{1\leq k<l\leq n}\sigma_k \sigma_l G_\Omega(u_k,u_l)}\partial_x G_\Omega(x,u_p)\partial_y G_\Omega(y,u_p).
\end{align*}
We note that in the notation of Lemma \ref{le:cont1}, $I_{n,4}^1$ is a sum of terms of the form $I_3$ while $I_{n,4}^2$ is a sum of terms of the form $I_2$. From the proof of Lemma \ref{le:cont1}, we see that $I_3$ (and hence $I_{n,4}^1$) is continuous on all of $\Omega\times \Omega$. From the same proof, we also deduce that for some constant $C_K$, we have 
\begin{align*}
\sup_{x,y\in K}|I_{n,4}^1(x,y;\sigma)|\leq C_K^n n^{\frac{\beta}{8\pi}n}.
\end{align*}
Thus
\begin{align*}
\sum_{n=0}^\infty \frac{\mu^n}{2^n n!}\sum_{\sigma_1,...,\sigma_n\in \{-1,1\}}I_{n,4}^1(x,y;\sigma)
\end{align*}
is a continuous function of $x,y$ on $K\times K$. 

Note that so far, we have argued that all of the contributions coming from $I_{n,1},I_{n,2},I_{n,3}$ and $I_{n,4}^1$ produce a continuous function in $x$ so their contribution affect only the regular part of the correlation function. So if there is a singular part, it must come from $I_{n,4}^2$. We will now show that there is indeed one. In the notation of the proof of Lemma \ref{le:cont1}, we see that $I_{n,4}^2$ is a sum of terms of the form $I_2$. By (the proof of) Lemma \ref{le:cont1}, we thus know that this is a continuous function \emph{off of the diagonal} $x=y$, and arguments such as Fubini are justified. To extract the diagonal behavior, we argue similarly to the proof of Lemma \ref{le:cont1}. First of all, using \eqref{eq:gomega}, we can write 
\begin{align*}
I_{n,4}^2(x,y;\sigma)&=\beta \sum_{p=1}^n\int_{\Omega^n}d^{2n}u\, \prod_{j=1}^n (\psi_j(u_j))e^{-\beta \sum_{1\leq k<l\leq n}\sigma_k \sigma_l G_\Omega(u_k,u_l)}\frac{1}{16\pi^2}\frac{1}{(x-u_p)(y-u_p)}\\
&\quad -\beta \sum_{p=1}^n\int_{\Omega^n}d^{2n}u\, \prod_{j=1}^n (\psi_j(u_j))e^{-\beta \sum_{1\leq k<l\leq n}\sigma_k \sigma_l G_\Omega(u_k,u_l)}\frac{1}{4\pi}\frac{1}{x-u_p}\partial_y g_\Omega(y,u_p)\\
&\quad -\beta \sum_{p=1}^n\int_{\Omega^n}d^{2n}u\, \prod_{j=1}^n (\psi_j(u_j))e^{-\beta \sum_{1\leq k<l\leq n}\sigma_k \sigma_l G_\Omega(u_k,u_l)}\frac{1}{4\pi}\frac{1}{y-u_p}\partial_x g_\Omega(x,u_p)\\
&\quad +\beta \sum_{p=1}^n\int_{\Omega^n}d^{2n}u\, \prod_{j=1}^n (\psi_j(u_j))e^{-\beta \sum_{1\leq k<l\leq n}\sigma_k \sigma_l G_\Omega(u_k,u_l)}\partial_x g_\Omega(x,u_p)\partial_y g_\Omega(y,u_p).
\end{align*}
The last three terms can be bounded similarly to the $I_{n,2}$ and $I_{n,3}$ terms, so we only need to focus on the first term. For this, we have by Fubini
\begin{align*}
&\beta \sum_{p=1}^n\int_{\Omega^n}d^{2n}u\, \prod_{j=1}^n (\psi_j(u_j))e^{-\beta \sum_{1\leq k<l\leq n}\sigma_k \sigma_l G_\Omega(u_k,u_l)}\frac{1}{16\pi^2}\frac{1}{(x-u_p)(y-u_p)}\\
&=\beta\sum_{p=1}^n \int_{\Omega}d^2 u_p \psi_p(u_p) \frac{1}{16\pi^2}\frac{1}{(x-u_p)(y-u_p)}G_n^p(u_p),
\end{align*}
where 
\begin{equation*}
G_n^p(u_p)=\int_{\Omega^{n-1}}d^2 u_1\cdots d^2 u_{p-1} d^2 u_{p+1}\cdots d^2 u_n \prod_{j\neq p}(\psi_j(u_j))e^{-\beta \sum_{1\leq k<l\leq n}\sigma_k \sigma_l G_\Omega(u_k,u_l)}.
\end{equation*}
Next we write 
\begin{align*}
&\beta\sum_{p=1}^n \int_{\Omega}d^2 u_p \psi_p(u_p) \frac{1}{16\pi^2}\frac{1}{(x-u_p)(y-u_p)}G_n^p(u_p)\\
&=\beta\sum_{p=1}^n \psi_p(y)G_n^p(y) \int_{\Omega}  d^2 u_p \frac{1}{16\pi^2}\frac{1}{(x-u_p)(y-u_p)}\\
&\quad+\beta\sum_{p=1}^n   \int_\Omega d^2 u_p \frac{1}{16\pi^2}\frac{\psi_p(u_p)G_n^p(u_p)-\psi_p(y)G_n^p(y)}{(x-u_p)(y-u_p)}.
\end{align*}
Arguing\footnote{To be more precise, let $f$ be compactly supported and Lipschitz and consider the function $h(x)= \int d^2 u\, \frac{f(u)-f(y)}{(x-u)(y-u)}$. We then have $|h(x)-h(x')|\leq \|f\|_{\mathrm{Lip}}|x-x'|\int_{\mathrm{supp}(f)}d^2 u\, \frac{1}{|x-u||x'-u|}\leq C_f |x-x'|(1+|\log |x-x'||)$ for some constant $C_f$ depending only on the support of $f$ and $\|f\|_{\mathrm{Lip}}$. Thus $h$ is (Hölder) continuous.} as in the proof of Lemma \ref{le:cont1}, one readily finds from the Lipschitz bound of Proposition \ref{pr:mombd2} that the last term is a continuous function of $x$ on $\Omega$, and that its sup-norm over any compact $K\subset \Omega$ satisfies our familiar bound $C_K^n n^{\frac{\beta}{8\pi}n}$ implying that the sum over $n$ and $\sigma$ produces a continuous function -- once again, the contribution of this term can be absorbed into the regular part.

Before completing our proof, we summarize what we have proven so far. We have argued that 
 \begin{align*}
&\gffcf{\partial \varphi(x)\partial\varphi(y)\O e^{\mu \Re(\:e^{i\sqrt{\beta}\varphi}\:(\psi))}}-\gffcf{\partial\varphi(x)\partial\varphi(y)}\gffcf{\O e^{\mu \Re(\:e^{i\sqrt{\beta}\varphi}\:(\psi))}}\\
&\qquad +\beta \int_{\Omega}d^2 u\frac{1}{16\pi^2}\frac{1}{(x-u)(y-u)} \gffcf{\O}\sum_{n=0}^\infty \frac{\mu^n}{2^n n!}\sum_{\sigma_1,...,\sigma_n\in \{-1,1\}}\sum_{p=1}^n \psi_p(y)G_n^p(y)
\end{align*}
is a continuous function of $x$ (in particular, a limit exists as $x\to y$).  It thus remains to identify the sum over $n$ as a correlation function, and to compute the small $|x-y|$ asymptotics of the integral. Let us first look at the integral.

\medskip

\underline{Asymptotics of the integral:} To study the $u$-integral, let $\delta>0$ be small enough and fixed and satisfy $x\in B(y,\delta)\subset \Omega$ so that we have 
\begin{align*}
\int_\Omega d^2 u \frac{1}{(u-x)(u-y)}&=\int_{B(y,\delta)}d^2 u \frac{1}{(u-x)(u-y)}+\int_{\Omega\setminus B(y,\delta)}d^2 u \frac{1}{(u-x)(u-y)}.
\end{align*}
The second term, and its contribution to the full correlation function, has a limit as $x\to y$, so again, any possible singularity comes from the first one. For this, we note that 
\begin{align*}
\int_{B(y,\delta)}d^2 u \frac{1}{(u-x)(u-y)}&=\int_{B(0,\delta)}d^2 u \frac{1}{u-(x-y)}\frac{1}{u}\\
&=-\int_{|u|<|x-y|}d^2 u \frac{1}{x-y} \sum_{m=0}^\infty \frac{u^m}{(x-y)^m}\frac{1}{u}\\
&\quad +\int_{|x-y|<|u|<\delta}d^2 u \sum_{m=0}^\infty \frac{(x-y)^m}{u^m}\frac{1}{u^2}\\
&=-\pi\frac{|x-y|^2}{(x-y)^2},
\end{align*}
where we noted that by Fubini and going into polar coordinates, only the $m=1$ term of the first sum integrates to something non-zero.

We conclude that 
\begin{align*}
\int_\Omega d^2 u \frac{1}{(u-x)(u-y)}=-\pi \frac{|x-y|^2}{(x-y)^2}+\text{continuous}.
\end{align*}

\medskip

\underline{Evaluating the sum:} To evaluate the sum, we recall our definition of $\psi_p$ and $G_n^p$, and write 
\begin{align*}
&\gffcf{\O}\sum_{n=0}^\infty \frac{\mu^n}{2^n n!}\sum_{\sigma_1,...,\sigma_n\in \{-1,1\}}\sum_{p=1}^n \psi_p(y)G_n^p(y)\\
&=\gffcf{\O}\sum_{n=0}^\infty \frac{\mu^n}{2^n n!}\sum_{p=1}^n\sum_{\sigma_1,...,\sigma_n\in\{-1,1\}} \int_{\Omega^{n-1}}d^2 u_1\cdots d^2 u_{p-1}d^2 u_{p+1}\cdots d^2 u_n \prod_{j=1}^n(\psi_j(u_j)) \\
&\qquad \qquad \qquad \qquad \qquad \qquad \qquad \qquad \qquad \times e^{-\beta \sum_{1\leq k<l\leq n}\sigma_k \sigma_l G_\Omega(u_k,u_l)}, 
\end{align*}
where in the integral, we understand that $u_p=y$. From Lemma \ref{le:coef2} and (the proof of) Proposition \ref{pr:sgpw}, we recognize this to be 
\begin{align*}
&\sum_{n=0}^\infty \frac{\mu^n}{2^n n!}\sum_{p=1}^n \sum_{\sigma_1,...,\sigma_n\in \{-1,1\}}\psi(y)\gffcf{\:e^{i\sigma_p \sqrt{\beta}\varphi(y)}\: \O \prod_{j\neq p}\:e^{i\sigma_j\sqrt{\beta}\varphi}\:(\psi)}\\
&=\sum_{n=1}^\infty \frac{\mu^n}{(n-1)!} \psi(y)\gffcf{\:\cos(\sqrt{\beta}\varphi(y))\: \O \left(\Re(\:e^{i\sqrt{\beta}\varphi}\:(\psi))\right)^{n-1}}\\
&=\mu \psi(y) \gffcf{\:\cos(\sqrt{\beta}\varphi(y))\: \O e^{\mu \Re(\:e^{i\sqrt{\beta}\varphi}\:(\psi))}}.
\end{align*}

Putting everything together, we conclude that 
 \begin{align*}
&\gffcf{\partial \varphi(x)\partial\varphi(y)\O e^{\mu \Re(\:e^{i\sqrt{\beta}\varphi}\:(\psi))}}-\gffcf{\partial\varphi(x)\partial\varphi(y)}\gffcf{\O e^{\mu \Re(\:e^{i\sqrt{\beta}\varphi}\:(\psi))}}\\
&\qquad -\frac{\beta}{16\pi} \frac{|x-y|^2}{(x-y)^2}\mu \psi(y)\gffcf{\:\cos(\sqrt{\beta}\varphi(y))\: \O e^{\mu \Re(\:e^{i\sqrt{\beta}\varphi}\:(\psi))}}
\end{align*}
has a limit as $x\to y$, which concludes the proof.
\end{proof}
The claim in Theorem \ref{th:main} about $\partial \varphi(x)\partial \varphi(y)$, namely \eqref{eq:OPE1} follows simply by dividing the result of Proposition \ref{pr:ope1} by $\gffcf{e^{\mu \Re(\:e^{i\sqrt{\beta}\varphi}\:(\psi))}}$. The claim for $\bar\partial\varphi(x)\bar\partial\varphi(y)$, namely \eqref{eq:OPE2}, follows by complex conjugating this result (and replacing $f$ by $-f$).

It thus remains to prove the claim for the correlation functions involving $\partial \varphi(x)\bar\partial\varphi(y)$, namely \eqref{eq:OPE3}. For this, we have the following result.
\begin{proposition}
For the functions defined in Proposition \ref{pr:sgpw}, we have that the following function
\begin{align*}
\gffcf{\partial\varphi(x)\bar \partial\varphi(y)\O e^{\mu \Re(\:e^{i\sqrt{\beta}\varphi}\:(\psi))}}& +\mu \frac{\beta}{8\pi} \log |x-y|^{-1} \psi(y)\gffcf{\:\cos (\sqrt{\beta}\varphi(y))\:\O e^{\mu \Re(\:e^{i\sqrt{\beta}\varphi}\:(\psi))}}
\end{align*}
has a limit as $x\to y$.
\end{proposition}

\begin{proof}
The proof is essentially the same as that of Proposition \ref{pr:ope1}. The first difference is that $\gffcf{\partial\varphi(x)\partial\varphi(y)}$ gets replaced by $\gffcf{\partial\varphi(x)\bar\partial\varphi(y)}$. This is continuous as $x\to y$ (since it is just $\partial_x \bar\partial_y g_\Omega(x,y)$), so it can be absorbed into the part that has a limit as $x\to y$.

The analysis of the terms corresponding to $I_{n,1},I_{n,2},I_{n,3}, I_{n,4}^1$ is completely analogous -- we only used size bounds for showing that these produce continuous functions. For $I_{n,4}^2$, one can again (with completely analogous arguments) reduce the question to understanding the contributing asymptotics of the integral
\begin{align*}
\int_{B(0,\delta)}d^2 u \frac{1}{u-(x-y)}\frac{1}{\bar u}.
\end{align*}
Using the same series expansions, we see that now the only non-zero contribution comes from 
\begin{equation*}
\int_{|x-y|<|u|<\delta}d^2 u \frac{1}{|u|^2}=2\pi (\log |x-y|^{-1}-\log \delta^{-1}).
\end{equation*}
After this, the proof continues exactly as in Proposition \ref{pr:ope1}, and we recover the claim.
\end{proof}
Again, by dividing by the normalization constant, we recover \eqref{eq:OPE3}. This concludes the proof of Theorem \ref{th:main}.

\appendix 

\section{A tool for computing Gaussian expectations}\label{app:gauss}

We review in this appendix a simple version of Girsanov's theorem for a finite dimensional Gaussian distribution which we use repeatedly.
\begin{lemma}\label{le:girsa}
Let $V_1,...,V_n$ be centered, real-valued, and jointly Gaussian random variables with covariance matrix $\Sigma$. Moreover, let $P:\C^n \to \C$ be a polynomial. We then have for any $z\in \C$
\begin{equation}\label{eq:girsa}
\E\left[P(V_1,...,V_n)e^{zV_1}\right]=\E\left[e^{zV_1}\right]\E\left[P\left(V_1+z\E(V_1^2),...,V_n+z\E(V_1V_n)\right)\right]
\end{equation} 
\end{lemma}
\begin{proof}
First of all, since $\E[e^{zV_1}]=e^{\frac{z^2}{2}\E(V_1^2)}$, one readily checks that the right hand side of \eqref{eq:girsa} is an entire function of $z$. Moreover, since we are dealing with Gaussian random variables, we have $P(V_1,...,V_n)\in L^2(\P)$ and $e^{R|V_1|}\in L^2(\P)$ for each $R>0$. Thus using Fubini and Morera's theorem, one sees that also the left hand side of \eqref{eq:girsa} is an entire function of $z$. It is thus sufficient for us to verify that \eqref{eq:girsa} holds for $z\in \R$.

The verification of this is a standard ``complete the square'' argument we now present. As we are dealing with centered jointly Gaussian random variables, we have 
\begin{align*}
\E\left[P(V_1,...,V_n)e^{zV_1}\right]=\frac{1}{\sqrt{(2\pi)^n \det \Sigma}}\int_{\R^n} d^n x P(x_1,...,x_n) e^{zx_1} e^{-\frac{1}{2}x^\mathsf T\Sigma^{-1}x}.
\end{align*}
Making the change of variables $x_i=y_i+z \Sigma_{1i}$, we have 
\begin{align*}
zx_1-\frac{1}{2}x^\mathsf T\Sigma^{-1}x&=zy_1+z^2\Sigma_{11}-\frac{1}{2}y^\mathsf T\Sigma^{-1}y-z\sum_{i,j=1}^n \Sigma_{1i}(\Sigma^{-1})_{ij}y_j-\frac{1}{2}z^2 \sum_{i,j=1}^n \Sigma_{1i}\Sigma_{1j}(\Sigma^{-1})_{ij}\\
&=\frac{z^2}{2}\Sigma_{11}-\frac{1}{2}y^\mathsf T\Sigma^{-1}y.
\end{align*}
Thus 
\begin{align*}
\E\left[P(V_1,...,V_n)e^{zV_1}\right]&=e^{\frac{z^2}{2}\Sigma_{11}}\frac{1}{\sqrt{(2\pi)^n \det \Sigma}}\int_{\R^n} d^n y P(y_1+z\Sigma_{11},...,y_n+z\Sigma_{1n})e^{-\frac{1}{2}y^\mathsf T \Sigma^{-1}y}\\
&=e^{\frac{z^2}{2}\E(V_1^2)}\E\left[P\left(V_1+z\E(V_1^2),...,V_n+\E(V_1V_n)\right)\right]
\end{align*}
which is the claim (for real valued $z$). This concludes the proof. 
\end{proof}

\section{Proof of some facts from Section \ref{sec:tools}}\label{app:tools}

In this appendix we provide proofs of Lemma \ref{le:nnnum} and Lemma \ref{le:intest} which we used in Section \ref{app:tools}. These are not new, and similar arguments can be found in \cite{GP} and \cite[Appendix A]{JSW}.

Before proving the bound in Lemma \ref{le:nnnum}, we prove an exact formula for the quantity $|\mathcal F_{2,k}^{(n)}|$.

\begin{proposition}\label{thm:LDGoTLRTs}
For any $n \in \mathbb{N}$ and $k \leq n/2$, the number of directed two-loop rooted forests $F \in \mathcal{F}^{(n)}_{2, k}$ with $n$ vertices and $k$ two-loops is
\begin{align*}
|\mathcal{F}^{(n)}_{2, k}|
=
\frac{n!}{2^k k!}
\frac{2 k n^{n-2 k -1} }{(n-2k)!}.
\end{align*}
\end{proposition}

\begin{proof}
Proofs can be found, e.g., in \cite{GP} or \cite[Appendix A]{JSW}. We explicate an elementary combinatorial argument and refer to \cite[Appendix A]{JSW} for the rest.

Let $C_m$, $m \in \mathbb{N}$, denote the number of rooted trees with $m$ labelled vertices. As in Section \ref{sec:tools}, we direct the trees toward the root. For instance, $C_1 = 1$, $C_2 = 2$ (the trees $ 1 \rightarrow 2 $ and $ 2 \rightarrow 1$) and $C_3 = 9$ ($6$ re-labellings of the tree $1 \rightarrow 2 \rightarrow 3$ and $3$ of the tree $1 \rightarrow 2 \leftarrow 3$). A directed two-loop rooted forest $F \in \mathcal{F}^{(n)}_{2, k}$ with $n$ vertices and $k$ two-loops ($1 \leq k \leq n/2$) can be seen as a collection of $2k$ such labelled trees, paired by the two-loops. Denoting the vertex count of the $i$:th tree by $m_i$ (so $1 \leq i \leq 2k$ and $m_1 + \ldots + m_{2 k} = n$) the number of ways to distribute the labels $1, \ldots, n$ to the different trees is given by the multinomial coefficient
\begin{align*}
\binom{n}{m_1, \ldots, m_{2 k}}=\frac{n!}{m_1! \ldots m_{2k} !}.
\end{align*}
With the labels fixed, each tree has $C_{m_i}$ possible structures.  Imposing for a moment an additional order to these trees (there are $k! 2^k$ such orders, by the freedom to permute and reverse the two-loops), it is easy to compute $\# \mathcal{F}^{(n)}_{2, k}$ in terms of the auxiliary quantity
\begin{align*}
T_m := C_m / m!.
\end{align*}

Altogether, this reasoning gives
\begin{align*}
|\mathcal{F}^{(n)}_{2, k}|
& =
\frac{1}{2^k k!} \sum_{ \substack{m \in \mathbb{Z}_{\geq 1}^{2 k} \\ m_1 + \ldots + m_{2 k} = n }} 
\binom{n}{m_1, \ldots, m_{2 k}} C_{m_1} \cdots C_{m_{2 k}} \\
& = \frac{n!}{2^k k!} \sum_{ \substack{m \in \mathbb{Z}_{\geq 1}^{2 k} \\ m_1 + \ldots + m_{2 k} = n }}  T_{m_1} \cdots T_{m_{2 k}}.
\end{align*}
Furthermore, in terms of the generating function $T(x) = \sum_{m=1}^\infty T_m x^m$ (a formal power series), one readily observes that the sum above is the coefficient of $x^n$ in the power series of $(T(x))^{2k}$, denoted $[T (x)^{2 k}]_n$:
\begin{align*}
|\mathcal{F}^{(n)}_{2, k}| =  \frac{n!}{2^k k!} [T (x)^{2 k}]_n.
\end{align*}
The problem of finding $|\mathcal{F}^{(n)}_{2, k}|$ is thus transferred into finding $[T (x)^{2 k}]_n$.

\begin{lemma}
The generating function $T(x) = \sum_{m \geq 1} T_m x^m$ satisfies $T(x) = x e^{T(x)}$. It is the inverse function of $\phi(x)=xe^{-x}$ as a formal power series, and as a convergent complex power series in a neighborhood of $0$.
\end{lemma}

\begin{proof}
The first claim follows from a recursive formula for $C_m$, derived similarly as the formula for $| \mathcal{F}^{(n)}_{2, k}|$: a labelled rooted tree of $m$ nodes can be constructed by gluing $d$ smaller labelled rooted trees by directed edges to the root vertex, where $1 \leq d \leq (m-1)$ is the (in-)degree of the root vertex. Arguing as above, one obtains
\begin{align*}
C_m
& =
\sum_{d=1}^{m-1} \frac{1}{d!}
\sum_{ \substack{r \in \mathbb{Z}_{\geq 1}^{d} \\ r_1 + \ldots + r_{d} = m-1}} 
\binom{m}{1, r_1, \ldots, r_{d}}  C_{r_1} \cdots C_{r_{d}} \\
\Rightarrow \quad 
T_m & = \sum_{d=1}^{m-1} \frac{1}{d!}
\sum_{ \substack{r \in \mathbb{Z}_{\geq 1}^{d} \\ r_1 + \ldots + r_{d} = m-1}}   T_{r_1} \cdots T_{r_{d}}.
\end{align*}
and this identity gives
\begin{align*}
x e^{T(x)} &= \sum_{d=0}^\infty x \frac{T(x)^d}{d!} 
= \sum_{d=0}^\infty x \frac{1}{d!} \sum_{ r \in \mathbb{Z}_{\geq 1}^{d} } T_{r_1} \cdots T_{r_{d}} x^{r_1} \cdots x^{r_{d}} 
= \sum_{m=1}^\infty T_m x^m = T(x).
\end{align*}
For the second claim, just compute $\phi(T(x))=T(x)e^{-T(x)}=x e^{T(x)}e^{-T(x)}=x$. Finally, to go to usual complex power series, note that $\phi(z)$ has a local, complex analytic inverse $g: U \to \mathbb{C}$ in some neighborhood $U$ of $0$. The equation $\phi(g(z))=z$ (resp. $\phi(T(x))=x$) determines recursively the power series coefficients of $g$ (resp. $T$). Hence, the power series coefficients of $g$ and $T$ coincide.
\end{proof}

Finally, using the fact that $ T^{-1}(z) = z e^{-z}$ in a neighbourhood of $0$, $[T (x)^{2 k}]_n$ can be found using the Lagrange--B\"urmann theorem \cite[Chapter VII.3]{WW}, see \cite[Appendix A]{JSW}.
\end{proof}

Proving Lemma \ref{le:nnnum} from this explicit formula is now simple.

\begin{proof}[Proof of Lemma \ref{le:nnnum}]
To obtain the form of the bound stated in Lemma \ref{le:nnnum}, note that
\begin{align*}
|\mathcal{F}^{(n)}_{2, k}|
&= \frac{n!}{2^k k!}\frac{2 k n^{n-2 k-1} }{(n-2 k)!} \\
& \leq
2^{n-k} (n-k)! \frac{ n^{n-2 k} }{(n-2k)!} 
\qquad \text{(since $2 k \leq n$ and $\tfrac{n!}{(n - k)! k!} \leq 2^n$)} \\
& \leq C
2^{n-k} (n-k)! \frac{ n^{n-2 k} }{(\tfrac{n-2k}{e})^{n-2 k}} 
\qquad \text{(Stirling)}.
\end{align*}
Writing $n-2 k = s \in [0, n]$ we have $\log \tfrac{ n^{n-2 k} }{(n-2 k)^{n-2 k}} = - s \log (s/n) \leq C n $. Hence, we obtain
\begin{align*}
|\mathcal{F}^{(n)}_{2, k}|\leq  C^n (n-k)! .
\end{align*}
\end{proof}

\begin{proof}[Proof of Lemma \ref{le:intest}]
Our first remark is that scaling the integration variables by $R$, we have for a suitable constant $C=C_{\alpha,R}$
\begin{align*}
\int_{U(R)}d^{2(l+m)}u \prod_{j=1}^l |u_j|^{-\frac{\alpha}{2\pi}}\prod_{k=l+1}^{l+m}|u_k|^{-\frac{\alpha}{4\pi}}\leq C^{l+m}\int_{U(1)}d^{2(l+m)}u \prod_{j=1}^l |u_j|^{-\frac{\alpha}{2\pi}}\prod_{k=l+1}^{l+m}|u_k|^{-\frac{\alpha}{4\pi}},
\end{align*}
so we focus on this integral. Going into polar coordinates in each $u_j$, and then defining $t_j=\frac{1}{4}r_j^2$, there exists again a constant $C=C_\alpha$ such that 
\begin{align*}
\int_{U(1)}d^{2(l+m)}u \prod_{j=1}^l |u_j|^{-\frac{\alpha}{2\pi}}\prod_{k=l+1}^{l+m}|u_k|^{-\frac{\alpha}{4\pi}}\leq C^{l+m}\int_{\mathcal U}d^{l+m}t \prod_{j=1}^l t_j^{-\frac{\alpha}{4\pi}}\prod_{k=l+1}^{l+m}t_k^{-\frac{\alpha}{8\pi}},
\end{align*}
where 
\begin{equation*}
\mathcal U=\left\{(t_1,...,t_{l+m})\in (0,1)^{l+m}: \sum_{j=1}^{l+m}t_j\leq 1\right\}.
\end{equation*}
This last integral can be evaluated (note that it is closely related to the Dirichlet distribution): from \cite[Section 12.5]{WW}, if $f:[0,1]\to \R$ is continuous and $\alpha_1,...,\alpha_N>0$, then
\begin{equation*}
\int_{\mathcal U}d^{N}t f(t_1+...+t_N) \prod_{j=1}^N t_j^{\alpha_j-1}=\frac{\Gamma(\alpha_1)\cdots \Gamma(\alpha_N)}{\Gamma(\alpha_1+\cdots+\alpha_N)}\int_0^1 f(t) t^{\sum_{j=1}^N \alpha_j-1}dt.
\end{equation*}
For us $f=1$, and we find the bound
\begin{align*}
\int_{U(R)}d^{2(l+m)}u \prod_{j=1}^l |u_j|^{-\frac{\alpha}{2\pi}}\prod_{k=l+1}^{l+m}|u_k|^{-\frac{\alpha}{4\pi}}\leq C^{l+m} \frac{\Gamma(1-\frac{\alpha}{4\pi})^l\Gamma(1-\frac{\alpha}{8\pi})^m}{\Gamma(1+l(1-\frac{\alpha}{4\pi})+m(1-\frac{\alpha}{8\pi}))},
\end{align*}
which is of the form we want.
\end{proof}

\end{document}